%% file: Causal_Inf_Groupwise_Matching_A46.tex
\documentclass[12pt, fullpage]{amsart}

\usepackage{amssymb}
\usepackage{amsmath}
\usepackage{amsfonts}
\usepackage{mathrsfs}
\usepackage{mathtools}
\usepackage{graphicx} 
\usepackage{color}
\usepackage{booktabs}
\usepackage{xcolor}
\usepackage{float}
\usepackage{algorithm}
\usepackage{algorithmic}
\usepackage[onehalfspacing]{setspace}
\usepackage{natbib}
\usepackage{enumerate}
\usepackage[utf8]{inputenc}
\usepackage[charter,cal=cmcal]{mathdesign}
\usepackage[colorlinks=true,citecolor=blue,urlcolor=blue,pdfpagemode=UseNone,pdfstartview=FitH]{hyperref}
\usepackage{tikz}
\usepackage{pgfplots}
\usepackage{caption}
\usepackage{subcaption}
\usepackage{makecell}
\usepackage{multirow}
\usepackage{chngcntr}

\newcommand{\CI}{\mathrel{\perp\mspace{-10mu}\perp}}

\makeatletter
\def\section{\@startsection{section}{1}
	\z@{.6\linespacing\@plus\linespacing}{.6\linespacing}{\Large}}

\def\subsection{\@startsection{subsection}{2}
	\z@{.4\linespacing\@plus.7\linespacing}{.4\linespacing}{\large}}

\def\subsubsection{\@startsection{subsubsection}{3}
	\z@{.4\linespacing\@plus.4\linespacing}{-.5em}{\normalfont\bfseries}}
    
\makeatother

\makeatletter
\renewenvironment{proof}[1][\proofname]{\par
	\normalfont \topsep6\p@\@plus6\p@\relax
	\trivlist
	\item[\hskip\labelsep
		\textbf{#1\@addpunct{:}}]\ignorespaces
}{%
	$\blacksquare$\endtrivlist\@endpefalse
}
\makeatother

\renewcommand{\footnoterule}{%
  \kern-3pt
  \hrule width 2in height 0.4pt
  \kern 5.6pt
}

\DeclareMathOperator*{\argmin}{arg\,min}

\numberwithin{equation}{section}

\newtheorem{theorem}{Theorem}[section]
\newtheorem{proposition}{Proposition}[section]
\newtheorem{lemma}{Lemma}[section]

\theoremstyle{definition}
\newtheorem{definition}{Definition}[section]

\theoremstyle{definition}
\newtheorem{assumption}{Assumption}[section]

\theoremstyle{definition}
\newtheorem{example}{Example}[section]

\DeclareTextFontCommand{\bi}{%
	\fontseries\bfdefault 
	\itshape
}

\newcommand{\eqdef}{\coloneqq}

\title{}
\begin{document}
	\vspace*{5ex minus 1ex}
	\begin{center}
		\Large \textsc{Causal Inference with Groupwise Matching}
		\bigskip
	\end{center}
	
	\date{%
		\today%
	}

	\vspace*{7ex minus 1ex}
	\begin{center}
		Ratzanyel Rinc\'on and Kyungchul Song\\
		
		\textit{Vancouver School of Economics, University of British Columbia}
		\bigskip
		
		\today
	\end{center}
	
\begin{abstract}
{\footnotesize This paper examines methods of causal inference based on groupwise matching when we observe multiple large groups of individuals over several periods. We formulate causal inference validity through a generalized matching condition, generalizing the parallel trend assumption in difference-in-differences designs. We show that difference-in-differences, synthetic control, and synthetic difference-in-differences designs are distinguished by the specific matching conditions that they invoke. Through regret analysis, we demonstrate that difference-in-differences and synthetic control with differencing are complementary; the former dominates the latter if and only if the latter's extrapolation error exceeds the former's matching error up to a term vanishing at the parametric rate. The analysis also reveals that synthetic control with differencing is equivalent to difference-in-differences when the parallel trend assumption holds for both the pre-treatment and post-treatment periods. We develop a statistical inference procedure based on synthetic control with differencing and present an empirical application demonstrating its usefulness.}\bigskip

{\footnotesize \ }

{\footnotesize \noindent \textsc{Key words.} Causal Inference; Difference-in-Differences; Synthetic Control Methods; Synthetic Difference-in-Differences; Parallel Trend Assumption; Generalized Matching Conditions}\bigskip

{\footnotesize \noindent \textsc{JEL Classification: C01, C18, C21}}
\end{abstract}
\thanks{We thank Tom Chan, Hiro Kasahara, Yiqi Liu, Paul Schrimpf and participants in Econometrics Lunch Seminar at UBC for valuable comments. We also thank Chun Pang Chow for his excellent research assistance and outstanding work in building the \texttt{R} package implementing the SCD method proposed in this paper. The package is publicly available at \url{https://github.com/ratzanyelrincon/scd}. Song acknowledges that this research was supported by Social Sciences and Humanities Research Council of Canada. All errors are ours. Corresponding
	address: Kyungchul Song, Vancouver School of Economics, University of
	British Columbia, 6000 Iona Drive, Vancouver, BC, V6T 1L4, Canada. Email address: kysong@mail.ubc.ca.}
\maketitle

\pagebreak

\section{Introduction}
 
A recent stream of literature provides a systematic comparison between different causal inference designs, especially between the synthetic control (SC) and other designs such as difference-in-differences (DID) or matching (\cite{Doudchenko/Imbens:ARXIV:17}, \cite{Ferman/Pinto:QE:21}, \cite{Kellogg/Mogstad/Pouliot/Torgovitsky:21:JASA}, \cite{Arkhangelsky/Athey/Hirshberg/Imbens/Wager:AER:21} and \cite{Chen:Eca:23}). However, the comparison falls short of giving a full picture, because it assumes a data structure inspired by the SC methods. The data structure assumes cross-sectional units of similar or smaller magnitude than the time periods. Furthermore, it is not uncommon in this literature that the treatment occurs only for a single cross-sectional unit.

We take an opposite direction by studying the SC design and its variants from the DID perspective, assuming a data structure that involves multiple large groups of individuals observed over a short period of time. Recent advances in the literature of DID designs consider multiple untreated groups such as in settings with staggered adoption and heterogenous causal effects (see \cite{Callaway/SantAnna:JoE:21}, \cite{deChaisemartin/DHaultfoeulle:AER:20}, \cite{Goodman-Bacon:21:JOE}, \cite{Sun/Abraham:JoE:21} and surveys by \cite{deChaisemartin/DHaultfoeulle:EJ:23} and \cite{Roth/SantAnna/Bilinksi/Poe:JoE2023}). Thus, the SC approach naturally maps to this DID framework with multiple ``donor groups'', by matching a counterfactual untreated group mean $\mu_0$ to a weighted average of group means $\mu_j$ in the ``donor pool'':  
\begin{align}
    \label{GMC 0}
    \mu_0 = \sum_{j} \mu_{j} w_j.
\end{align}
We call such causal inference methods \textit{groupwise matching}.\footnote{There are works that use groupwise matching in the SC approach (see \cite{Robbins/Saunders/Kilmer:17:JASA}, \cite{Xu:PA:17}, and \cite{Sun/Xie/Zhang:arXiv:25}). See also \cite{Gunsilius:Eca:23} who use quantiles instead of means in groupwise matching.}

As we show in this paper, the DID design can be thought of as arising from groupwise matching like SC. The main difference lies in the choice of the weights $w_j$. In the SC approach, the weights are chosen to minimize the pre-treatment matching errors, whereas in the DID approach, as this paper shows, the weights are chosen based on the fraction of the group sizes in the donor pool. The difference originates from two distinct thoughts on how we extrapolate the observed untreated outcomes to the counterfactual untreated outcomes for the treated units. The DID method matches the counterfactual mean untreated outcome to a pre-specified surrogate control group, whereas the SC method relies on the \textit{stability of matching} as we move from the pre-treatment to the post-treatment periods.

In this paper, we formalize the complementarity of these two thoughts using a generalized version of the condition (\ref{GMC 0}) that we call Generalized Matching Condition (GMC). More specifically, let $\mu_{j,t}(0)$ be a within-group-differenced, mean untreated potential outcome for group $j$ at time $t$. For a choice of weights $w_j$, the population-level matching error from matching to target group 0 is defined as follows:
\begin{align}
    \label{matching error2}
    e_t(w) = \mu_{0,t}(0) - \sum_{j} \mu_{j,t}(0) w_j,
\end{align}
where the sum is over the groups in the donor pool. Then the GMC simply says that $e_t(w) = 0$ for all post-treatment periods $t$.\footnote{See \cite{Shi/Sridhar/Misra/Blei:22:AISTATS} for an investigation of primitive assumptions that yield this condition.} Recent advances in the SC literature inspire various causal inference methods in this groupwise matching setting, including the classic synthetic control (SC), synthetic difference-in-differences (SDID), and synthetic control with differencing (SCD), and as we show later, the GMC captures their key identifying assumptions.\footnote{The SDID design was proposed by \cite{Arkhangelsky/Athey/Hirshberg/Imbens/Wager:AER:21} and the SCD was considered in their comparison studies in \cite{Ferman/Pinto:QE:21} and \cite{Chen:Eca:23}. Like other methods, they are distinguished by the way the weights and within-group differencing method are chosen in the GMC. Details follow below.}

Within this GMC framework, we focus on the SCD design which applies the SC weights after performing within-group differencing to eliminate time-invariant individual heterogeneity in potential outcomes. DID assigns weights based on the relative sizes of groups within the donor pool. In contrast, SCD chooses weights that best match the weighted average of donor group outcomes to the untreated outcomes for the treated group, yet this matching occurs only on the pre-treatment outcomes, not on the post-treatment outcomes. Consequently, SCD suffers from extrapolation error when the weights that achieve the best pre-treatment match fail to provide an adequate post-treatment match. On the other hand, DID's reliance on group-size-based weights makes it vulnerable to matching error if the surrogate control group is misspecified. Therefore, the relative performance of SCD versus DID depends fundamentally on SCD's extrapolation error against DID's matching error.\footnote{Extrapolation in the SC literature usually refers to the use of a match lying outside the convex combination of the outcomes in the donor pool. On the other hand, extrapolation here refers to the use of the same weights obtained from the pre-treatment fit to produce a surrogate for the post-treatment counterfactual untreated mean outcome.}

We formalize this observation focusing on the setting with point-identified weights. Using the matching errors $e_{t}(w)$ in (\ref{matching error2}), we define the squared sum of matching errors: 
\begin{align*}
    \mathsf{SSME}_d(w) = \frac{1}{|\mathcal{T}_d|} \sum_{t \in \mathcal{T}_d} e_{t}^2(w), \quad d = 0,1,
\end{align*}
where $\mathcal{T}_0$ denotes the set of pre-treatment periods and $\mathcal{T}_1$ that of post-treatment periods. From this, we construct two quantities that are used to evaluate the choice of the weight vector $w$:
\begin{align*}
    \text{Matching Error in Regret: } \mathsf{MER}_d(w) &=  \mathsf{SSME}_d(w) - \inf_{\tilde w \in \Delta_{K-1}} \mathsf{SSME}_d(\tilde w), \text{ and }\\
    \text{Extrapolation Error: } \mathsf{\Delta MER}(w) &=  \mathsf{MER}_1(w) - \mathsf{MER}_0(w).
\end{align*}
Thus, $\mathsf{MER}_d(w)$ measures the matching error in regret form for the choice of weight $w$, whereas $\mathsf{\Delta MER}(w)$ measures how well the matching error in regret is extrapolated from the pre-treatment periods to the post-treatment periods. Let $w^{\mathsf{DID}}$ be the population-level weights specified by the DID design and $w^{\mathsf{SCD}}$ those by the SCD design. Our main result shows that
\begin{align*}
    \text{DID regret-dominates SCD, if }& \mathsf{\Delta MER}(w^{\mathsf{SCD}}) > \mathsf{MER}_1(w^{\mathsf{DID}}) + C \epsilon_n \text{ and }\\
    \text{SCD regret-dominates DID, if }& \mathsf{\Delta MER}(w^{\mathsf{SCD}}) < \mathsf{MER}_1(w^{\mathsf{DID}}) - C \epsilon_n,
\end{align*}
where $\epsilon_n$ is a term that vanishes at the parametric rate (with respect to the size of the cross-sectional units) and $C$ is a universal positive constant. Therefore, the domination of SCD over DID depends on the relative size of the matching error in regret to the extrapolation error.

One might wonder when the designs of DID and SCD are ``equivalent'', in the sense that 
\begin{align*}
    w^{\mathsf{DID}} = w^{\mathsf{SCD}}.
\end{align*}
We demonstrate that this equivalence holds when both pre-treatment and post-treatment parallel trend assumptions hold simultaneously. This latter condition is implicitly invoked in practice when researchers use pre-treatment parallel trend tests as supporting evidence for the post-treatment parallel trend assumption. Such usage assumes that satisfying the post-treatment parallel trend assumption necessarily implies satisfying the pre-treatment parallel trend assumption (\cite{Kahn-Lang/Lang:20:JBES}).\footnote{See \cite{Bilinski/Hatfield:19:arXiv} also for issues with the usual pre-treatment tests and new proposals of tests addressing them.} Under these conditions, our results show that DID and SCD employ identical weights and therefore rely on the same identifying assumption. Nevertheless, the finite-sample performance of estimates from these approaches may still differ.

Our complementarity result demonstrates that SCD emerges as a viable alternative to DID when the parallel trend assumption fails. Unlike approaches that robustify DID against the failure of the parallel trend assumption (see \cite{Manski/Pepper:18:ReStat} and \cite{Rambachan/Roth:23:ReStud}), SCD is inspired by the SC design and replaces the parallel trend assumption by the existence and stability of matching weights before and after the treatment.\footnote{There have been variants of DID that do not require parallel trend assumption. For example, \cite{Freyaldenhoven/Hansen/Shapiro:19:AER} considered a linear panel framework where the violation of parallel trends is permitted and identification is achieved by removing possible confounding through the use of covariates. \cite{Kwon/Roth:24:AEAPP} proposed an empirical Bayes approach.} Just as the plausibility of the parallel trend assumption has to be examined in the specific context of application, so does the stable matching weight assumption of SCD.

While SCD has already been considered in the literature (\cite{Ferman/Pinto:QE:21} and \cite{Chen:Eca:23}), the uniformly valid asymptotic inference for the SCD design for the groupwise matching setting has not been formally developed to the best of our knowledge. We fill this gap by applying the uniformly valid inference on the simplex-valued weights in \cite{Canen/Song:arXiv:25} and developing estimation and asymptotic inference methods for SCD. Our Monte Carlo simulations show how the complementarity between SCD and DID manifests in finite sample performance of the estimators.

To illustrate the usefulness of SCD as a causal inference method, we revisit the empirical setting analyzed in \cite{Bohn/Lofstrom/Raphael:TRES:14} and examine the impact of the 2007 Legal Arizona Workers Act (LAWA) on Arizona's internal composition. We use CPS data between January 1998 and December 2009 and exploit its cross-sectional dimension to provide valid confidence intervals for the treatment effects estimated by SCD. Following the authors, we include 46 states in Arizona's donor pool that did not implement any similar regulation during the period of analysis and focus on the population that is most likely to be affected by the policy change: non-citizen Hispanics. We find that Arizona's share of this demographic group declined by 2 percentage points after LAWA's enactment on average, consistent with the 1.5 percentage point reduction reported in \cite{Bohn/Lofstrom/Raphael:TRES:14}. The average decrease is 1.2 percentage points larger when looking at Arizona's proportion of low-educated non-citizen Hispanics among the population aged 15-45. These results are robust to alternative choices of the pre-treatment window and differencing parameters in the SCD design.
\medskip

\noindent \textbf{Related Literature} The literature of SC designs and DID designs is vast and fast growing. We refer the readers to the survey papers by \cite{Abadie:JEL2021} for the SC approaches, and \cite{deChaisemartin/DHaultfoeulle:EJ:23} and \cite{Roth/SantAnna/Bilinksi/Poe:JoE2023} for the DID designs. Here, we will briefly focus only on some recent studies that attempt to synthesize and/or compare the SC and DID designs.

\cite{Doudchenko/Imbens:ARXIV:17} presented a unifying framework that encompasses four major causal inference approaches (SC, DID, matching and regression). \cite{Kellogg/Mogstad/Pouliot/Torgovitsky:21:JASA} compared SC with matching methods in terms of extrapolation and interpolation bias and proposed a model average estimator of the two approaches. \cite{Arkhangelsky/Athey/Hirshberg/Imbens/Wager:AER:21} synthesized SC and DID into what they called SDID (synthetic difference-in-differences). \cite{Xu:PA:17} assumed a linear factor structure for untreated potential outcomes and proposed extrapolating the estimated factor loadings and factors to accommodate time-varying confounders. In contrast, our focus is on formalizing the complementarity between DID and SC under short panels without relying on a factor structure.

Our findings contrast with recent work by \cite{Ferman/Pinto:QE:21} and \cite{Chen:Eca:23}, both of which showed that SCD dominates DID. \cite{Ferman/Pinto:QE:21} employ a linear factor model to demonstrate that SCD's asymptotic mean-squared error dominates that of DID under large-$|\mathcal{T}_0|$ asymptotics. Our analysis differs fundamentally: we impose no linear factor structure and examine environments with many individuals over short time periods. \cite{Chen:Eca:23} compares SCD and DID through a regret analysis but assumes long time periods, rendering his results uninformative for our setting. Moreover, Chen's risk definition assumes treatment timing is drawn randomly from an approximately uniform distribution. Our analysis takes the opposite extreme, assuming fully known treatment timing.

Closely related is \cite{Sun/Xie/Zhang:arXiv:25}, who also consider short panels and propose identification and inference on the ATT accommodating both DID and SC settings. However, their parallel trend assumption is strong enough to identify the ATT under any design, whereas our comparison uses a weaker version sufficient for identification via DID but not necessarily SC. Also related is \cite{Liu:25:arXiv}, who introduces a synthetic parallel trend assumption similar to our generalized matching condition. Independently, she found that DID methods can be viewed as matching with group-size-based weights. Her focus, however, is on the identified set of counterfactuals and inference thereon, without further restrictions on weights other than that they sum up to one or simplex-constrained. By contrast, our generalized matching condition is parametrized by weights and time-differencing methods, as our goal is to characterize the core identifying scheme underlying both DID and SC, two designs that differ precisely in these choices. The resulting complementarity and equivalence results between DID and SCD are, to our knowledge, new.

The paper is organized as follows. In Section \ref{sec:causal inf groupwise matching}, we provide a basic set-up of causal inference with groupwise matching and the notion of GMC. In Section \ref{sec: causal inf generalized matching}, we show how GMC provides a unifying identification scheme that encompasses various causal inference designs. In Section \ref{sec: Comparison}, we present a regret analysis that shows how the designs of DID and SCD are complementary to each other. In Section \ref{sec: SCD}, we provide estimation and inference of the SCD methods and results on asymptotic theory, followed by the Monte Carlo simulation results. In Section \ref{sec: emp app}, we present an empirical application. In Section \ref{sec: conclusion}, the paper concludes. The mathematical proofs of the results in the paper are found in the Supplemental Note.

\section{Causal Inference with Groupwise Matching}\label{sec:causal inf groupwise matching}

\subsection{The Set-Up}
\label{sec: the model and the target parameter}

We consider a setting with the set $N$ of individuals $i$, divided into $K+1$ disjoint groups. Let $\mathcal{G}\eqdef \{0,1,...,K\}$ be a finite set of group indexes and denote $G_i = j \in \mathcal{G}$ if and only if the individual $i$ belongs to group $j$. The individuals are observed over time $t \in \mathcal{T} \eqdef \{1,2,...,T\}$. Each individual belongs either to the treatment group ($D_{i} = 1$) or the untreated group $(D_{i} = 0)$. All the groups stay untreated until time $t= T^* > 1$, and at time $T^*$, those individuals with $D_{i} = 1$ are treated. We partition $\mathcal{T}$ into $\mathcal{T}_{0}$ and $\mathcal{T}_{1}$, with
\begin{align*}
    \mathcal{T}_0 = \{1,...,T^*-1\}, \text{ and } \mathcal{T}_1 = \{T^*,...,T\}.
\end{align*}
The set $\mathcal{T}_0$ collects the time periods before treatment occurs and $\mathcal{T}_1$ the time periods following treatment. Hereafter, we call $\mathcal{T}_0$ the pre-treatment periods and $\mathcal{T}_1$ the post-treatment periods.  

The potential outcome of an individual $i$ in time $t$ when the individual is treated is denoted by $Y_{i,t}(1)$ and otherwise $Y_{i,t}(0)$. Define the average treatment effect on the treated in period $t$ as
\begin{align}
    \label{theta_t^*}
	\theta_t^* = \mathbf{E}\left[ Y_{i,t}(1) - Y_{i,t}(0) \mid  D_{i} = 1 \right].
\end{align}
The observed outcomes, $Y_{i,t}$, are defined as follows: 
\begin{align}
    \label{Y_{it}}
    Y_{i,t} = D_{i} Y_{i,t}(1) + (1-D_{i}) Y_{i,t}(0).
\end{align}

We introduce basic conditions maintained throughout the paper.

\begin{assumption} 
    \label{assump:basic}

	(i) $P\{D_i = 1\} >0$ and $P\{G_i=j,D_{i} = 0\}>0$, for all $j = 1,...,K$ and $i \in N$.

	(ii) For each $i \in N$ such that $D_i = 1$, we have $Y_{i,t}(1)=Y_{i,t}(0)$, whenever $t< T^*$.
\end{assumption}

Assumption \ref{assump:basic}(i) says that each group consists of a positive fraction of individuals in the population. Assumption \ref{assump:basic}(ii) supposes \textit{no anticipation} of treatment. It says that each individual's potential outcome at time $t$ before the treatment at time $s$ is the same as that when the person is never treated. As we show in the following example, this setting accommodates the DID with discrete covariates and the staggered adoption in the DID literature (\cite{Callaway/SantAnna:JoE:21}, \cite{deChaisemartin/DHaultfoeulle:AER:20}, and \cite{Sun/Abraham:JoE:21}).

\begin{example}[Difference-in-Differences with Discrete Covariates]
    Consider the two-period DID setting with covariates. Suppose that $T = 2$ and each individual $i \in N$ is endowed with the discrete covariate $X_i \in \{x_1,...,x_K\}$, and belongs to either the treated group ($D_i = 1$) or the untreated group ($D_i = 0$). The potential outcomes are given as $Y_{i,t}(1)$ and $Y_{i,t}(0)$ for $t=1,2$. This setting is mapped to the above general setting, by setting $G_i = j$ if and only if $X_i = x_j$. The treatment time $T^*$ is taken to be the second time, i.e., $T^*=2$. $\blacksquare$
\end{example}

\begin{example}[Difference-in-Differences with Staggered Adoption]
    \label{exmpl: staggered adoption}
In this example, each group may experience different treatment timing. First, we define $D_{j,t}$ as a binary variable equal to one if group $j$ is treated at time $t$ and zero otherwise. We decompose the group index set $\mathcal{G}$ as follows:
\begin{align*}
    \mathcal{G} = \{0\} \cup \mathcal{G}_{\mathsf{don}},
\end{align*}
where $\mathcal{G}_{\mathsf{don}} = \{1,2,...,K\}$. Our focus is on the effect of the first-time treatment of the target group $0$. Let $T^*$ be the first time that group 0 is treated. The groups in $\mathcal{G}_{\mathsf{don}}$ have not been treated until after $t = T$ and thus form a ``donor pool'' for the target group.

We assume that $T^*>1$. For each individual $i \in N$ such that $D_{G_i,t} = 1$ for some $t \in \mathcal{T}$, let 
\begin{align*}
    T_i = \min\{ t \in \mathcal{T}:D_{G_i,t}=1\}
\end{align*}
be the period in which individual $i$ is first treated. Hence, $T_i = T^*$ for all $i$ in the target group 0. We set $T_i=0$ for an individual who is never treated. We let $Y_{i,t}^*(s)$ be the potential outcome of individual $i$ at time $t$ when group $G_i$ is first treated at time $1 \le s \le T$. The quantity $Y_{i,t}^*(0)$ represents the potential outcome of the individual $i$ when never treated. 

Our parameter of interest is the average treatment effect on the target group $0$: 
\begin{align*}
    \theta_t^* = \mathbf{E}[Y_{i,t}^*(T^*) - Y_{i,t}^*(0) \mid G_i = 0], \text{ for } t \in \mathcal{T}_{1}.
\end{align*}
The observed outcomes are given as follows:
\begin{align}
    \label{Y_{it}2}
	Y_{i,t} = D_{G_i,t} Y_{i,t}^*(T_i) + (1 - D_{G_i,t}) Y_{i,t}^*(0).
\end{align}

We consider the following assumptions.\medskip

(i) $Y_{i,t}^*(s) = Y_{i,t}^*(0)$ for all $t < s$ with $t,s \in \mathcal{T}$ and $i \in N$.

(ii) $D_{j,1} = 0$ for all $j \in \mathcal{G}$, if $T \ge 2$.

(iii) $D_{j,t} \le D_{j,t+1}$ for all $j \in \mathcal{G}$ and $t = 1,...,T-1$.

(iv) (a) $T_i = T^*$ whenever $G_i = 0$, and (b) $T_i = 0$ whenever $G_i \in \mathcal{G}_{\mathsf{don}}$.

(v) For each $j \in \mathcal{G}$, $P\{G_i = j\}>0$.\medskip 

Condition (i) is the standard no-anticipation assumption. The next two conditions require that no group is treated at the initial period (ii) and that treatments arise in a staggered manner (iii). Furthermore, (iv) stipulates that $T^*$ is the first time group $0$ is treated and that the donor pool groups remain untreated through period $T$, while (v) ensures that each group has a positive membership probability. Note that under Condition (iii), $D_{G_i,t} = 0$ indicates that individual $i$ has not been treated by time $t$, whereas $D_{G_i,t} = 1$ indicates treatment at or before time $t$.

Now, we show how this setting maps to the general setting above. We define
\begin{align*}
    D_{i} = 1\{G_i = 0\}, \text{ for all } i \in N.
\end{align*}
In this case, individuals with $D_i = 1$ represent people belonging to group 0, while those with $D_i = 0$ correspond to people who belong to groups that remain untreated until after $T$. Then, we set $Y_{i,1}(1) = Y_{i,1}(0) = Y_{i,1}^*(0)$ for all $i \in N$. And, for $1 < t \in \mathcal{T}$, we set
\begin{align*}
    Y_{i,t}(1) = Y_{i,t}^*(T_i) \text{ and } Y_{i,t}(0) = Y_{i,t}^*(0), \text{ for all } i \in N.
\end{align*}
Hence, for $t \in \mathcal{T}_{1}$, we can write
\begin{align*}
    \theta_t^* = \mathbf{E}[Y_{i,t}(1) - Y_{i,t}(0) \mid D_i = 1].
\end{align*}
Furthermore, the observed outcomes, $Y_{i,t}$, defined in (\ref{Y_{it}2}) coincide with those defined in (\ref{Y_{it}}), for all $i \in N$ and $t \in \mathcal{T}$. It is not hard to see that Assumption \ref{assump:basic} is also satisfied. $\blacksquare$
\end{example}     

\subsection{Generalized Matching}

The causal effect of a treatment in an experimental setting is captured by the difference in outcomes between the treated and control groups. In a non-experimental setting, a control group is not available, which requires constructing a comparison group as a surrogate for the control group. This approach is valid only if the outcomes of the comparison group are ``matched'' to the counterfactual untreated outcomes of the treated group. 

To express this idea, for each $j \in \mathcal{G}_{\mathsf{don}}$, and $t \in \mathcal{T}$, define
\begin{align*}
    m_{0,t}(0) = \mathbf{E}\left[ Y_{i,t}(0) \mid D_i = 1 \right] \text{ and } m_{j,t}(0) = \mathbf{E}\left[ Y_{i,t}(0) \mid D_i = 0, G_i = j \right],
\end{align*}
and their observed counterparts: 
\begin{align*}
    m_{0,t} = \mathbf{E}\left[ Y_{i,t} \mid D_i = 1 \right] \text{ and } m_{j,t} = \mathbf{E}\left[ Y_{i,t} \mid D_i = 0, G_i = j \right].
\end{align*}

Let $\Delta_{|\mathcal{T}_0|-1} \subset \mathbf{R}^{|\mathcal{T}_0|}$ be the simplex in $\mathbf{R}^{|\mathcal{T}_0|}$. Given $\lambda = (\lambda_s)_{s \in \mathcal{T}_0} \in \Delta_{|\mathcal{T}_0|-1}$, and $j \in \mathcal{G}$, we introduce a \bi{within-group $\lambda$-differencing} of $m_{j,t}(0)$ and $m_{j,t}$ as follows:
\begin{align*}
	\mu_{j,t}(0;\lambda) = m_{j,t}(0) - \sum_{s \in \mathcal{T}_0} m_{j,s}(0)\lambda_s \text{ and } \mu_{j,t}(\lambda) = m_{j,t} - \sum_{s \in \mathcal{T}_0} m_{j,s} \lambda_s.
\end{align*}
One example is to subtract the most recent pre-treatment potential outcome so that
\begin{align}
    \label{stochastic differencing}
    \mu_{j,t}(0;\lambda^{\mathsf{DID}}) = m_{j,t}(0) - m_{j,T^*-1}(0),
\end{align}
where $\lambda_s^{\mathsf{DID}} = 1\{s = T^* - 1 \}$. This differencing is adopted in the DID designs of \cite{Callaway/SantAnna:JoE:21}, \cite{Sun/Abraham:JoE:21}, and \cite{deChaisemartin/DHaultfoeulle:AER:20}. Another example is the uniform differencing  
\begin{align*}
    \mu_{j,t}(0;\lambda^{\mathsf{unif}}) = m_{j,t}(0) - \frac{1}{|\mathcal{T}_0|} \sum_{s \in \mathcal{T}_0} m_{j,s}(0),
\end{align*}
where $\lambda_s^{\mathsf{unif}}=1/|\mathcal{T}_0|$ (see \cite{Wooldridge:21:WP} and \cite{Lee/Wooldridge:24:WP} for DID applications and \cite{Arkhangelsky/Athey/Hirshberg/Imbens/Wager:AER:21} and \cite{Chen:Eca:23} in the SDID and SCD designs).

For each $w \in \Delta_{K-1}$, $\lambda \in \Delta_{|\mathcal{T}_0|-1}$ and $t \in \mathcal{T}$, we define a \bi{between-group $w$-differencing} of the $\lambda$-differenced average potential outcomes as follows:
\begin{align*}
    e_t(\lambda,w) = \mu_{0,t}(0;\lambda) - \sum_{j=1}^K \mu_{j,t}(0;\lambda) w_j.
\end{align*}
We call the quantity $e_t(\lambda,w)$ the \bi{matching error} from matching $\mu_{0,t}(0;\lambda)$ with a weighted average of group means $\mu_{j,t}(0;\lambda)$ in the donor pool. 

\begin{definition} Let $\lambda \in \Delta_{|\mathcal{T}_0|-1}$ be given.
    
(i) For $w \in \Delta_{K-1}$, we say that \bi{Generalized Matching Condition (GMC) holds at $(\lambda,w)$}, if 
\begin{align*}
    e_t(\lambda,w) = 0, \text{ for all } t \in \mathcal{T}_1.
\end{align*}

(ii) We say that \bi{Stable Matching Condition (SMC) holds at $\lambda$}, if for some $w \in \Delta_{K-1}$,
\begin{align*}
    e_t(\lambda,w) = 0, \text{ for all } t \in \mathcal{T}.
\end{align*}
\end{definition}

Suppose that GMC holds at $(\lambda,w)$. This means that we can \textit{transfer} the $w$-weighted average of the expected untreated potential outcomes in the donor pool (after the within-group $\lambda$-differencing) to the corresponding counterfactual quantity in the target group. To see the role of GMC in identifying $\boldsymbol{\theta}^*$, we decompose the target parameter $\theta_t^*$ as follows:\footnote{The proof is simple and found in the Supplemental Note.}
\begin{align}
    \label{theta^* decomp}
    \theta_t^* = \theta_t(\lambda,w) - e_t(\lambda,w),
\end{align} 
where
\begin{align*}
    \theta_t(\lambda,w) = \mu_{0,t}(\lambda) - \sum_{j=1}^{K} \mu_{j,t}(\lambda) w_j.
\end{align*}
Once we invoke GMC at $(\lambda,w)$, we obtain the following identification:
\begin{align}
    \label{ident ATT}
    \boldsymbol{\theta}^* = \boldsymbol{\theta}(\lambda,w),
\end{align}
where 
\begin{align*}
    \boldsymbol{\theta}(\lambda,w) = [\theta_{T^*}(\lambda,w),...,\theta_{T}(\lambda,w)]'.
\end{align*}
As we will see later, many causal inference designs are distinguished by how $\lambda$ and $w$ are specified. Due to the use of groupwise matching, the estimand $\boldsymbol{\theta}(\lambda,w)$ depends on the individual-level observations only through the group averages $m_{j,t}$. Hence, the causal inference framework accommodates both repeated cross-sections and panel data.

\subsubsection{Generalized Matching Conditions under a Linear Factor Model}

The literature often specifies the potential outcomes as a linear factor model to analyze a causal inference method (see \cite{Abadie/Diamond/Hainmueller:10:JASA}, \cite{Xu:PA:17}, \cite{Ferman/Pinto:QE:21}, \cite{Arkhangelsky/Athey/Hirshberg/Imbens/Wager:AER:21}). While our results do not rely on a linear factor model, it is interesting to study the implication of this model for GMC. 

Consider the untreated potential outcomes specified as a factor model:
\begin{align}
    \label{factor model}
    Y_{i,t}(0) = \Lambda_{i}' F_t + \varepsilon_{i,t} \text{ and } Y_{i,t}(1) = Y_{i,t}(0) + \tau_{i,t},
\end{align} 
where $\Lambda_i \in \mathbf{R}^M$ denotes the factor loading of individual $i$, $F_t \in \mathbf{R}^M$, the factor at period $t$, $\varepsilon_{i,t}$, idiosyncratic components, and $\tau_{i,t}$ denotes the time-varying, heterogeneous treatment effects. We assume that $D_i = 1$ if and only if $G_i = 0$, so that there is a treated group $G_i = 0$ and all other groups are control groups. The number $M$ represents the number of factors. As for the factor model, we make the following assumption.

\begin{assumption}
    \label{assump: factor}
    (i) The factor loadings, $\Lambda_i$, $i \in N$, are i.i.d.

    (ii) The factors, $F_t$, $t \in \mathcal{T}$, are constants.
    
    (iii) For each $i \in N$ and $t \in \mathcal{T}$, $\mathbf{E}[\varepsilon_{i,t} \mid G_i] = 0.$ 
\end{assumption}

The condition (ii) is motivated by the data structure of our setting where our observations span over only a short period. Hence, the distribution of the factors is not consistently estimable even if the factors are observed (see \cite{Kuersteiner/Prucha:Eca:20}). The rest of the analysis carries over to the case of stochastic factors, once we replace probabilities and expectations by conditional probabilities and conditional expectations given the factors.

We introduce the $\lambda$-differenced versions of the factors:
\begin{align*}
    F_t(\lambda) = F_t - \sum_{s \in \mathcal{T}_0} F_s \lambda_{s},
\end{align*}
and collect them into a matrix, $\mathbf{F}(\lambda) = \left[ F_{T^*}(\lambda),...,F_{T}(\lambda) \right]$. Then, it follows that when $\mathbf{F}(\lambda)$ is full row rank, the GMC holds at some $(\lambda,w)$ if and only if the GMC holds at $(\tilde \lambda,w)$ for all $\tilde \lambda \in \Delta_{|\mathcal{T}_0|-1}$. Hence, the choice of $\lambda$ in the $\lambda$-differencing does not matter for identification, as long as the GMC holds at some $\lambda$. We formalize this into the following proposition.

\begin{proposition}
    \label{prop: GMC factor}
    Suppose that Assumption \ref{assump: factor} holds and let $\lambda \in \Delta_{|\mathcal{T}_0|-1}$ and $w \in \Delta_{K-1}$.\medskip
    
    (i) GMC holds at $(\lambda,w)$.
    
    (ii) GMC holds at $(\tilde \lambda,w)$ for all $\tilde \lambda \in \Delta_{|\mathcal{T}_0|-1}$.

    (iii)
    \begin{align}
        \label{factor loadings equality}
        \mathbf{E}[ \Lambda_i \mid G_i = 0] = \sum_{j=1}^K \mathbf{E}[ \Lambda_i \mid G_i = j] w_j.
    \end{align}
    Then, (iii) $\Rightarrow$ (ii) $\Rightarrow$ (i). If, furthermore, $\mathbf{F}(\lambda)$ is full row rank for some $\lambda \in \Delta_{|\mathcal{T}_0|-1}$, then the three statements are equivalent for all $w \in \Delta_{K-1}$.
\end{proposition}

The full row rank condition for $\mathbf{F}(\lambda)$ requires that $M \le |\mathcal{T}_1|$, that is, the number of the factors is less than the number of the post-treatment periods. This condition is immediately satisfied in the case of a single-factor model. The full rank condition is not required for identification of $\boldsymbol{\theta}^*$ once (\ref{factor loadings equality}) is satisfied for some $(\lambda,w)$. However, if the full rank condition holds, we have 
\begin{align*}
    \boldsymbol{\theta}^* = \boldsymbol{\theta}(\tilde \lambda,w) \text{ for all } \tilde \lambda \in \Delta_{|\mathcal{T}_0|-1},
\end{align*}
i.e., $\boldsymbol{\theta}^*$ is overidentified. For the identification, the researcher does not need to specify the within-group differencing, $\tilde \lambda$, that satisfies GMC.

\section{Causal Inference Methods using Generalized Matching}
\label{sec: causal inf generalized matching}

In this section, we show how GMC is used as key identifying restrictions in various causal inference designs. We classify them into two categories, one using GMC with weights based on group sizes and the other using GMC with weights based on pre-treatment fit.

\subsection{Matching with Weights Based on Group Sizes}

\subsubsection{Randomized Controlled Trials}

First, note that the design of randomized control trials (RCT) can be viewed as a degenerate example of GMC, where we do not have the initial period of no treatment, i.e., $|\mathcal{T}_0| = 0$. The design assumes that the potential outcomes are independent of the treatment status and yields the following form of GMC at $(0,1)$:
\begin{align*}
    e_1(0,1) = \mathbf{E}[Y_{i,1}(0) \mid D_i = 1] - \mathbf{E}[Y_{i,1}(0) \mid D_i = 0] = 0.
\end{align*}
The parameter $\theta_t^*$ is identified as 
\begin{align*}
    \theta_1^* = \theta_1(0,1) = \mathbf{E}[Y_{i,1} \mid D_i = 1] - \mathbf{E}[Y_{i,1} \mid D_i = 0].
\end{align*}

\subsubsection{Unconfoundedness Condition with Discrete Covariates}

Consider the unconfoundedness condition on discrete random vector $X_i$ with support $\{x_1,...,x_K\}$: 
\begin{align*}
    (Y_i(1),Y_i(0)) \CI D_i \mid X_i.
\end{align*}
Our object of interest is the ATT, $\theta_1^* = \mathbf{E}[Y_{i,1}(1) - Y_{i,1}(0) \mid D_i = 1]$. We take $G_i = j$ if and only if $X_i = x_j$. The unconfoundedness condition, together with the overlap condition, $P\{D_i = 1 \mid G_i = j\}>0$ for all $j$, yields the following: 
\begin{align*}
    0 &= \mathbf{E}[Y_{i,1}(0) \mid D_i = 1] - \mathbf{E}\left[\mathbf{E}[Y_{i,1}(0) \mid D_i = 0, G_i] \mid D_i = 1 \right]\\
    &= \mathbf{E}[Y_{i,1}(0) \mid D_i = 1] - \sum_{j=1}^K \mathbf{E}[Y_{i,1}(0) \mid D_i = 0, G_i = j] w_j^{\mathsf{C}},
\end{align*}
where 
\begin{align}
    \label{group size}
    w_j^{\mathsf{C}} = P\{G_i = j \mid D_i = 1\}.
\end{align}
The identifying assumption is essentially GMC at $(0,w^{\mathsf{C}})$ with $w^{\mathsf{C}} = (w_j^{\mathsf{C}})$.

\subsubsection{Difference-in-Differences}

First, consider the two-period setting $T = 2$ of the classic DID, where $Y_{i,t}(d)$ denotes the potential outcome at time $t = 1,2$ at the treatment state $d \in \{0,1\}$. We consider the following form of the RCT after within-group differencing:
\begin{align*}
    Y_{i,t}(0;\lambda) \CI D_i,
\end{align*}
where $Y_{i,t}(0;\lambda) = Y_{i,t}(0) - \sum_{s \in \mathcal{T}_0} \lambda_s Y_{i,s}(0)$ and $\lambda = 1$ (since $|\mathcal{T}_0| = 1$). This yields the following parallel trend assumption: 
\begin{align*}
    0 &= \mathbf{E}[Y_{i,t}(0;\lambda) \mid D_i = 1] - \mathbf{E}[Y_{i,t}(0;\lambda) \mid D_i = 0] \\
    &= e_t(\lambda,1) = e_t(1,1).
\end{align*}
Thus, the parallel trend assumption is nothing but the GMC at $(1,1)$.

\subsubsection{Difference-in-Differences with Discrete Covariates}

We consider the two-period setting as before, but consider the following form of the unconfoundedness condition instead:\footnote{The unconfoundedness condition for within-group differenced outcomes was studied by \cite{Heckman/Ichimura/Todd:TRES:97}. They showed the efficacy of the differencing using Job Training Program Act (JTPA) data. See \cite{Smith/Todd:JoE:05} for a similar observation using National Supported Work (NSW) data.} 
\begin{align*}
    Y_{i,t}(0;\lambda) \CI D_i \mid X_i,
\end{align*}
with $X_i \in \{x_1,...,x_K\}$ being a discrete random vector. As before, if we let $G_i = j$ if and only if $X_i = x_j$, this condition, together with the overlap condition, yields GMC at $(\lambda,w^{\mathsf{C}})$ as follows: 
\begin{align*}
    0 &= \mathbf{E}[Y_{i,t}(0;\lambda) \mid D_i = 1] - \mathbf{E}\left[\mathbf{E}[Y_{i,t}(0;\lambda) \mid D_i = 0, G_i] \mid D_i = 1\right]\\
     &=  e_t(\lambda,w^{\mathsf{C}}) = e_t(1,w^{\mathsf{C}}),
\end{align*}
where $w^{\mathsf{C}} = (w_j^{\mathsf{C}})$ with $w_j^{\mathsf{C}}$ defined in (\ref{group size}).

\subsubsection{Difference-in-Differences with Staggered Adoption}
\label{subsubsec: DiD SA}
We consider the staggered adoption setting of Example \ref{exmpl: staggered adoption}. We show that GMC characterizes the key identifying assumptions in the DID settings. Let us consider the parallel trend assumptions (PTA) used in the literature. Let $\Delta Y_{i,t}(0) = Y_{i,t}(0) - Y_{i,t-1}(0)$, for $t \in \{2,...,T\}$. Consider the two types of PTA as follows.\medskip

\textbf{PTA-I:} $\mathbf{E}[\Delta Y_{i,t}(0) \mid G_i = 0] = \mathbf{E}[\Delta Y_{i,t}(0) \mid G_i \in \mathcal{G}_{\mathsf{don}}]$, for all $t \in \mathcal{T}_1$.

\textbf{PTA-II:} $\mathbf{E}[\Delta Y_{i,t}(0) \mid G_i = 0] = \mathbf{E}[\Delta Y_{i,t}(0) \mid G_i  = j]$, for all $t \in \mathcal{T}_1$ and $j \in \mathcal{G}_{\mathsf{don}}$.\medskip

PTA-I states that the average of the untreated potential outcomes of group 0 and those in its donor pool would have evolved in parallel in the absence of treatment (\cite{Callaway/SantAnna:JoE:21}). PTA-II is a stronger version of PTA-I, imposing parallel trends of untreated outcomes across all groups (similar to the exogeneity condition in \cite{deChaisemartin/DHaultfoeulle:AER:20} and \cite{Sun/Abraham:JoE:21}).

The following result shows a close connection between the PTA and the GMC.\footnote{The connection between PTA and GMC can be viewed as an extension of the observation in \cite{Doudchenko/Imbens:ARXIV:17} to the setting of multiple donor groups. See \cite{Liu:25:arXiv} for a related result.}

\begin{proposition}
    \label{prop: PTA}
    Suppose that Assumption \ref{assump:basic} holds, and let 
    \begin{align*}
        w^{\mathsf{DID}} = [w_{1}^{\mathsf{DID}},...,w_{K}^{\mathsf{DID}}]', 
    \end{align*}
    where
    \begin{align}
    \label{DID}
        w_{j}^{\mathsf{DID}} = P\left\{G_i = j \mid G_i \in \mathcal{G}_{\mathsf{don}} \right\}.
    \end{align} 
    Then, the following statements hold.

    (i) PTA-I holds if and only if GMC holds at $(\lambda^{\mathsf{DID}},w^{\mathsf{DID}})$.

    (ii) PTA-II holds if and only if GMC holds at $(\lambda^{\mathsf{DID}},w)$ for all $w \in \Delta_{K-1}$.
\end{proposition}

This proposition shows that the PTA is represented as GMC. Thus, the target parameter $\boldsymbol{\theta}^*$ is identified as $\boldsymbol{\theta}(\lambda^{\mathsf{DID}},w^{\mathsf{DID}})$ under either PTA-I or PTA-II. Instead of choosing the weight $w$ based on the pre-treatment matching of the potential outcomes as in the SC design, the DID design simply chooses the weight $w$ to be the group size-based one $w^{\mathsf{DID}}$ in (\ref{DID}). Under the stronger version PTA-II, the choice of the weight $w$ is irrelevant, as GMC holds for all weights.

It is interesting to note that the identification scheme (\ref{ident ATT}) is related to the proposal by \cite{Sun/Xie/Zhang:arXiv:25}.\footnote{\cite{Sun/Xie/Zhang:arXiv:25} also considered conditioning on covariates, and for estimation, proposed using an estimated weight $\hat w^{\mathsf{SC}}$ that is not restricted to the simplex $\Delta_{K-1}$. For brevity, we do not consider conditioning on covariates throughout the paper and focus on the main conceptual difference between the two approaches of DID and SCD.} The unconditional version of their model (without covariates) involves PTA-II, which is equivalent to GMC at $(\lambda^{\mathsf{DID}},w)$ for all $w \in \Delta_{K-1}$, and the SC assumption which is tantamount to SMC at $0$, with $w^{\mathsf{SC}}$ identified as the weight $w$ satisfying $e_t(0,w) = 0$ for all $t \in \mathcal{T}_0$. It is not hard to see that both PTA-II and the SC assumption imply GMC at $(\lambda^{\mathsf{DID}},w^{\mathsf{SC}})$. Hence, we can identify $\boldsymbol{\theta}^*$ as $\boldsymbol{\theta}(\lambda^{\mathsf{DID}},w^{\mathsf{SC}})$ when either PTA-II or the SC assumption holds. This is the essence of their doubly robust identification of the ATT in \cite{Sun/Xie/Zhang:arXiv:25}. However, PTA-II is stronger than PTA-I and the latter is enough to identify the ATT in the DID design.

\subsection{Matching with Weights Based on Pre-Treatment Fit}

The literature of SC inspires various causal inference methods in the groupwise matching setting. These methods are distinct from the previous methods, as they rely on GMC with weights based on the pre-treatment fit of the outcomes. For the following examples, we focus on the setting of staggered adoption in Example \ref{exmpl: staggered adoption} and Section \ref{subsubsec: DiD SA}.

\subsubsection{Synthetic Control}

The synthetic control method applied to the setting of groupwise matching identifies 
\begin{align*}
    \theta_t^* = \mathbf{E}[Y_{i,t} \mid G_i = 0] - \sum_{j=1}^K \mathbf{E}[Y_{i,t} \mid G_i = j] w_j^*(0),
\end{align*}
where $w^*(0) = [w_1^*(0),...,w_K^*(0)]'$ is a minimizer of $Q(w)$ over $w = [w_1,...,w_K]' \in \Delta_{K-1}$, with 
\begin{align*}
    Q(w) = \sum_{t =1}^{T^*-1} \left( m_{0,t} - \sum_{j=1}^{K} m_{j,t} w_j \right)^2.
\end{align*}
This identification scheme relies on GMC holding at $(0,w^*(0))$. 

The choice of $w^*(0)$ is motivated as follows. First, consider the weight $w_j^*(0)$ that gives the population-level perfect pre-treatment matching: for all $t \in \mathcal{T}_0$,
\begin{align}
    \label{pre-treatment matching}
    \mathbf{E}[Y_{i,t} \mid G_i = 0] = \sum_{j=1}^K \mathbf{E}[Y_{i,t} \mid G_i = j] w_j^*(0).
\end{align}
Then, we assume that the same weight $w_j^*(0)$ yields the perfect post-treatment matching as well, i.e., (\ref{pre-treatment matching}) holds for $t \in \mathcal{T}_1$. In other words, the SC design relies on SMC at $0$.

\subsubsection{Synthetic Control with Differencing}

The synthetic control with differencing (SCD) applies the SC design after applying a within-group $\lambda$-differencing of the potential outcomes (see \cite{Chen:Eca:23} and references therein). Let $\lambda$ be a researcher-chosen differencing method. For example, one may choose $\lambda = \lambda^{\mathsf{DID}} \text{ or } \lambda^{\mathsf{unif}}.$ Define $Q: \Delta_{|\mathcal{T}_0|-1} \times \Delta_{K-1} \rightarrow \mathbf{R}$ as follows:
\begin{align}
    \label{Q(lambda,w)}
	Q(\lambda,w) = \sum_{t =1}^{T^*-1} \left( \mu_{0,t}(\lambda) - \sum_{j=1}^{K} \mu_{j,t}(\lambda) w_j \right)^2,
\end{align}
and choose $w^*(\lambda)$ as a minimizer of $Q(\lambda,w)$: 
\begin{align}
    \label{w^*}
    w^*(\lambda) \in \argmin_{w \in \Delta_{K-1}} Q(\lambda,w).
\end{align}
Then the SCD design invokes GMC at $(\lambda,w^*(\lambda))$ and identifies $\theta_t^*$ as follows: 
\begin{align*}
    \theta_t^* = \theta_{t}(\lambda,w^*(\lambda)) = \mu_{0,t}(\lambda) - \sum_{j=1}^{K} \mu_{j,t}(\lambda) w_j^*(\lambda).
\end{align*}
The GMC at $(\lambda,w^*(\lambda))$ requires that the weight $w^*(\lambda)$ that achieves the optimal pre-treatment matching delivers the perfect post-treatment matching. 

Again, the choice of $w^*$ can be motivated in terms of SMC. We first consider $w_j^*(\lambda)$ such that 
\begin{align*}
    \mu_{0,t}(\lambda) = \sum_{j=1}^K \mu_{j,t}(\lambda)w_j^*(\lambda),
\end{align*}
for $t \in \mathcal{T}_0$. Then, the SCD design assumes that this weight $w_j^*(\lambda)$ delivers the perfect post-treatment matching as well, i.e., SMC holds at $\lambda$.

\subsubsection{Synthetic Difference-in-Differences}

\cite{Arkhangelsky/Athey/Hirshberg/Imbens/Wager:AER:21} developed the synthetic difference-in-differences (SDID) method, which integrates the synthetic control approach with the difference-in-differences design. While their original framework targeted a data structure different from our groupwise matching setting, the core idea of SDID can be adapted to this setting. 

To facilitate the comparison, suppose that our target parameter is the same as before $\boldsymbol{\theta}^*$. Consider the time-differenced version of the population objective function in \cite{Arkhangelsky/Athey/Hirshberg/Imbens/Wager:AER:21}:
\begin{align}
    \label{Q(tilde lambda,w)}
    \Tilde Q(\lambda,w) = \sum_{k=1}^{K} \left( \sum_{s \in \mathcal{T}_1} \left\{ \mu_{k,s}(\lambda) - \sum_{j=1}^{K} \mu_{j,s}(\lambda) w_j \right\} \right)^2.
\end{align}
We let $\lambda^{\mathsf{unif}}$ and $w^{\mathsf{unif}}$ be the uniform weights given by $\lambda_s^{\mathsf{unif}}=1/|\mathcal{T}_0|$ and $w_j^{\mathsf{unif}} = 1/K$. Then, the identification strategy of the SDID can be formulated as follows:
\begin{align*}
    \theta_t^* = \theta_t(\lambda^*(w^{\mathsf{unif}}),w^*(\lambda^{\mathsf{unif}})),
\end{align*}
where 
\begin{align}
    \label{lambda opt}
    \lambda^*(w) = \argmin_{\lambda \in \Delta_{|\mathcal{T}_0|-1}} \Tilde Q(\lambda,w).
\end{align}
Thus, the identification strategy invokes GMC at $(\lambda^*(w^{\mathsf{unif}}),w^*(\lambda^{\mathsf{unif}}))$.\footnote{Here the optimization problems defining $\lambda^*(w^{\mathsf{unif}})$ and $w^*(\lambda^{\mathsf{unif}})$ are equivalent to those proposed by \cite{Arkhangelsky/Athey/Hirshberg/Imbens/Wager:AER:21} without regularization.} 

Note that unless $\lambda^{\mathsf{unif}} = \lambda^*(w^{\mathsf{unif}})$, the SDID design is not reduced to the SCD design in terms of GMC. More specifically, we cannot motivate the weight $w^*(\lambda^{\mathsf{unif}})$ using SMC. This is because the within-group differencing used for the pre-treatment matching $(\lambda^{\mathsf{unif}})$ is different from that used for the post-treatment matching $(\lambda^*(w^{\mathsf{unif}}))$. Since the within-group differencing changes after the treatment, we cannot say that SDID extrapolates the weight from the pre-treatment fit to the post-treatment periods like SCD. In other words, SCD and SDID are distinct designs.

\begin{table}[t]
    \centering
    \caption{Generalized Matching Conditions of Causal Inference Methods}

    \small
    \begin{tabular}{ll}
        \hline \hline
        Research Designs & Generalized Matching Conditions\\
        \hline \\

        \textit{Using Weights Based on Group Sizes} & \\ \\

        RCT & GMC holds at $(0,1)$ \\ 
        Unconfoundedness & GMC holds at $(0,w^{\mathsf{C}})$ \\ 
        DID with Two Periods & GMC holds at $(1,1)$ \\ 
        DID with Two Periods and Discrete Covariates & GMC holds at $(1,w^{\mathsf{C}})$ \\ 
        DID with Staggered Adoption under PTA-I & GMC holds at $(\lambda^{\mathsf{DID}},w^{\mathsf{DID}})$ \\ 
        DID with Staggered Adoption under PTA-II & GMC holds at $(\lambda^{\mathsf{DID}},w)$, for all $w \in \Delta_{K-1}$ \\ \\
        \hline
        \\ 

        \textit{Using Weights Based on Pre-Treatment Fit} & \\ \\

        SC & SMC holds at $0$ $\Rightarrow$ GMC holds at $(0,w^*(0))$ \\ 
        SCD & SMC holds at $\lambda$ $\Rightarrow$ GMC holds at $(\lambda, w^*(\lambda))$\\ 
        SDID & GMC holds at $(\lambda^*(w^{\mathsf{unif}}),w^*(\lambda^{\mathsf{unif}}))$\\ 
         \\
        \hline \hline
    \end{tabular}\medskip

    \parbox{6.4in}{\footnotesize
        \textit{Notes:}  The table shows how GMC is used for various causal inference methods. Here, recall that $w_j^{\mathsf{DID}} = P\left\{G_i = j \mid G_i \in \mathcal{G}_{\mathsf{don}} \right\}$, $\lambda_{s}^{\mathsf{DID}} = 1\{s = T^* - 1\}$, $w_j^{\mathsf{unif}} = 1/K$ and $\lambda_{s}^{\mathsf{unif}} = 1/(T^*-1)$, and $w^*(\lambda)$ and $\lambda^*(w)$ are the solutions to the optimization problems in (\ref{w^*}) and (\ref{lambda opt}), respectively. The SCD method is based on the researcher-determined $\lambda$, for example, either $\lambda = \lambda^{\mathsf{DID}}$ or $\lambda = \lambda^{\mathsf{unif}}$. Note that while the DID method adopts the identified quantities $w^{\mathsf{DID}}$ and $\lambda^{\mathsf{DID}}$ directly, the SC, SDID and SCD methods need to invoke rank conditions to identify $w^*(\lambda)$ and $\lambda^*(w)$.}
        \label{table: GMC}
\end{table}

In summary, the major causal inference designs invoke different types of the GMC. Each type involves a choice of a differencing method ($\lambda$) and the groupwise matching weights ($w$). Table \ref{table: GMC} summarizes the comparison of the designs in terms of GMC.

\section{A Comparison Between DID and SCD}
\label{sec: Comparison}

\subsection{Extended Parallel Trend Assumption}

In this section, we compare the two approaches of DID and SCD in the setting of staggered adoption. To facilitate the comparison, we introduce an extended form of PTA. First, for each $t \in \mathcal{T}$, we define
\begin{align*}
    e_t^{\mathsf{DID}}(\lambda) &= \mathbf{E}[Y_{i,t}(0;\lambda) \mid G_i = 0] - \mathbf{E}[Y_{i,t}(0;\lambda) \mid G_i \in \mathcal{G}_{\mathsf{don}}], \text{ and }\\
    e_{j,t}^{\mathsf{DID}}(\lambda) &= \mathbf{E}[Y_{i,t}(0;\lambda) \mid G_i=0] -\mathbf{E}[Y_{i,t}(0;\lambda) \mid G_i=j], \text{ for } j \in \mathcal{G}_{\mathsf{don}}.
\end{align*}
The quantities $e_t^{\mathsf{DID}}(\lambda)$ and $e_{j,t}^{\mathsf{DID}}(\lambda)$ represent matching errors from matching the $\lambda$-differenced average potential untreated outcome for the target group with that from the donor groups. Then, we consider the two types of PTA involving within-group differencing $\lambda \in \Delta_{|\mathcal{T}_0|-1}$.\medskip

\textbf{PTA($\lambda$):} $e_{t}^{\mathsf{DID}}(\lambda) = 0$, for all $t \in \mathcal{T}_1$.

\textbf{PTA-U($\lambda$):} $e_{j,t}^{\mathsf{DID}}(\lambda) = 0$, for all $j \in \mathcal{G}_{\mathsf{don}}$, and for all $t \in \mathcal{T}_1$.\medskip

The following proposition shows their connection with the PTA used in the literature.\footnote{This result also suggests that when DID is used under PTA-I, the target parameter is overidentified using PTA($\lambda$), for any $\lambda$ satisfying $e_{T^*-1}^{\mathsf{DID}}(\lambda) = 0$. \cite{Chen/SantAnna/Xie:arXiv:25} proposed semiparametrically efficient estimation of ATT under overidentifying restrictions in PTA-U($\lambda$) with covariates.} 

\begin{proposition}
\label{prop: PTA connection}
(i) Suppose that $e_{T^*-1}^{\mathsf{DID}}(\lambda) = 0$ holds for some $\lambda \in \Delta_{|\mathcal{T}_0|-1}$. Then, PTA-I holds if and only if PTA($\lambda$) holds.

(ii) Suppose that for some $\lambda \in \Delta_{|\mathcal{T}_0|-1}$, $e_{j,T^*-1}^{\mathsf{DID}}(\lambda) = 0$ holds for each $j \in \mathcal{G}_{\mathsf{don}}$. Then, PTA-II holds if and only if PTA-U($\lambda$) holds.
\end{proposition}

Certainly, we have $e_{j,T^*-1}^{\mathsf{DID}}(\lambda^{\mathsf{DID}}) = 0$ for all $j  \in \mathcal{G}_{\mathsf{don}}$. Hence, we have 
\begin{align*}
    \text{PTA-I} \Leftrightarrow \text{PTA}(\lambda^{\mathsf{DID}}) \text{ and } \text{PTA-II} \Leftrightarrow \text{PTA-U}(\lambda^{\mathsf{DID}}).
\end{align*}
On the other hand, PTA($\lambda$) and PTA-U($\lambda$) allows other choices of $\lambda$. The comparison results below apply to such $\lambda$'s. From here on, we focus on PTA($\lambda$). 

\subsection{Regret Analysis}

In this section, we compare the research designs of DID and SCD in terms of regret in the staggered adoption setting in Example \ref{exmpl: staggered adoption}. Let $\mathcal{P}$ be the collection of the distributions of the variables under consideration. We fix a within-group differencing $\lambda \in \Delta_{|\mathcal{T}_0|-1}$ such that $e_{T^*-1}^{\mathsf{DID}}(\lambda) = 0$. To facilitate the comparison, we introduce the squared sum of matching errors (SSME): for $w \in \Delta_{K-1}$ and $P \in \mathcal{P}$,
\begin{align*}
    \mathsf{SSME}_{d,P}(w) = \frac{1}{|\mathcal{T}_d|} \sum_{t \in \mathcal{T}_d} e_t^2(\lambda,w), \quad d = 0,1.
\end{align*}
We make explicit its dependence on $P \in \mathcal{P}$ through the matching errors $e_t(\lambda,w)$. Then, we define the matching error in regret (MER) and the extrapolation error of MER, respectively,
\begin{align*}
        \mathsf{MER}_d(\mathbf{w}) &= \sup_{P \in \mathcal{P}}\left\{\mathsf{SSME}_{d,P}(w_P) - \inf_{w \in \Delta_{K-1}} \mathsf{SSME}_{d,P}(w) \right\}, \text{ and }\\ 
        \mathsf{\Delta MER}(\mathbf{w}) &= \mathsf{MER}_1(\mathbf{w}) - \mathsf{MER}_0(\mathbf{w}),
\end{align*}
where $\mathbf{w} = (w_P)_{P \in \mathcal{P}}$. The quantity $\mathsf{MER}_d(\mathbf{w})$ captures the matching error of the weights $w_P$, $P \in \mathcal{P}$, in the maximal regret form, whereas $\mathsf{\Delta MER}(\mathbf{w})$ measures the stability of the MER as we move from the pre-treatment regime to the post-treatment regime. We tend to have small $\mathsf{\Delta MER}(\mathbf{w})$ if the post-treatment matching errors are close to the pre-treatment matching errors. Thus, we call $\mathsf{\Delta MER}(\mathbf{w})$ \bi{the extrapolation error}, which essentially captures an error that arises from extrapolating the weight optimized for the pre-treatment data to the post-treatment outcomes. 

We compare SCD and DID in terms of MER. We define the population version of the weights by DID and SCD: for each $P \in \mathcal{P}$,\footnote{In fact, when the minimizer $w^*(\lambda)$ in (\ref{w^*}) is unique, the definition of $w_P^\mathsf{SCD}$ coincides with $w^*(\lambda)$. We assume uniqueness for this regret analysis.}
\begin{align*}
   w_P^\mathsf{SCD} = \argmin_{w \in \Delta_{K-1}} \mathsf{SSME}_{0,P}(w),
\end{align*}
and $w_P^{\mathsf{DID}} = [w_{1,P}^{\mathsf{DID}},...,w_{K,P}^{\mathsf{DID}}]'$ with 
\begin{align}
    \label{w^{DID}}
   w_{j,P}^\mathsf{DID} \eqdef \frac{\sum_{i \in N} P\{G_i = j\}}{\sum_{i \in N} P\{G_i \in \mathcal{G}_{\mathsf{don}}\}} . 
\end{align}
We define $\mathbf{w}^{\mathsf{DID}} = (w_P^{\mathsf{DID}})_{P \in \mathcal{P}}$ and $\mathbf{w}^{\mathsf{SCD}} = (w_P^{\mathsf{SCD}})_{P \in \mathcal{P}}$. The SCD invokes the Stable Matching Condition (SMC) and DID the Parallel Trend Assumption (PTA). These assumptions can be formulated in terms of the matching errors: 
\begin{align}
    \label{comparison}
    \text{PTA}(\lambda) &\Leftrightarrow \mathsf{SSME}_{1,P}(\mathbf{w}^{\mathsf{DID}}) = 0 \text{ for all } P \in \mathcal{P}.\\ \notag
    &\Rightarrow \mathsf{MER}_1(\mathbf{w}^{\mathsf{DID}}) = 0.\\ \notag 
    \text{SMC}(\lambda) &\Leftrightarrow \text{ For some } \mathbf{w}, \mathsf{MER}_0(\mathbf{w}) = 0 \text{ and }\mathsf{MER}_1(\mathbf{w}) = 0.\\ \notag
    &\Rightarrow \mathsf{\Delta MER}(\mathbf{w}^{\mathsf{SCD}}) = 0,
\end{align}
where SMC($\lambda$) denotes that SMC holds at $\lambda$. Thus the DID design fails if $\mathsf{MER}_1(\mathbf{w}^{\mathsf{DID}}) \ne 0$ whereas the SCD design fails if $\mathsf{\Delta MER}(\mathbf{w}^{\mathsf{SCD}}) \ne 0$. We will now formalize this complementarity in terms of maximal regret. 

For a concrete analysis, we define the sample analog estimator of $m_{j,t}$ and $\mu_{j,t}(\lambda)$ as follows:
\begin{align}
    \label{sample analogue est mu}
    \hat m_{j,t} = \frac{1}{n_j} \sum_{i \in N_j} Y_{i,t}, \text{ and } 
    \hat \mu_{j,t}(\lambda) &= \hat m_{j,t} - \sum_{s \in \mathcal{T}_0} \hat m_{j,s} \lambda_s,
\end{align}
where $n_j$ denotes the number of individuals in the sample belonging to group $j$. Then, given within-group differencing $\lambda$ and a choice of data-dependent weight $\hat w$, we can estimate $\theta_t^*$ as follows:
\begin{align}
    \label{DID est}
    \hat \theta_t(\lambda,\hat w) = \hat \mu_{0,t}(\lambda) - \sum_{j=1}^{K} \hat \mu_{j,t}(\lambda) \hat w_j.
\end{align}
Thus, the selection between the DID and SCD designs boils down to choosing the matching weight $\hat w$.

We define the weight $\hat w^{\mathsf{DID}}$ for the DID design as follows: $\hat w^{\mathsf{DID}} = [\hat w_1^{\mathsf{DID}},...,\hat w_K^{\mathsf{DID}}]'$ with $\hat w_{j}^{\mathsf{DID}}$ defined as the sample version of $w_{j,P}^{\mathsf{DID}}$ in (\ref{w^{DID}}):
\begin{align*}
    \hat w_j^{\mathsf{DID}} \eqdef \frac{\sum_{i \in N} 1\{G_i = j\}}{\sum_{i \in N} 1\{G_i \in \mathcal{G}_{\mathsf{don}}\}}, \text{ for } j \in \mathcal{G}_{\mathsf{don}}.
\end{align*}
The sample weight $\hat w_j^{\mathsf{DID}}$ represents the sample fraction of individuals in group $j$ relative to the total units in the donor pool. The DID design suggests estimating $\theta_t^*$ as $\hat \theta_t(\lambda,\hat w^{\mathsf{DID}})$. When $\lambda = \lambda^{\mathsf{DID}}$, this estimator can be viewed as a special case of an estimator proposed by \cite{Callaway/SantAnna:JoE:21} without covariates. When $K = 1$ and $T^*=2$ (i.e., the two periods and two groups setting), $\lambda$ and $\hat{w}^{\mathsf{DID}}$ are equal to one, so we obtain 
\begin{align*}
    \hat \theta_t(\lambda,\hat w^{\mathsf{DID}}) = \Delta \overline Y_1 - \Delta \overline Y_0,
\end{align*}
where $\Delta \overline Y_d$ denotes the first difference average outcomes for the group with treatment status $d = 0,1$. Thus, we can view $\hat \theta_t(\lambda,\hat w^{\mathsf{DID}})$ as an extension of the standard DID estimator of $\theta_t^*$ to the case with more than two periods and groups.

The SCD design uses the weight $\hat w^{\mathsf{SCD}}$ defined as 
\begin{align*}
    \hat w^{\mathsf{SCD}} \in \enspace \argmin_{w \in \Delta_{K-1}} \enspace \frac{1}{|\mathcal{T}_0|} \sum_{t \in \mathcal{T}_0} \left(\hat \mu_{0,t}(\lambda) - \sum_{j=1}^{K} \hat \mu_{j,t}(\lambda) w_j \right)^2.
\end{align*}
Hence, the weights $\hat w^{\mathsf{SCD}}$ are chosen to minimize the sample version of the pre-treatment SSME. The SCD design suggests estimating $\theta_t^*$ by
\begin{align}
    \label{SCD est}
        \hat \theta_t(\lambda,\hat w^{\mathsf{SCD}}) = \hat \mu_{0,t}(\lambda) - \sum_{j=1}^{K} \hat \mu_{j,t}(\lambda) \hat w_j^{\mathsf{SCD}}.
\end{align}
Notice that the SCD estimator, $\hat \theta_t(\lambda,\hat w^{\mathsf{SCD}})$, and the DID estimator, $\hat \theta_t(\lambda,\hat w^{\mathsf{DID}})$, differ only by the choice of the estimated weights for the donor pool.

To build a decision-theoretic comparison between different research designs, we introduce the average squared error loss from estimating $\theta_t^*$ by $\hat \theta_t(\lambda,\hat w)$: 
\begin{align*}
    \ell_1(\hat w) &= \frac{1}{|\mathcal{T}_1|} \sum_{t \in \mathcal{T}_1} \left(\theta_t^* - \hat \theta_t(\lambda,\hat w) \right)^2.
\end{align*}
We define the maximal regret associated with the choice of $\hat w$:
\begin{align*}
    \mathsf{MaxRegret}(\hat w) =  \sup_{P \in \mathcal{P}}\left\{\mathbf{E}_P\left[ \ell_1(\hat w) \right] - \inf_{\widetilde w \in \mathcal{D}} \mathbf{E}_P\left[ \ell_1(\widetilde w(Z)) \right]\right\},
\end{align*}
where $\mathcal{D}$ is the set of $\Delta_{K-1}$-valued functions that are measurable with respect to $Z$ and the random vector $Z$ represents the vector of all the observed random variables.\footnote{The expectation $\mathbf{E}_P\left[ \ell_1(\hat w) \right]$ is with respect to the distribution of both $\hat w$ and $\ell_1(\cdot)$.} We compare the DID and SCD designs in terms of their maximal regrets. 

We introduce assumptions used for the regret analysis. Let $Y_i^*(s) = (Y_{i,t}^*(s))_{t \in \mathcal{T}}$ and $Y_i^* = (Y_i^*(s))_{s \in \mathcal{T} \cup \{0\}}$. 

\begin{assumption}
    \label{assump: sampling design}
    The random vectors, $(Y_i^*,G_i)$, are independent across $i \in N$, under each $P \in \mathcal{P}$.
\end{assumption}

This assumption requires that the variables be independent across the cross-sectional units. This condition allows for arbitrary dependence between the potential outcomes across different treatment timing or the time periods. The framework allows for both the settings of repeated cross-sections and panel data. It also allows for factor models for the potential outcomes; we can simply take the factors to be constants.

\begin{assumption}
    \label{assump: moments} 
    There exist constants $\pi_0>0$, $\overline m \ge 1$, and $0< c \le 1$, such that for all $j \in \mathcal{G}$, $t \in \mathcal{T}$ and $s \in \mathcal{T} \cup \{0\}$, we have
    \begin{align}
        \label{bounds}
        \max_{1 \le i \le n} \sup_{P \in \mathcal{P}} \mathbf{E}_P[Y_{i,t}^4(s) \mid G_i = j] \le \overline m^4 \text{ and }
        \inf_{P \in \mathcal{P}} \frac{1}{n}\sum_{i \in N} P\{G_i = j \} \ge \pi_0,
    \end{align}
    and 
    \begin{align}
        \label{strong ID SCD}
        \inf_{P \in \mathcal{P}} \lambda_{\min}\left(\Gamma_P' \Gamma_P\right)  \ge c |\mathcal{T}_0|,
    \end{align}
    where $\Gamma_P$ is the $|\mathcal{T}_0| \times K$ matrix whose $(t,j)$-th entry is given by $\mu_{j,t}(\lambda)$.
\end{assumption}
The condition (\ref{bounds}) in Assumption \ref{assump: moments} requires the existence of uniform upper and lower bounds for the fourth moment of potential outcomes in each group and the probability of the group membership, respectively. The condition (\ref{strong ID SCD}) says that the time-paths of the within-group differenced mean outcomes are not linearly dependent. This condition requires that $|\mathcal{T}_0| \ge K$ and ensures that $w_P^{\mathsf{SCD}}$ is identified.

The theorem below presents the regret-comparison result between the DID and SCD designs.

\begin{theorem}
    \label{thm: regret analysis}
    Suppose that Assumptions \ref{assump: sampling design}-\ref{assump: moments} hold. For each $n \ge 1$, let 
    \begin{align*}
        \epsilon_{n} \eqdef \frac{(K+1) \overline m^4}{c}\left\{ \frac{1}{\pi_0 \sqrt{n}} + \frac{1}{\pi_0^2 n} + \exp\left( - \frac{\pi_0 n}{8} \right) \right\},
    \end{align*} 
    where $\overline m$, $\pi_0$ and $c$ are the constants in Assumption \ref{assump: moments}. Then, there exists a universal constant $C>0$ such that the following statements hold for all $n \ge 1$.\medskip
    
    (i) If $\mathsf{\Delta MER}(\mathbf{w}^{\mathsf{SCD}}) > \mathsf{MER}_1(\mathbf{w}^{\mathsf{DID}}) + C \epsilon_{n}$, then, 
    \begin{align*}
        \mathsf{MaxRegret}(\hat w^{\mathsf{SCD}}) > \mathsf{MaxRegret}(\hat w^{\mathsf{DID}}).
    \end{align*}

    (ii) If $\mathsf{\Delta MER}(\mathbf{w}^{\mathsf{SCD}}) < \mathsf{MER}_1(\mathbf{w}^{\mathsf{DID}}) - C \epsilon_{n}$, then, 
    \begin{align*}
        \mathsf{MaxRegret}(\hat w^{\mathsf{SCD}}) < \mathsf{MaxRegret}(\hat w^{\mathsf{DID}}).
    \end{align*}
\end{theorem}

The result shows that the DID design regret-dominates the SCD design if and only if the extrapolation error of the SCD design dominates the post-treatment MER of the DID design up to a term that vanishes at the parametric rate $\sqrt{n}$.

The DID design specifies the matching weights to be the group size-based weights, and hence does not need to invoke extrapolation of weights from the pre-treatment fit. On the other hand, the SCD design obtains the weights that exhibit a best pre-treatment fit, and extrapolates the weights to the post-treatment periods. The comparison shows when the DID design or the SCD design is appropriate or not. The DID design is not appropriate in a setting where it is doubtful that the relevance of each group in matching is proportional to the size of the group, whereas the SCD design is not appropriate if the relevance of the groups in matching is not stable before and after the treatment.

\subsection{Equivalence of DID and SCD}

One might wonder when the DID and SCD designs are equivalent in terms of GMC. The analysis in (\ref{comparison}) gives an answer. Let us introduce pre-treatment PTA with within-group differencing $\lambda$ as follows:\medskip

\textbf{Pre-treatment PTA($\lambda$):} $e_{t}^{\mathsf{DID}}(\lambda) = 0$, for all $t \in \mathcal{T}_0$.\medskip

Then, following the same arguments in the proof of Proposition \ref{prop: PTA connection}, we can show that the pre-treatment PTA($\lambda$) is equivalent to the following:\medskip

\textbf{Pre-treatment PTA-I: } $\mathbf{E}[\Delta Y_{i,t}(0) \mid G_i = 0] = \mathbf{E}[\Delta Y_{i,t}(0) \mid G_i \in \mathcal{G}_{\mathsf{don}}]$, for all $t \in \mathcal{T}_0 \setminus \{1\}$.\medskip

Thus, Pre-treatment PTA-I refers to the parallel trend assumption that holds for the pre-treatment periods. Now, suppose that both the PTA($\lambda$) and the pre-treatment PTA($\lambda$) hold. This implies that 
\begin{align*}
    \mathsf{MER}_1(\mathbf{w}^{\mathsf{DID}}) = \mathsf{MER}_0(\mathbf{w}^{\mathsf{DID}}) = 0.
\end{align*}
On the other hand, by the definition of $\mathbf{w}^{\mathsf{SCD}}$, we have 
\begin{align*}
    \mathsf{MER}_0(\mathbf{w}^{\mathsf{SCD}}) = 0.
\end{align*}
Since there is a unique $\mathbf{w}$ such that $\mathsf{MER}_0(\mathbf{w}) = 0$ by (\ref{strong ID SCD}) in Assumption \ref{assump: moments}, we must have $\mathbf{w}^{\mathsf{SCD}} = \mathbf{w}^{\mathsf{DID}}$. We formalize this into the following proposition. 

\begin{proposition}
    \label{pr: DID equivalence}
    Suppose that Assumptions \ref{assump: sampling design}-\ref{assump: moments} hold, and the PTA($\lambda$) and the pre-treatment PTA($\lambda$) hold. Then, 
    \begin{align*}
        \mathbf{w}^\mathsf{SCD} = \mathbf{w}^{\mathsf{DID}}.
    \end{align*}
    Hence, the DID and SCD designs are equivalent in terms of GMC. 
\end{proposition}

The result shows that when we use the same differencing method for both DID and SCD designs, and the PTA holds at all periods, the two designs are equivalent in terms of GMC. For the identification of the ATT, we do not require the pre-treatment PTA. However, it is a common practice to perform a pre-trend PTA test to assess the plausibility of the post-treatment PTA. This procedure is valid only if PTA implies the pre-treatment PTA. Then, the proposition says that under this implication, if the DID identifies the ATT through the post-treatment PTA, this means that the weights used by the DID are exactly the same as the weights chosen by the SCD. Hence, both designs are equivalent in terms of GMC.

Now, when the post-treatment PTA fails, the equivalence between the DID and the SCD breaks down, and the SCD can be an alternative to the DID design. The GMC for the SCD emerges as an alternative identifying assumption replacing PTA.

\section{Inference for Synthetic Control with Differencing}
\label{sec: SCD}

We saw that the SCD can serve as an alternative to DID when the parallel trend assumption fails. To the best of our knowledge, the estimation and inference methods for SCD in our data structure have not been developed. Thus, we present the methods here together with their asymptotic properties. The proofs of the results are found in the Supplemental Note.

\subsection{The Sampling Process and Estimation}

\subsubsection{The Sampling Process}

In this section, we explain the sampling process that links the population objects to the sample. As for the population objects, we first assume that the random vectors, $(Y_i^*,G_i)$, are i.i.d.\ across $i \in N$, under each $P \in \mathcal{P}$. Let $P_{j,t}$ be the conditional distribution of $Y_{i,t}$ given $G_i = j$, and let $p_j = P\{G_i = j\}$.

For each $t \in \mathcal{T}$, we first draw $G_{i,t} \in \mathcal{G}$, i.i.d.\ across $i \in N_t$, with probability $P\{G_{i,t} = j\}$ equal to $p_j$ for each $j \in \mathcal{G}$. Then, we draw $Y_{i,t}$, $i \in N_t$, i.i.d.\ from the conditional distribution $P_{j,t}$. By the sampling process, for each $t \in \mathcal{T}$, and $i \in N_t$, we have 
\begin{align*}
    m_{j,t} = \mathbf{E}[Y_{i,t} \mid G_{i,t} = j].
\end{align*}
We let $N = \bigcup_{t=1}^T N_t$ and $n = |N|$. We also define $N_{j,t} = \{i \in N: G_{i,t} = j\}$ and $n_{j,t} = |N_{j,t}|$. 

This sampling process accommodates the empirical setting where the size of cross-sectional units varies over time. It also accommodates both balanced or unbalanced panel settings and repeated cross-sections. In the balanced panel setting, we have $N_t = N$ for all $t \in \mathcal{T}$, and assume that the random vectors 
\begin{align*}
    [(Y_{i,1},G_{i,1}),...,(Y_{i,T},G_{i,T})]
\end{align*}
are i.i.d.\ across $i \in N$, whereas in the repeated cross-sections setting, we assume that $(Y_{i,t},G_{i,t})$ are i.i.d.\ across $i \in N_t$ and independent across $t \in \mathcal{T}$.

For estimation and inference, we fix within-group differencing $\lambda \in \Delta_{|\mathcal{T}_0|-1}$ and assume that we are under SMC at $\lambda$ so that we have 
\begin{align}
    \label{eq003}
    \mu_{0,t}(\lambda) = \sum_{j=1}^K \mu_{j,t}(\lambda) w_{j},
\end{align}
for all $t \in \mathcal{T}$, for some $w \in \Delta_{K-1}$. Thus, SCD recovers this weight using (\ref{w^*}).\footnote{Note that SCD extrapolates the weight obtained from the pre-treatment matching to the post-treatment periods. It appears strange that the weight that did not give a perfect pre-treatment matching now achieves a perfect post-treatment matching. Hence, we assume that the weight gives a perfect matching on both pre- and post-treatment periods, i.e., SMC at $\lambda$.}

\subsubsection{Estimation}

First, we consider a setting where $\boldsymbol{\theta}^*$ is identified. Since $\boldsymbol{\theta}^*$ is equal to $\boldsymbol{\theta}(\lambda,w^*(\lambda))$, the identification of $\boldsymbol{\theta}^*$ boils down to that of $w^*(\lambda)$. We simply write 
\begin{align*}
    \hat \mu_{j,t} = \hat \mu_{j,t}(\lambda),
\end{align*}
where $\hat \mu_{j,t}(\lambda)$ is defined in (\ref{sample analogue est mu}). We propose the following estimator of the weight $w^*(\lambda)$:
\begin{align*}
	\hat w(\lambda) = \argmin_{w \in \Delta_{K-1}} \enspace \hat Q(\lambda,w),
\end{align*}
where, with $\boldsymbol{\hat \mu}_t = [\hat \mu_{1,t},...,\hat \mu_{K,t}]'$, 
\begin{align*}
	\hat Q(\lambda,w) = \sum_{t =1}^{T^*-1} \left( \hat \mu_{0,t} - \boldsymbol{\hat \mu}_t' w \right)^2.
\end{align*}
Lastly, we consider the following estimator for the target parameter $\theta_{t}^*$: 
\begin{align}\label{par:theta_hat}
    \hat \theta_{t}(\hat w(\lambda)) = \hat \mu_{0,t} - \boldsymbol{\hat \mu}_t' \hat w(\lambda).
\end{align}
Then, we can show that our estimators for the weight $w^*(\lambda)$ defined in (\ref{w^*}) and the target parameter $\theta_{t}^*$ are $\sqrt{n}$-consistent. 

\begin{theorem}\label{theorem:3.1}
	Suppose that Assumptions \ref{assump:basic}, \ref{assump: sampling design} and \ref{assump: moments} hold. Then, for all $t \in \mathcal{T}_1$, as $n \rightarrow \infty$, 
    \begin{align*}
        \hat w(\lambda) = w^*(\lambda) + O_P(n^{-1/2}) \text{ and } \hat \theta_{t}(\hat w(\lambda)) = \theta_{t}(\lambda,w^*(\lambda)) + O_P(n^{-1/2}).
    \end{align*}
\end{theorem}

\subsection{Inference}

We consider statistical inference on vector $\boldsymbol{\theta}^*$ without assuming its point-identification, adapting the proposal of \cite{Canen/Song:arXiv:25} to our setting.

First, we construct a confidence set for $w^*(\lambda)$. Define
\begin{align*}
    \hat H = \frac{1}{T^*-1}\sum_{t =1}^{T^*-1} \hat H_t \text{ and }  \boldsymbol{\hat h} = \frac{1}{T^*-1}\sum_{t =1}^{T^*-1} \boldsymbol{\hat h}_t,
\end{align*}
where $\hat H_t = \boldsymbol{\hat \mu}_t \boldsymbol{\hat \mu}_t' \text{ and } \boldsymbol{\hat h}_t = \hat \mu_{0,t} \boldsymbol{\hat \mu}_{t}.$ Let 
\begin{align*}
    \hat \varphi(w) = \hat H w - \boldsymbol{\hat h} \text{ and } \varphi_P(w) = H w - \boldsymbol{h}.
\end{align*}
Let $B = [\mathbf{1}/\sqrt{K},B_2]$ be the $K \times K$ orthogonal matrix $B$ such that $B' B = I_K$ and $B_2' B_2 = I_{K-1}$, where $\mathbf{1}$ denotes the $K$-dimensional vector of ones.\footnote{The matrix $B_2$ can be computed as follows. First, we obtain a spectral decomposition: $I_{K} - \mathbf{1} \mathbf{1}'/K = U D U'$, where $\mathbf{1}$ denotes the $K$-dimensional vector of ones. From this, we set $B_2$ to be the $K \times (K-1)$ matrix after removing the eigenvector from $U$ that corresponds to the zero diagonal element of $D$.} Note that $B_2$ is a $K \times (K-1)$ matrix. First, note that for each $i \in N$ and $t \in \mathcal{T}_0$,
\begin{align*}
    \sqrt{n}(\hat \mu_{j,t} - \mu_{j,t}) = \frac{1}{\sqrt{n}} \sum_{i \in N} \psi_{ij,t} + o_P(1), \text{ as } n \to \infty,
\end{align*}
where $\psi_{ij,t} = \psi_{ij,t}^* - \sum_{s \in \mathcal{T}_0} \lambda_s \psi_{ij,s}^*$, with\footnote{In the case of a balanced panel setting with $N = N_t$ and $G_{i,t} = G_i$ for all $t \in \mathcal{T}$ and $i \in N$, we have $\psi_{ij,t} = 1\{G_{i,t} = j\}(y_{i,t} - \mu_{j,t})/p_j$, with $y_{i,t} = Y_{i,t} - \sum_{s \in \mathcal{T}_0} \lambda_s Y_{i,s}$.}
\begin{align*}
   \psi_{ij,t}^* = \frac{n}{n_t}\frac{1\{G_{i,t}= j\}}{p_{j}} (Y_{i,t} - m_{j,t}).
\end{align*}
For notational brevity, we define 
\begin{align*}
    \boldsymbol{z}_{ij} = \frac{1}{T^*-1}\sum_{t=1}^{T^*-1} \boldsymbol{\mu}_t \psi_{ij,t} \text{ and } \boldsymbol{z}_{i}(w) = \sum_{j=1}^K w_j \boldsymbol{z}_{ij} - \boldsymbol{z}_{i0}.
\end{align*}
 Using this, we find that 
\begin{align*}
    \sqrt{n} B_2'(\hat \varphi(w) - \varphi_P(w)) \to_d N(0,V_P(w)),
\end{align*}
where $V_P(w) = B_2'\text{Var}_P \left( \boldsymbol{z}_{i}(w) \right)B_2.$ Let us estimate $V_P(w)$ by $\hat V(w)$ as follows: with $\boldsymbol{\hat z}_{ij} = \frac{1}{T^*-1}\sum_{t=1}^{T^*-1} \boldsymbol{\hat \mu}_t \hat \psi_{ij,t}$ and $\boldsymbol{\hat z}_{i}(w) = \sum_{j=1}^K w_j \boldsymbol{\hat z}_{ij} - \boldsymbol{\hat z}_{i0}$,
\begin{align}
    \label{hat V(w)}
	\hat V(w) = B_2' \left(\frac{1}{n}\sum_{i \in N} \boldsymbol{\hat z}_{i}(w) \boldsymbol{\hat z}_{i}'(w) \right) B_2,
\end{align} 
where $\hat \psi_{ij,t}$ is the same as $\psi_{ij,t}$ except that $p_{j}$ and $m_{j,t}$ are replaced by $\hat p_{j,t} = n_{j,t}/n_t$ and $\hat m_{j,t} = (1/{n_{j,t}}) \sum_{i \in N_{j,t}} Y_{i,t}$.

When the sample is repeated cross-sections, the observations are independent across time. In this case, we can obtain sharper inference by modifying $\hat V(w)$ as follows: 
\begin{align}
    \label{hat V(w)2}
	\hat V_{RC}(w) &= \frac{1}{T^*-1} \sum_{t=1}^{T^*-1} \left\{\frac{1}{n}\sum_{i \in N} \left(\sum_{j=1}^K w_j  \hat \psi_{ij,t}^* - \hat \psi_{i0,t}^*\right)^2  \right.\\ \notag
    &\quad \quad \quad \quad \quad \quad \left.\times B_2'\left( \frac{\boldsymbol{\hat \mu}_t \boldsymbol{\hat \mu}_t'}{T^*-1} - \lambda_t \left(\boldsymbol{\bar \mu} \boldsymbol{\hat \mu}_t' + \boldsymbol{\hat \mu}_t \boldsymbol{\bar \mu}'\right) + (T^*-1) \lambda_t^2 \boldsymbol{\bar \mu} \boldsymbol{\bar \mu}' \right)B_2 \right\},
\end{align}
where 
\begin{align*}
   \hat \psi_{ij,t}^* = \frac{n}{n_t}\frac{1\{G_{i,t}= j\}}{\hat p_{j,t}} (Y_{i,t} - \hat m_{j,t}) \text{ and } \boldsymbol{\bar \mu} = \frac{1}{T^* - 1} \sum_{t=1}^{T^*-1} \boldsymbol{\hat \mu}_t.
\end{align*}
For each $w \in \Delta_{K-1}$, define\footnote{Due to the constraint, we have $\hat r(w) = 0$ if all entries of $w$ are positive. Hence, we perform the numerical optimization only if some of the entries of $w$ are zeros.} 
\begin{align}
    \label{r(w)}
    \hat r(w) = \argmin_{r} \enspace (\hat H w - \boldsymbol{\hat h} - r)'B_2 \hat V^{-1}(w) B_2'(\hat H w - \boldsymbol{\hat h} - r),
\end{align}
where the minimization over $r$ is under the constraint that $w'r = 0$ and $r \ge 0$. We define 
\begin{align*}
    \hat d(w) = \left| \left\{j=1,...,K: \hat \gamma_j(w) = 0 \text{ and } w_j = 0 \right\} \right|,
\end{align*}
where 
\begin{align*}
    \hat \gamma(w) = B_2\hat V^{-1}(w)B_2'(\hat H w - \boldsymbol{\hat h} - \hat r(w)).
\end{align*}
We set $\hat c_{1-\kappa,\mathsf{bf}}(w)$ to be the $(1-\kappa)$-th quantile of the $\chi^2$ distribution with degrees of freedom equal to 
\begin{align}
    \label{k(w)}
    \hat k(w) := \max\{K-1 - \hat d(w),1\}.
\end{align}
Then, the confidence set for $w^*(\lambda)$ is given by 
\begin{align*}
   \tilde C_{1-\kappa} = \{w \in \Delta_{K-1}: T_{\mathsf{bf}}(w) \le \hat c_{1-\kappa,\mathsf{bf}}(w)\}, 
\end{align*}
where \begin{align*}
    T_{\mathsf{bf}}(w) = n(\hat H w - \boldsymbol{\hat h} - \hat r(w))'B_2 \hat V^{-1}(w) B_2'(\hat H w - \boldsymbol{\hat h} - \hat r(w)).
\end{align*}

\begin{algorithm}[t]
\caption{Computing Confidence Intervals for $\theta_t^*$: Bonferroni Method}
\label{alg:confidence_intervals 1}
\renewcommand{\baselinestretch}{1.5}
\begin{algorithmic}[1]
\REQUIRE Consistent estimator $\hat{w}(\lambda)$ of $w^*(\lambda)$, $\hat \sigma_t(\hat w(\lambda))$ and  $\hat V(\hat{w}(\lambda))$.
\STATE Draw $w_1, \ldots, w_R \in \Delta_{K-1}$ i.i.d.\ from a distribution that has a full support on $\Delta_{K-1}$.
\STATE Compute $T_{\mathsf{bf}}(w_r)$ and $\hat{c}_{1-\kappa,{\mathsf{bf}}}(w_r)$ with $\hat V(w)$ replaced by $\hat V(\hat{w}(\lambda))$ for each $r=1,...,R$.
\STATE Let 
\begin{align*}
    c_{U,R} &= \max_{1 \le r \le R: T_{\mathsf{bf}}(w_r) \le \hat{c}_{1-\kappa,{\mathsf{bf}}}(w_r)} \left\{\hat \theta_t (w_r) + \frac{z_{1-\beta(\alpha,\kappa)} \hat \sigma_t(\hat w(\lambda))}{\sqrt{n}} \right\} \text{ and }\\
    c_{L,R} &= \min_{1 \le r \le R: T_{\mathsf{bf}}(w_r) \le \hat{c}_{1-\kappa,{\mathsf{bf}}}(w_r)} \left\{\hat \theta_t (w_r) - \frac{z_{1-\beta(\alpha,\kappa)} \hat \sigma_t(\hat w(\lambda))}{\sqrt{n}} \right\},
\end{align*}
where $\beta(\alpha,\kappa) = (\alpha - \kappa)/2$.
\ENSURE Confidence interval for $\theta_t^*$:
\begin{align*}
    C_{1-\alpha,R} = [c_{L,R},c_{U,R}].
\end{align*}
\end{algorithmic}
\renewcommand{\baselinestretch}{1.0}
\begin{minipage}{\textwidth}
\vspace{0.4cm}
\footnotesize \textit{Notes:}  The algorithm assumes that the weight $w^*(\lambda)$ is point-identified. When it is partially identified, we can modify the above algorithm by replacing $\hat V(\hat w(\lambda))$ and $\hat \sigma(\hat w(\lambda))$ in Steps 2 and 3 by $\hat V(w_r)$ and $\hat \sigma(w_r)$.
\end{minipage}
\vspace{0.4cm}
\end{algorithm}

Let $\tilde C_{1-\kappa}$ be the confidence set for $w^*(\lambda)$. We construct a confidence interval for $\theta_t^*$. Note that 
\begin{align*}
  \sqrt{n}(\hat \theta_t(w^*(\lambda)) - \theta_t^*) = \frac{1}{\sqrt{n}}\sum_{i \in N} \psi_{it,\theta}(w^*(\lambda)) + o_P(1),
\end{align*}
where 
\begin{align*}
    \psi_{it,\theta}(w^*(\lambda)) = \psi_{i0,t} - \sum_{j=1}^{K} \psi_{ij,t} w_{j}^*(\lambda).
\end{align*}
Define 
\begin{align*}
    \hat \sigma_t^2(w) = \frac{1}{n}\sum_{i \in N} \hat \psi_{it,\theta}^2(w), 
\end{align*}
where $\hat \psi_{it,\theta}(w) = \hat \psi_{i0,t} - \sum_{j=1}^{K} \hat \psi_{ij,t} w_j.$ Then, the confidence interval for $\theta_t^*$ is given as follows: with $\kappa \in (0,\alpha)$, (say, $\kappa = 0.005$)
\begin{align*}
    C_{1-\alpha}^\mathsf{bf} = \left\{ \tau \in \mathbf{R}: \inf_{w \in \tilde C_{1-\kappa} }\left| \frac{\sqrt{n}(\hat \theta_t(w) - \tau)}{\hat \sigma_t(w)} \right| \le z_{1 - \beta(\alpha,\kappa)} \right\},
\end{align*}
where $\beta(\alpha,\kappa) = (\alpha - \kappa)/2$.

The computation of a confidence interval for $\theta_t^*$ involves inverting a test for the weight vector. For the case of point-identified $w^*(\lambda)$, we present an algorithm that computes the convex hull of $C_{1-\alpha}$ directly without constructing $\tilde C_{1-\kappa}$ first. See Algorithm \ref{alg:confidence_intervals 1}. Computational experiments in Section \ref{subsec: Monte Carlo} below demonstrate that the algorithm computes the confidence set efficiently in practical data dimensions ($n_t = 14,000 \sim 130,000$, $T = 84$ and $K = 46$).

\subsection{Asymptotic Validity}

Let us introduce assumptions we use for the uniform asymptotic validity of the confidence set for $\theta_t^*$, without requiring the point-identification of $w^*(\lambda)$:

\begin{assumption} 
    \label{assump: nonsingularity}
    There exist constants $C,c>0$ such that for all $n \ge 1$,
\begin{align*}
    \sup_{P \in \mathcal{P}} \sup_{w \in \Delta_{K-1}}\|V_P(w)\| < C \text{ and } \inf_{P \in \mathcal{P}} \inf_{w \in \Delta_{K-1}} \lambda_{\min}(V_P(w)) > c.
\end{align*} 
\end{assumption}
Assumption \ref{assump: nonsingularity} requires that the asymptotic variance $V_P(w)$ is well behaved uniformly over $P \in \mathcal{P}$ and $w \in \Delta_{K-1}$: it should be both bounded and non-singular.

Under these conditions, we obtain the following validity result.
\begin{theorem} 
    \label{thm: asymptotic validity2}
    Suppose that Assumptions \ref{assump:basic}, \ref{assump: sampling design}, \ref{assump: nonsingularity}, (\ref{bounds}) in Assumption \ref{assump: moments}, and SMC holds at $\lambda$. Then, for all $t \in \mathcal{T}_1$, as $n \rightarrow \infty$, we have
\begin{align*}
\liminf\limits_{n\rightarrow\infty} \inf_{P \in \mathcal{P}} P\left\{\theta_t^* \in C_{1-\alpha}\right\} \geq 1-\alpha.
\end{align*}
\end{theorem}
The proofs are found in the Supplemental Note.

\subsection{Monte Carlo Simulations}
\label{subsec: Monte Carlo}

In this subsection we study the finite sample properties of our estimator of the target parameter. Our focus is on comparing SCD and DID and examining their complementarities. We consider a short panel setting of length $T$ where individual data is available and analyze the simple case of one treated group and $K$ untreated groups. Group $0$ becomes treated at time $T^*$, and the remaining groups $\{1,\ldots,K\}$ form the donor pool $\mathcal{G}_{\mathsf{don}}$. We set $T^*=T \in \{50, 100\}$ and $K \in \{10, 30\}$.\footnote{Here the length of the post-treatment window is equal to one.} We compare the performance of our SCD estimator with the standard DID estimator in terms of the mean absolute deviation (MAD), the empirical coverage probability (ECP), and the average length of the 95\% confidence intervals. We consider a sample size of $n \in \{1500,3000\}$ and set the number of Monte Carlo simulations to $1,000$.

When comparing SCD and DID approaches, we implement the DID estimator of \cite{Callaway/SantAnna:JoE:21} without covariates. In the study, we consider three different scenarios: one (Scenario A) in which PTA holds and there are parallel pre-trends, a second one (Scenario B) in which PTA is violated but SMC holds throughout all time periods, and a last one (Scenario C) where PTA holds, SMC holds, but the weights for donors cannot be recovered from pre-treatment data. Thus, in Scenario A, both DID and SCD produce consistent and asymptotically normal estimators of the treatment effect. However, in Scenario B, while SCD works, DID is not consistent. The opposite occurs in Scenario C.   

We now describe the data generating process used in the simulations. First, for the baseline setting, we define the probability of an individual belonging to group $j \in \mathcal{G}=\{0,1,...,K\} $ as simply $1/(K+1)$, so that $G_i$ is drawn i.i.d.\ from the uniform distribution over $\mathcal{G}$ with probability $1/(K+1)$. As for the generation of potential outcomes, we adopt a factor model:
\begin{align}
\label{MC:factor_model1}    
Y_{i,t}(0) = \Lambda_i' F_t + \varepsilon_{i,t},
\end{align}
where, conditional on $G_i$, $\Lambda_i \sim N(m_{G_i}, I_3)$.\footnote{We use $\mathbf{1}_m$ and $I_m$ to denote the $m$-dimensional vector of ones and the $m \times m$ identity matrix, respectively.} Here $m_j$ denotes the population mean of factor loadings in group $j$, whose components are drawn independently from a normal distribution with mean zero and variance $2.5^2$. In addition, time factors and idiosyncratic shocks are generated as $F_t \sim N(0.02 \sqrt{t} \cdot \mathbf{1}_3, 0.5^2 \cdot I_3)$ and $\varepsilon_{i,t} \sim N(0,1)$. Lastly, we set treated potential outcomes for individuals in group $0$ as
\begin{align}
Y_{i,t}(T^*) = Y_{i,t}(0) + \tau_{i,t}(T^*),
\end{align}
where $\tau_{i,t}(T^*) = \eta_{i,t}^2$, and each $\eta_{i,t}$ is drawn from a normal distribution with mean zero and variance $0.1$. This setup implies that $\theta_{T^*}^* = 0.1.$ In other words, the average treatment effect for individuals in group $0$ is equal to the variance of the random variable $\eta_{i,t}$.

\input{a_big_table_finalOG_B}

We select different combinations of the differencing parameter ($\lambda \in \Delta_{|\mathcal{T}_0|-1}$) and the population mean of individual factor loadings in the treated group ($m_{0} \in \mathbf{R}^3$) across scenarios. In Scenario A, we choose
\begin{align*}
    \lambda = \lambda^{\mathsf{DID}} \text{ and } m_{0} = \sum_{j=1}^K m_{j} w_j^{\mathsf{DID}},
\end{align*}
where $w^{\mathsf{DID}} = [1/K, \ldots, 1/K]' \in \mathbf{R}^K$ by the simulation design. By choosing these values, we guarantee that PTA is satisfied, parallel pre-trends are present, and GMC holds in both the pre- and post-treatment periods at ($\lambda^{\mathsf{DID}}, w^{\mathsf{DID}}$). In Scenario B, we let
\begin{align*}
    \lambda = \lambda^{\mathsf{unif}} \text{ and } m_{0} = \sum_{j=1}^K m_{j} w_j^{\mathsf{SCD}},
\end{align*}
where $w^{\mathsf{SCD}} = [0, \ldots, 0, 0.1, 0.9]' \in \mathbf{R}^K$ is a vector whose first $K-2$ entries are zero. In this case, PTA is violated since $w^{\mathsf{SCD}} \neq w^{\mathsf{DID}}$, but GMC still holds at ($\lambda^{\mathsf{unif}}, w^{\mathsf{SCD}}$). Lastly, we consider a Scenario C where PTA and GMC hold at ($\lambda^{\mathsf{DID}}, w^{\mathsf{DID}}$) for the post-treatment period, but the SCD approach is unable to recover $w^{\mathsf{DID}}$ using pre-treatment data. More precisely, we allow for time-varying factor loadings for individuals in the treated group as follows
\begin{align*}
    \Lambda_{i,t} = \tilde{\Lambda}_{i} 1\{t \leq T^*-2\} + \Lambda_{i} 1\{t \geq T^*-1\},
\end{align*}
where $\Lambda_{i}$ is defined as in Scenario A, but, conditional on $G_i$, $\tilde{\Lambda}_{i} \sim N(\tilde{m}_{G_i}, I_3)$, and
\begin{align*}
    \tilde{m}_0 = \sum_{j=1}^K m_j w_j^{\mathsf{OUT}},
\end{align*}
with $w^{\mathsf{OUT}} = [0, \ldots, 0, -0.3, 0.4, 0.9]' \in \mathbf{R}^K$. In this case, we allow individual factor loadings to be drawn from different distributions between pre- and post-treatment periods so that weights for control groups cannot be estimated by SCD using pre-treatment data.
 
The simulation results from this baseline setting are reported in Table \ref{table:simulation_results_newOG}. Inference results for SCD are based on the Bonferroni approach detailed in Algorithm \ref{alg:confidence_intervals 1}. In all scenarios, when the number of donor groups $K$ increases (keeping fixed $T$ and $n$), the accuracy of the estimators in terms of MAD deteriorates for both SCD and DID. In particular, in Scenario A, we have $w^{\mathsf{SCD}} = w^{\mathsf{DID}}$, so both designs generate a consistent estimator of $\theta_{T^*}^*$ and the confidence intervals are asymptotically valid as $n \to \infty$. As the number of groups increases, the estimation error of the weights accumulates. This explains the performance deterioration as $K$ increases from 10 to 30. When $T$ increases, the performance of SCD in terms of MAD remains similar, while confidence intervals become narrower.\footnote{Since the simulation includes only one post-treatment period, an increase in $T$ reflects an increase in the number of pre-treatment periods.} This pattern arises primarily because our setting is a panel design. In a repeated cross-section setting, the observations are independent across time and the accuracy would have increased as $T$ increased. The empirical coverage probability of DID and SCD shows conservativeness in this scenario, whereas CS-DID produces confidence bands that are, on average, 24\% narrower than those obtained with SCD.

In Scenario B, our simulation design is chosen so that PTA fails but SMC holds at $\lambda^{\mathsf{unif}}$. As expected, in this case we find that SCD outperforms DID in terms of MAD, and it also exhibits some conservativeness as in Scenario A. In Scenario C, we consider an opposite setting, that is, PTA holds but the stability of the weights fails. In this case, only DID provides a consistent estimate of $\theta_{T^*}^*$. As a result, DID displays a smaller MAD than SCD. The performance of SCD in Scenario C still appears better than that of DID in Scenario B. However, we believe this difference is largely driven by the simulation design.

In Tables \ref{table:simulation_results largep}–\ref{table:simulation_results_tloading} of Appendix \ref{sec: AdditionalResults}, we conduct three robustness checks for our simulation results. First, we depart from the equal-size group setting and instead consider a donor pool in which the last group is relatively larger than the others. Specifically, we set the group membership probabilities to
$p = [0.8/K, \ldots, 0.8/K, 0.2]' \in \mathbf{R}^{K+1}$ for $K=10$, and
$p = [0.925/K, \ldots, 0.925/K, 0.075]' \in \mathbf{R}^{K+1}$ for $K=30$. Secondly, we allow each component of time factors $F_t$ to have a different population mean by drawing it from a multivariate normal distribution with mean $\xi \sqrt{t}$ and covariance matrix $0.5^2 I_3$, where $\xi = [0.01, 0.02, 0.04]'$. Lastly, we use a multivariate $t$ distribution with 5 degrees of freedom, mean $m_{G_i}$, and scale matrix $I_3$, as the conditional distribution of individual factor loadings $\Lambda_i$. Across all these additional exercises, the main findings reported in Table \ref{table:simulation_results_newOG} continue to hold.

Overall, our simulation results highlight the complementarity between SCD and DID. Each method performs well when its identifying assumptions hold and deteriorates when they fail. In particular, SCD provides a reliable alternative when PTA is violated and performs as well as DID when both sets of identifying assumptions are satisfied.

\subsection{Computation Time}

\begin{figure}[tbp]
    \centering
    \caption{Computation Time.}
    \label{fig:SCD RC}
    \vspace{0.4cm}

    \begin{subfigure}{0.48\textwidth}
        \centering
        \includegraphics[width=\textwidth]{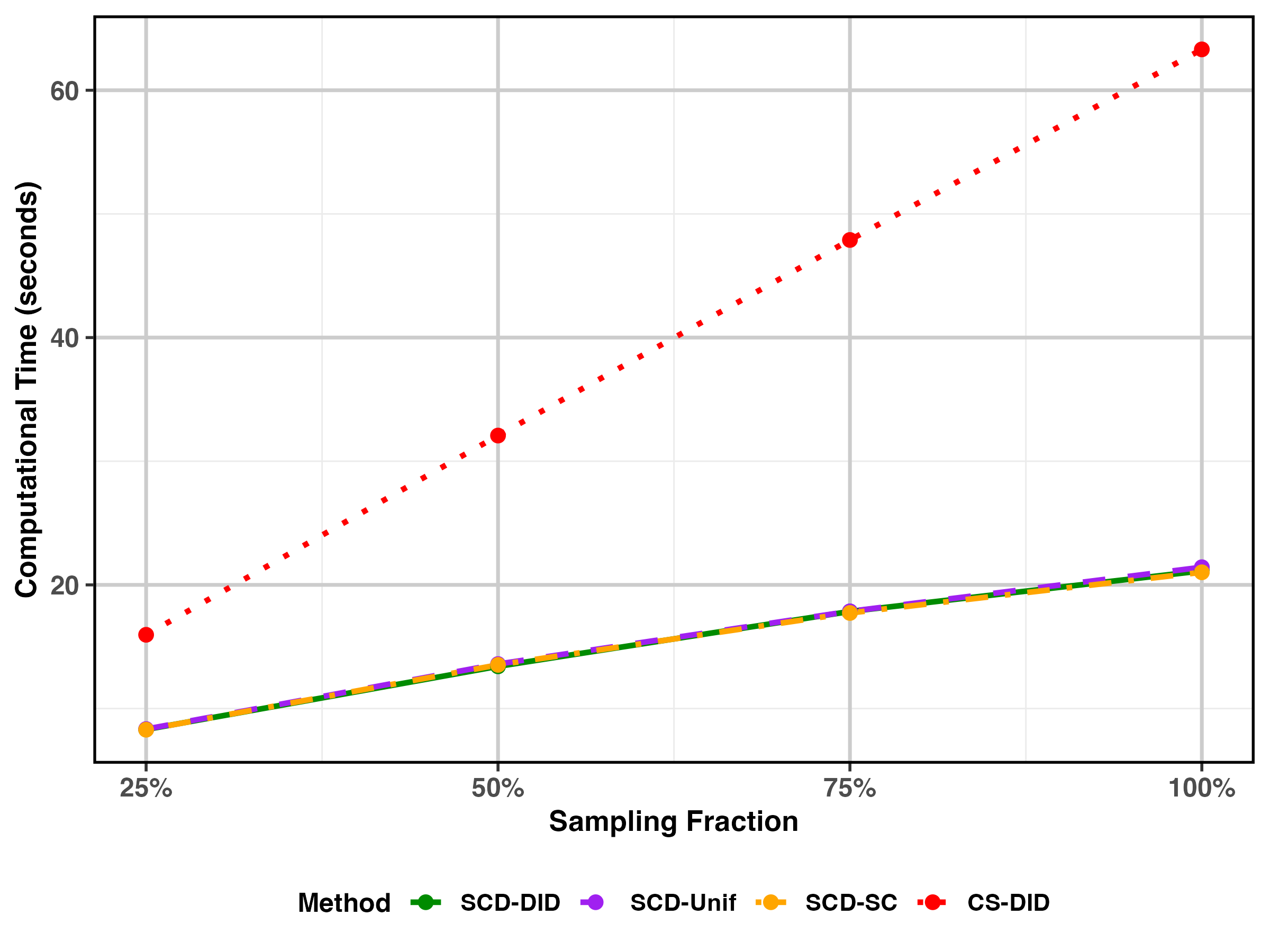}
        \caption{Discrete outcome, Repeated Cross-Section}
    \end{subfigure}
    \hfill
    \begin{subfigure}{0.48\textwidth}
        \centering
        \includegraphics[width=\textwidth]{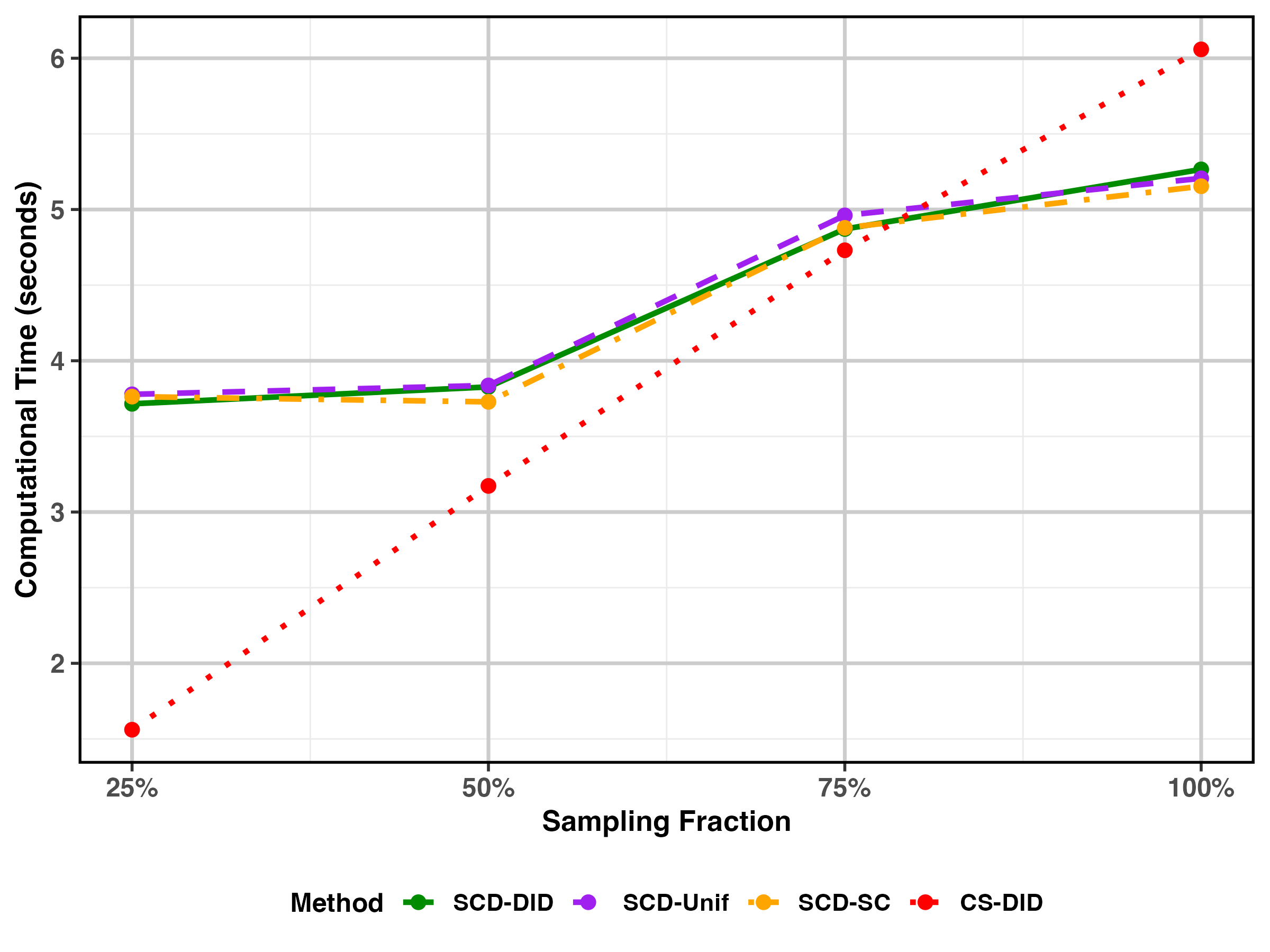}
        \caption{Continuous outcome, Repeated Cross-Section}
    \end{subfigure}

    \vspace{0.6cm}

    \begin{subfigure}{0.48\textwidth}
        \centering
        \includegraphics[width=\textwidth]{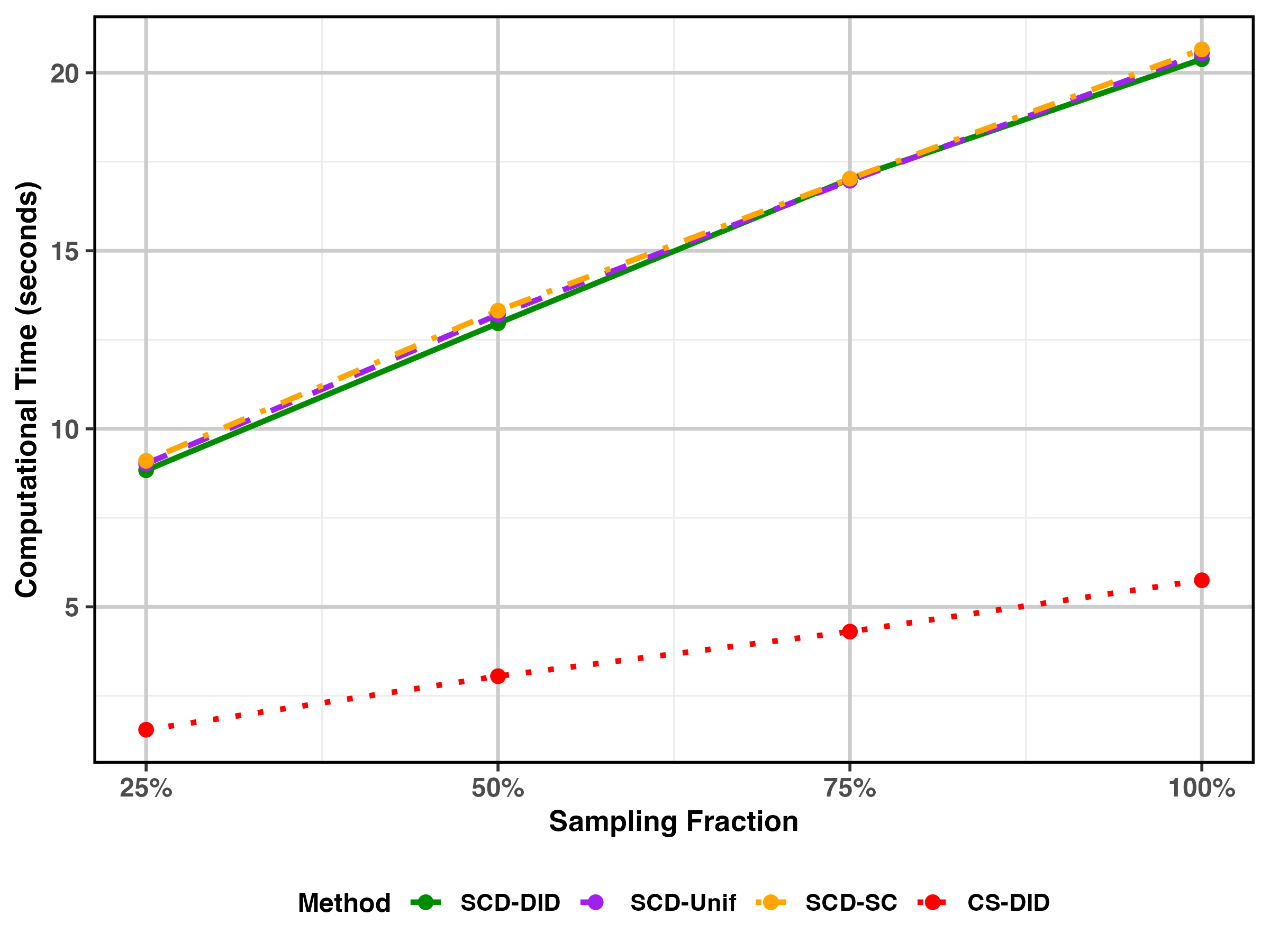}
        \caption{Discrete outcome, Panel}
    \end{subfigure}
    \hfill
    \begin{subfigure}{0.48\textwidth}
        \centering
        \includegraphics[width=\textwidth]{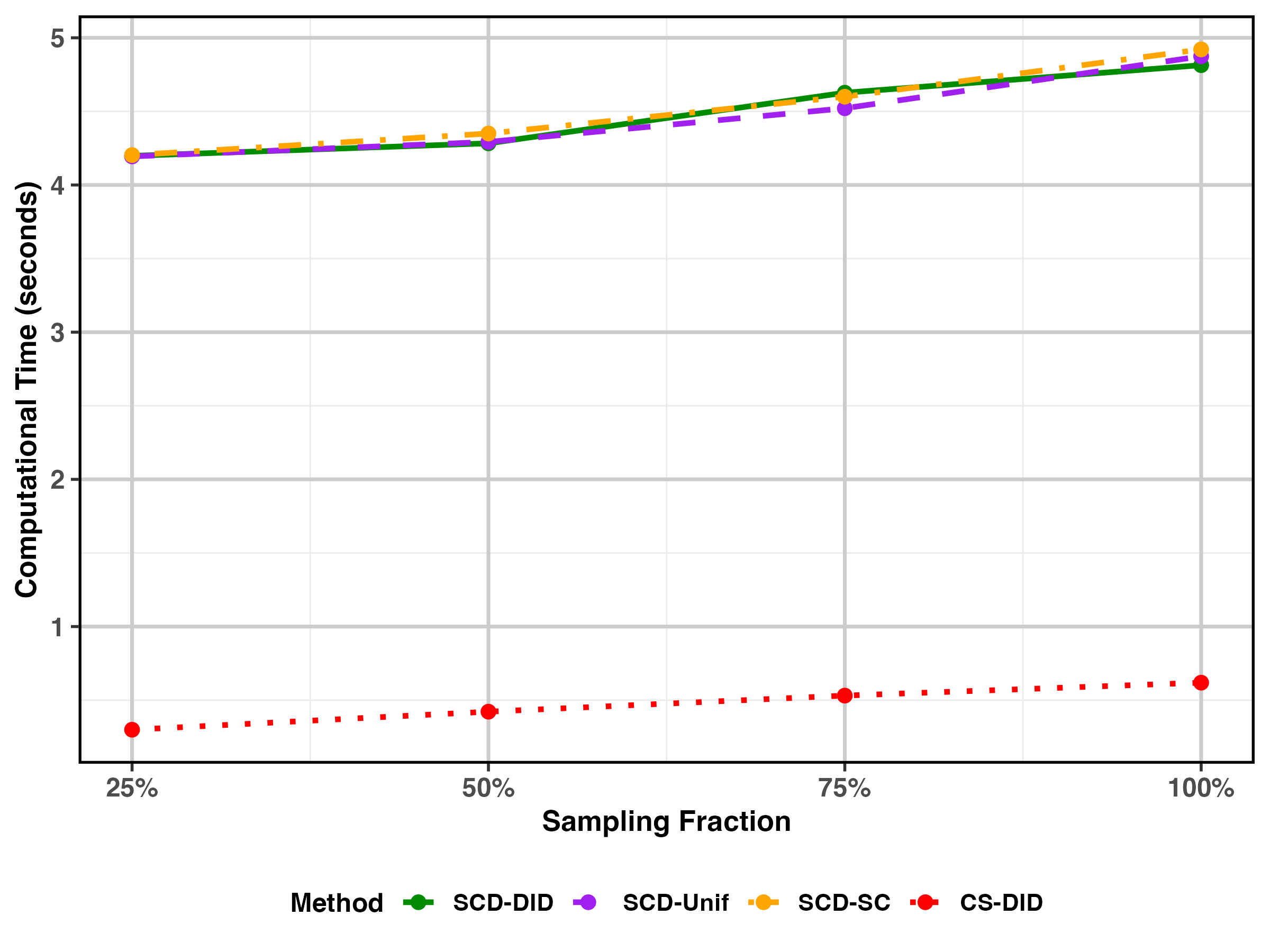}
        \caption{Continuous outcome, Panel}
    \end{subfigure}

    \vspace{0.4cm}
    
    \begin{minipage}{\textwidth}
    \footnotesize
    \textit{Notes:} The figure reports the computation time required to implement the SCD estimator across different combinations of sample size, outcome variables, and data structures. As a benchmark, we also report the computation time of the \texttt{did} R package version 2.3.0 of \cite{Callaway/SantAnna:JoE:21} (CS-DID). We consider two outcomes: (i) a discrete outcome equal to an indicator for being a non-U.S.-citizen Hispanic, and (ii) a continuous outcome given by the log of weekly earnings. Results are reported for both repeated cross-section data from the CPS between January 2003 and December 2009 and for a fictitious panel constructed from the same CPS data. For each sample fraction, we randomly draw that share five times and implement SCD on each draw. The reported computation time is the average across these five repetitions. SCD-DID refers to SCD with $\lambda=\lambda^{\mathsf{DID}}$, SCD-Unif uses $\lambda=\lambda^{\mathsf{unif}}$, and SCD-SC corresponds to SCD without differencing.     
    \end{minipage}
\end{figure}

Our SCD method relies on a Bonferroni approach to construct a confidence set for the weight $w^*$, so a natural concern is whether the computational cost of this procedure is prohibitive in practice. In this section, we demonstrate that Algorithm \ref{alg:confidence_intervals 1}, which uses simulated draws from the simplex, is computationally feasible for data dimensions commonly found in practice.

We first report the average computation times of constructing confidence intervals using Algorithm \ref{alg:confidence_intervals 1} on the monthly CPS data used in our empirical application in Section \ref{sec: emp app}. We restrict the time periods to January 2003-December 2009 and consider Arizona as treated after July 2007, so that $T = 84$ and $T^*=55$. In this application the number of donors is $K = 46$ states. We also consider two outcome variables: the indicator of the individual being non-U.S. citizen Hispanic (binary outcome) and the individual's log of weekly earnings (continuous outcome). We restrict our attention to individuals with strictly positive weekly earnings when analyzing the continuous outcome variable. On average, the total number of cross-sectional units per month is 128,932 for the case with a binary outcome, and 13,977 when considering a continuous outcome. To assess scalability, four sample fractions are considered: 25\%, 50\%, 75\%, and 100\%. For a given fraction, we randomly draw that share of the full data five times and implement our SCD method on each draw. The reported computation time is the average time across these five repetitions.\footnote{All computations are performed on an Apple M4 Max with 64GB of RAM.}

\begin{figure}[tbp]
    \centering
    \caption{Length of Post-Treatment Confidence Intervals.}
    \label{fig:CI lengths RC}

    \begin{subfigure}{0.9\textwidth}
        \centering
        \includegraphics[width=\textwidth]{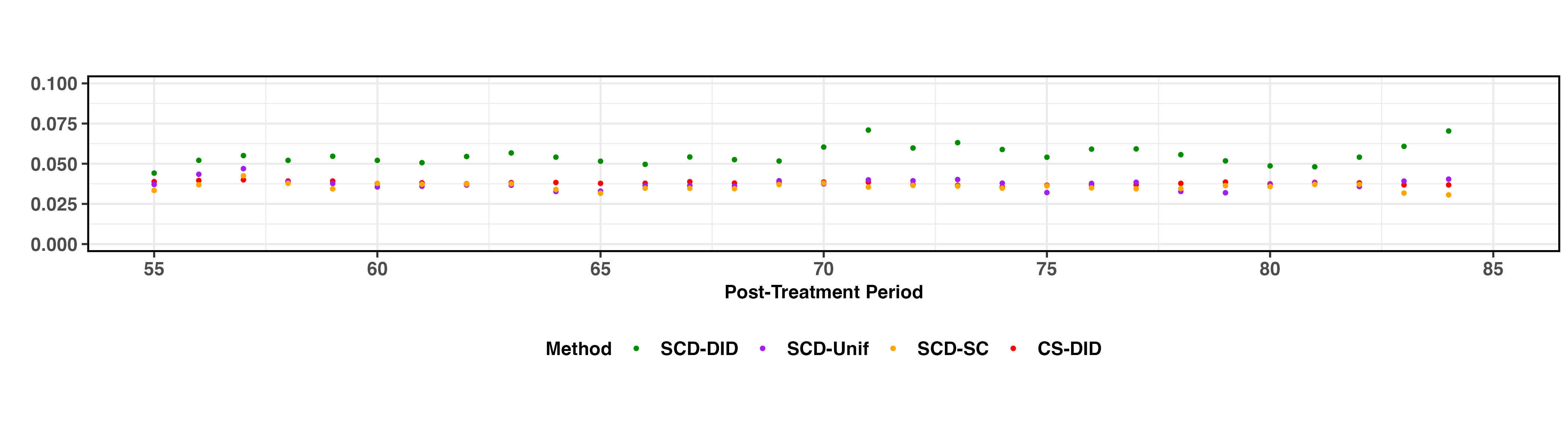}
        \vspace{-1cm}
        \caption{Discrete outcome, Repeated Cross-Section}
    \end{subfigure}
    
    \begin{subfigure}{0.9\textwidth}
        \centering
        \includegraphics[width=\textwidth]{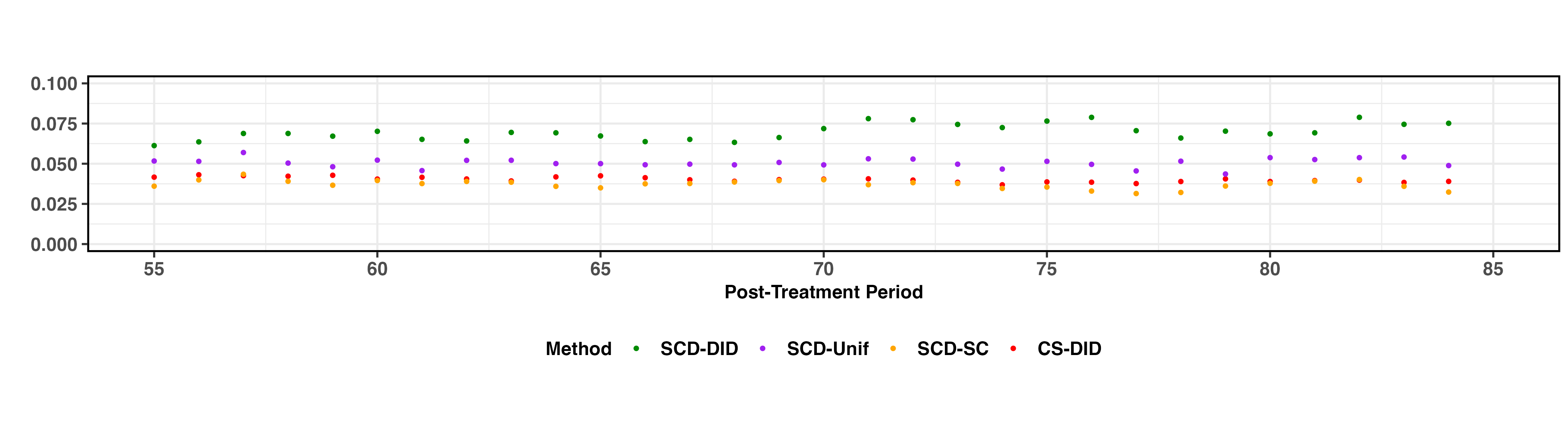}
        \vspace{-1cm}
        \caption{Discrete outcome, Panel}
    \end{subfigure}

    \begin{subfigure}{0.9\textwidth}
        \centering
        \includegraphics[width=\textwidth]{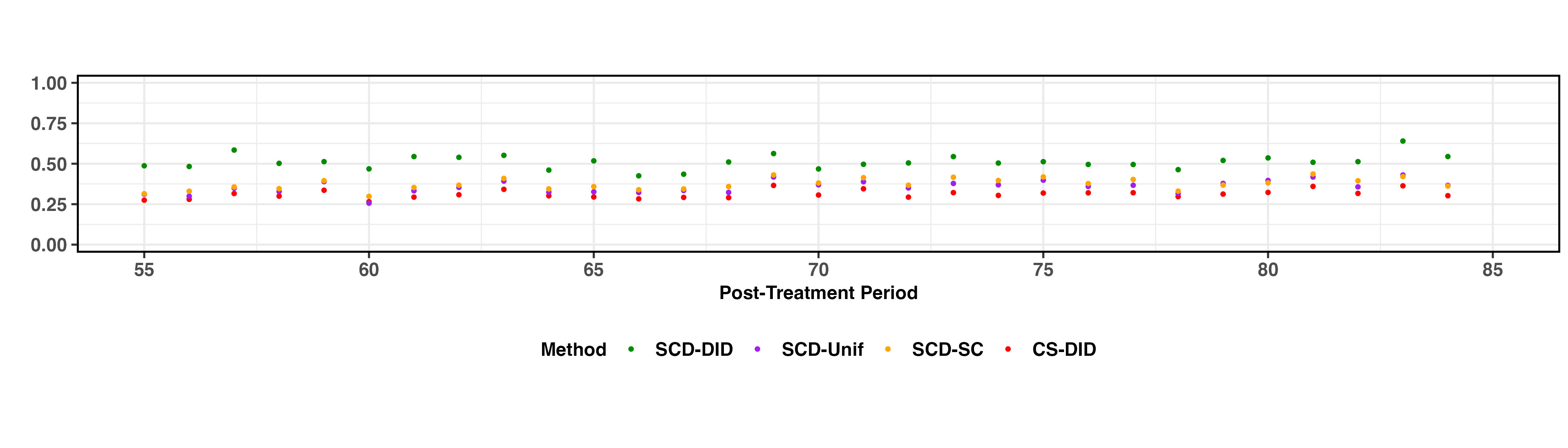}
        \vspace{-1cm}
        \caption{Continuous outcome, Repeated Cross-Section}
    \end{subfigure}
    
    \begin{subfigure}{0.9\textwidth}
        \centering
        \includegraphics[width=\textwidth]{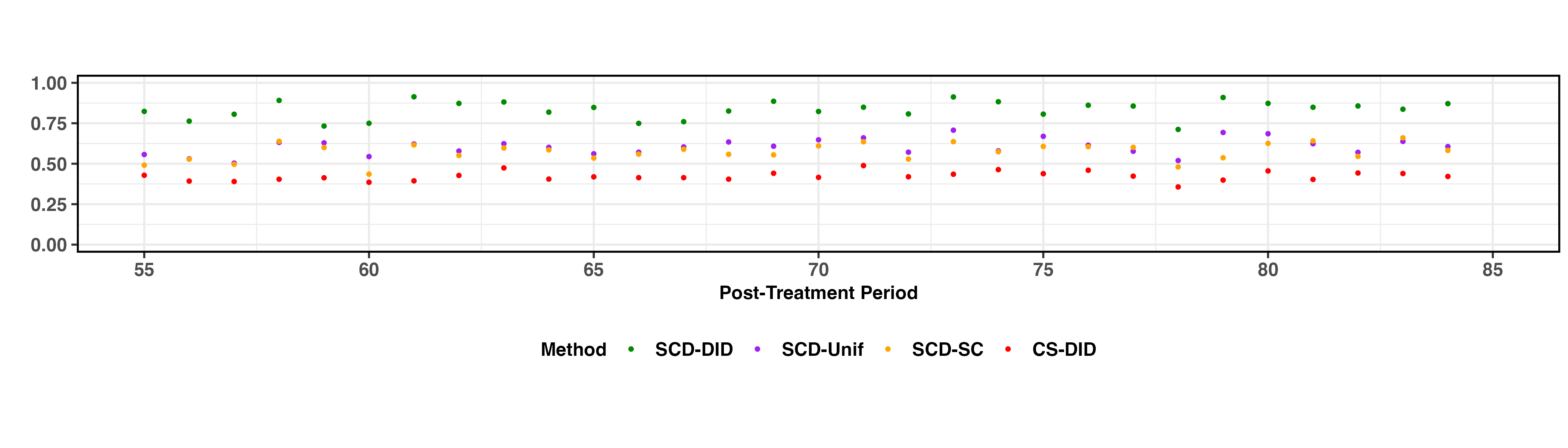}
        \vspace{-1cm}
        \caption{Continuous outcome, Panel}
    \end{subfigure}
    
    \vspace{0.4cm}
    
    \begin{minipage}{\textwidth}
        \footnotesize
        \textit{Notes:}  The figures report the length of the confidence intervals only for the confidence intervals in the post-treatment periods. Repeated cross-sectional data comes from the monthly CPS between January 2003 and December 2009. The panel data is created by independently drawing the same number of individuals per state-month from the CPS. Treatment occurs in July 2007 ($T^*=55$) in Arizona, after LAWA was passed. The number of donors is $K=46$ states. Panels (a) and (c) use as main outcome the indicator of the individual being non-U.S. citizen Hispanic, and panels (b) and (d) use the individual's log of weekly earnings. SCD inference is based on Algorithm \ref{alg:confidence_intervals 1}. CS-DID refers to the \texttt{did} R package of \cite{Callaway/SantAnna:JoE:21}. SCD-DID refers to the method of SCD using $\lambda = \lambda^{\mathsf{DID}}$, and SCD-Unif using $\lambda = \lambda^{\mathsf{unif}}$.
    \end{minipage}
\end{figure}

The first set of results is shown in panels (a) and (b) of Figure \ref{fig:SCD RC}. In panel (a), we document that SCD takes on average 20 seconds to construct the confidence intervals of treatment effects with a discrete outcome and the full sample. For comparison, we also report computation times for the \texttt{did} R package version 2.3.0 of \cite{Callaway/SantAnna:JoE:21} (CS-DID).\footnote{The package is available on the website: \url{https://cran.r-project.org/web/packages/did/index.html}} In this case, their package performs similar to SCD for small samples, but their computation time increases faster than our SCD method as the sample becomes larger. This difference likely reflects the greater generality of the \cite{Callaway/SantAnna:JoE:21} package, which accommodates multiple covariates and various estimation options. A similar pattern is observed in panel (b), where we consider a continuous outcome variable. CS-DID starts outperforming our procedure by a couple of seconds in small samples but its computation time grows faster, taking almost one second more than SCD once we consider the full sample. 

We next compare SCD and CS-DID under a panel data structure. To this end, we construct a balanced panel from the CPS used before by randomly drawing, in each month, the same number of individuals within a given state. For each state, the cross-sectional dimension of the panel is defined as the minimum number of observations available in that state across all months in the CPS. Sampled individuals are then assigned a panel identifier by interacting the state code with a within-state row number, producing individual ids that are consistent across time periods. The resulting panel contains $n=112,744$ observations every period for the discrete case, and $n=8,969$ for the continuous case. The results, shown in panels (c) and (d) of Figure \ref{fig:SCD RC}, indicate that although SCD remains computationally feasible, CS-DID outperforms our approach in all cases. This reversal relative to the repeated cross-section results arises because our inference procedure exploits the time-independence across samples that is present in repeated cross-section data, an advantage that disappears in panel settings. We also find that computation times are not affected by the choice of differencing parameter and that the computational cost of SCD does not increase exponentially with the number of cross-sectional units. 

Finally, we compare SCD and CS-DID in terms of the average length of their confidence intervals for the post-treatment periods across the different data structures and outcome variables considered above. The results are reported in Figure \ref{fig:CI lengths RC}. Consistent with the findings in Table \ref{table:simulation_results_newOG}, SCD produces, on average, longer confidence intervals than CS-DID. Interestingly, the gap between the two methods becomes less pronounced when the uniform differencing parameter is used or when no differencing is applied. These patterns hold in both repeated cross-section and panel settings.

\section{Empirical Application}
\label{sec: emp app}

\begin{table}[tbp]

\centering
\caption{Summary Statistics.}
\label{tab:summary_stats}
\small
\begin{tabular}{l*{6}{c}}
\toprule\toprule
                &\multicolumn{3}{c}{Arizona}&\multicolumn{3}{c}{Donor pool}\\
                \cmidrule(lr){2-4} \cmidrule(lr){5-7}
                &\multicolumn{1}{c}{2006}&\multicolumn{1}{c}{2009}&\multicolumn{1}{c}{Diff.}
                &\multicolumn{1}{c}{2006}&\multicolumn{1}{c}{2009}&\multicolumn{1}{c}{Diff.}\\                
\midrule
Age             &   35.168&   35.674&   0.506&   36.369&   36.814&   0.445\\
\addlinespace
Female          &    0.503&    0.503&   0.000&    0.511&    0.510&  -0.001\\
\addlinespace \addlinespace
\multicolumn{7}{l}{\textit{Educational attainment}} \\\addlinespace
\hspace*{0.1cm} Less than high school \hspace*{0.5cm}&    0.413&    0.375&  -0.038&    0.362&    0.349&  -0.013\\
\addlinespace
\hspace*{0.1cm} High school graduate&    0.227&    0.211&  -0.016&    0.240&    0.237&  -0.003\\
\addlinespace
\hspace*{0.1cm} Some college    &    0.201&    0.223&   0.022&    0.205&    0.209&   0.004\\
\addlinespace
\hspace*{0.1cm} College or more &    0.159&    0.190&   0.031&    0.193&    0.205&   0.012\\
\addlinespace \addlinespace
Employment        &    0.462&    0.462&   0.000&    0.485&    0.470&  -0.015\\
\addlinespace
Non-citizen Hispanic &    0.095&    0.063&  -0.032&    0.043&    0.042&  -0.001\\
\midrule
Observations    &     1,944&     1,627&   &   127,040&   124,880&  \\
\bottomrule\bottomrule
\end{tabular}
\begin{minipage}{\textwidth}
\vspace{0.4cm}
\footnotesize \textit{Notes:}  Cells for age display the mean and cells for other variables show proportions. Arizona's donor pool consists of 46 states without a similar regulation during the period analyzed. Columns 2 and 5 report January 2006 CPS statistics; Columns 3 and 6 report January 2009 CPS statistics; Columns 4 and 7 report the change between 2009 and 2006. Survey weights are used.
\end{minipage}

\end{table}

To illustrate our method, we revisit the empirical setting analyzed by \cite{Bohn/Lofstrom/Raphael:TRES:14} and study the effects of the 2007 Legal Arizona Workers Act (LAWA) on Arizona's internal composition. LAWA was passed in July 2007 and prohibited businesses from knowingly hiring unauthorized workers after December 31, 2007. In addition, this new law required all Arizona employers to verify the identity and work eligibility of new hires using an online system (called E-Verify) that cross-checks employee information against federal earnings and immigration databases. Employers who did not comply with the new rules faced sanctions like suspensions or permanent revocation of their business licenses. As one of the strictest state-level immigration laws at the time, it raised the costs of unauthorized employment for both employers and undocumented immigrants. 

In this context, the group membership variable ($G_{i}$) is defined as the state in which individual $i$ lives, the treated group is Arizona and the post-treatment period begins once LAWA is passed in July 2007. We use CPS microdata from January 1998 to December 2009 and follow the authors in considering 46 states ($K$) in Arizona's donor pool that did not implement any similar regulation during the period of analysis.\footnote{The excluded states are Mississippi, Rhode Island, South Carolina, and Utah. CPS data is provided in the replication package of \cite{Bohn/Lofstrom/Raphael:TRES:14}.} Nevertheless, unlike \cite{Bohn/Lofstrom/Raphael:TRES:14}, we do not aggregate the monthly CPS data to the annual level, which allows us to point identify the weights for Arizona's donor pool using SCD. Our dataset contains 114 months and 30 months in the pre and post-treatment periods, respectively, with a total of 144 time periods, i.e., $T^* = 115$ and $T = 144$. We focus on the population that is most likely to be affected by the policy change, so our primary outcome of interest, $Y_{i,t}$, is defined as an indicator variable equal to one if individual $i$ is Hispanic but not a U.S. citizen at time $t$ and zero otherwise. We apply SCD with the DID differencing parameter ($\lambda^{\mathsf{DID}}$) to estimate the average treatment effect on the treated, which, in this case, captures the causal impact of LAWA on the share of non-citizen Hispanic individuals in Arizona.

\begin{table}[tbp]

\centering
\caption{Arizona's Donors with Positive SCD Weights.}
\label{tab:donors_weights}
\small
\begin{tabular}{p{0.4\textwidth} 
                >{\centering\arraybackslash}p{0.1\textwidth}}
  \toprule\toprule
State & Weights \\ 
  \midrule
Connecticut & 0.204 \\ 
  Florida & 0.029 \\ 
  Georgia & 0.066 \\ 
  Idaho & 0.047 \\ 
  Illinois & 0.037 \\ 
  Kansas & 0.086 \\ 
  New Jersey & 0.378 \\ 
  New York & 0.017 \\ 
  Washington & 0.138 \\
  \bottomrule\bottomrule
\end{tabular}
\begin{minipage}{\textwidth}
\vspace{0.4cm}
\footnotesize \textit{Notes:}  Weights are obtained by applying SCD to Arizona and its donor pool, using as main outcome the proportion of non-citizen Hispanic and setting $\lambda = \lambda^{\mathsf{DID}}$. Arizona's donor pool consists of 46 states without any similar regulation during the period analyzed. Data come from the monthly CPS between January 1998 and December 2009. The sum of weights may differ from one due to rounding. Survey weights are used.
\end{minipage}

\end{table}

Table \ref{tab:summary_stats} presents descriptive statistics for Arizona and its donor pool one and a half years before and after LAWA's enactment. We observe small changes over time in both Arizona and its donor pool in terms of age, gender composition, and the employment-to-population ratio. In contrast, changes in Arizona's educational attainment distribution are more pronounced than in the donor pool between 2006 and 2009. In particular, the share of low-educated individuals (those with a high school diploma or less) declined by 5.4 percentage points in Arizona, compared to a 1.6 percentage-point reduction among donor states. Likewise, the variable of interest, the proportion of non-citizen Hispanic, fell by 3.2 percentage points (a 34\% drop) in Arizona, whereas the donor pool experienced only a marginal 0.1 percentage-point (a 2\% fall) decrease over the same period. These patterns are in line with the hypothesis that LAWA reshaped Arizona's demographic composition by tightening immigrants' access to employment opportunities. In the next subsection we provide an estimate of LAWA's causal effect on the internal composition of Arizona using SCD.

\subsection{Results}\label{sec:Emp_App_results}

\begin{figure}[tbp]
    \caption{Estimated Effects on Arizona's Share of Non-citizen Hispanic.}
    \label{fig:main_results}
    \vspace{0.4cm}
    \centering
    \begin{subfigure}{0.48\textwidth}
        \centering
        \resizebox{\textwidth}{!}{\input{synthetic_hispnoncitizen}}
        \caption{Proportion Non-citizen Hispanic}
        \label{fig:synthetic_hispnoncitizen}
    \end{subfigure}\hfill
    \begin{subfigure}{0.48\textwidth}
        \centering
        \resizebox{\textwidth}{!}{\input{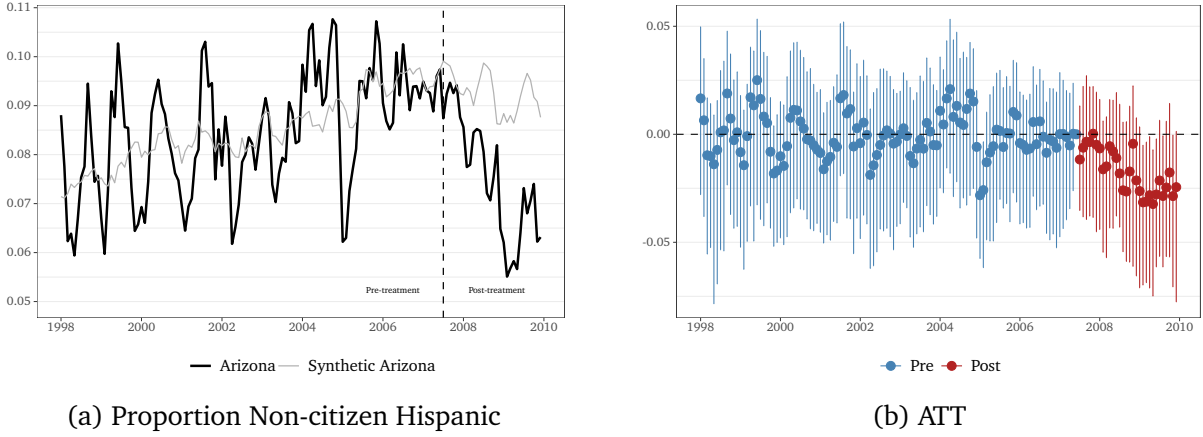}}
        \caption{ATT}
        \label{fig:att_hispnoncitizen}
    \end{subfigure}

    \begin{minipage}{\textwidth}
    \vspace{0.4cm}
    \footnotesize \textit{Notes:}  Panel (a) in this figure shows the evolution of Arizona's share of non-citizen Hispanic individuals compared to its synthetic version and panel (b) displays the corresponding ATT after LAWA's enactment in July 2007. We apply SCD with $\lambda^{\mathsf{DID}}$, where Arizona's donor pool consists of 46 states without any similar regulation during the period analyzed. The states in the donor pool with positive SCD weights are Connecticut (0.204), Florida (0.029), Georgia (0.066), Idaho (0.047), Illinois (0.037), Kansas (0.086), New Jersey (0.378), New York (0.017), and Washington (0.138). Data come from the monthly CPS between January 1998 and December 2009. The blue and red lines correspond to 95\% CIs constructed using Algorithm \ref{alg:confidence_intervals 1} for repeated cross-sectional data. Survey weights are used.
    \end{minipage}
\end{figure}

Table \ref{tab:donors_weights} reports the subset of states in Arizona's donor pool with positive weights from the SCD estimation using as main outcome the proportion of non-citizen Hispanic. The largest weight is assigned to New Jersey, followed by Connecticut and Washington. Smaller positive weights are assigned to Kansas, Georgia, Idaho, Illinois, Florida, and New York. Interestingly, the fact that all of Arizona's neighboring states receive a zero weight by SCD in the construction of synthetic Arizona suggests the presence of potential spillover effects following LAWA's enactment. In addition, none of the three states with positive SC weights found by \cite{Bohn/Lofstrom/Raphael:TRES:14} (California, Maryland, and North Carolina) are shown in Table \ref{tab:donors_weights}. Two main factors contribute to this discrepancy. First, our identification strategies are different. We invoke GMC with $\lambda^{\mathsf{DID}}$, so we need trends in averaged untreated potential outcomes to match between Arizona and its donor pool, which is less restrictive than the traditional SC approach that matches averaged untreated potential outcomes between Arizona and its donors directly. Secondly, the authors combine the CPS data at the annual level before applying SC, whereas we exploit the frequency of the CPS to obtain point-identification of SCD weights.\footnote{In their main SC analysis, the authors also incorporate covariates such as state unemployment rates and industrial composition of the workforce, yet their results remain virtually unchanged when these covariates are excluded.}

\begin{table}[tbp]

\centering
\caption{SCD Weights for Arizona's Donors Across Robustness Exercises.}
\label{tab:weights_robustness}
\begin{tabular}{lcccc}
  \toprule\toprule
State & \makecell[c]{Non-citizen Hispanic \\ Low-Educated} & \makecell[c]{Shorter \\ Pre-Treatment} & \makecell[c]{Uniform \\ Differencing} & \makecell[c]{No \\ Differencing} \\  
  \midrule
Alabama & 0.195 & 0 & 0 & 0 \\ 
Arkansas & 0 & 0.010 & 0 & 0 \\ 
California & 0 & 0 & 0 & 0.621 \\ 
Colorado & 0.070 & 0.247 & 0 & 0 \\ 
Connecticut & 0 & 0 & 0.134 & 0 \\ 
District of Columbia & 0 & 0 & 0.060 & 0 \\ 
Florida & 0 & 0 & 0 & 0.213 \\ 
Georgia & 0.168 & 0 & 0.155 & 0 \\ 
Idaho & 0 & 0 & 0.103 & 0 \\ 
Illinois & 0.013 & 0 & 0.010 & 0 \\ 
Kansas & 0 & 0.096 & 0 & 0 \\ 
Louisiana & 0 & 0.201 & 0.037 & 0 \\ 
Missouri & 0 & 0 & 0.023 & 0 \\ 
New Jersey & 0.245 & 0.382 & 0.265 & 0.166 \\ 
New York & 0 & 0 & 0.050 & 0 \\ 
North Carolina & 0.120 & 0.055 & 0 & 0 \\ 
South Dakota & 0 & 0 & 0.025 & 0 \\ 
Washington & 0.188 & 0 & 0.139 & 0 \\ 
Wisconsin & 0 & 0.009 & 0 & 0 \\ 
  \bottomrule\bottomrule
\end{tabular}
\begin{minipage}{\textwidth}
\vspace{0.4cm}
\footnotesize \textit{Notes:}  Cells contain the SCD weights for each robustness exercise described at the top of each column. The second column uses as main outcome the number of non-citizen Hispanics with a high-school diploma or less as a proportion of the prime-working age (15-45) state population. The third column applies SCD on the proportion of non-citizen Hispanics with a shorter pre-treatment window that starts in January 2003. The fourth column shows the states' weights after applying SCD with a uniform differencing parameter ($\lambda^{\mathsf{unif}}$), and the last column does not apply any differencing to the data. Arizona's donor pool consists of 46 states without any similar regulation during the period analyzed. Data come from the monthly CPS between January 1998 and December 2009. The sum of weights may differ from one due to rounding. Survey weights are used.
\end{minipage}

\end{table}

Figure \ref{fig:main_results} shows our main results for the share of non-citizen Hispanic after applying SCD. Panel (a) mirrors the standard plot used in the SC literature, displaying two time-series lines: one for Arizona (black) and another for its synthetic control (grey).\footnote{In SCD, the synthetic counterfactual outcomes for the treated unit are computed as follows: 
\begin{align*}
\widehat{\mathbf{E}[Y_{i,t}(0) \mid G_i = 0]}=\hat m_{0,t}-\hat \theta_{t}(\hat w), \text{ for all }t \in \mathcal{T}.    
\end{align*}} Overall, both lines follow a similar trend during the pre-treatment period, with Arizona's series exhibiting higher volatility than its synthetic counterpart. On the other hand, following the passage of LAWA, we observe a big drop in Arizona's proportion of non-citizen Hispanic relative to its synthetic control, going from 9.1\% to 6.4\% between June 2006 and December 2009. 

\begin{figure}[tbp]
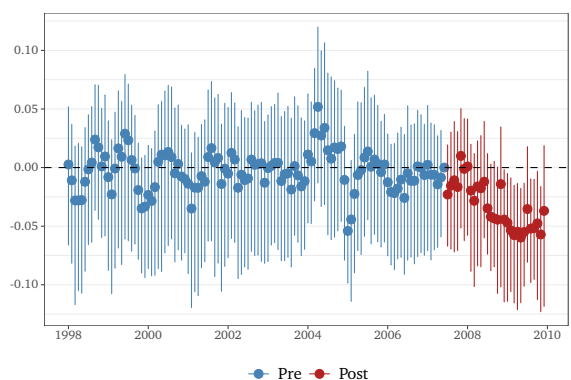

    \centering
    \caption{Robustness Checks for ATT.}
    \label{fig:robustness}
    \vspace{0.4cm}

    \begin{subfigure}{0.48\textwidth}
        \centering
        \resizebox{\textwidth}{!}{\input{att_hispnoncitizene1e2}}
        \caption{Low-Educated Non-citizen Hispanic}
        \label{fig:att_1}
    \end{subfigure}
    \hfill
    \begin{subfigure}{0.48\textwidth}
        \centering
        \resizebox{\textwidth}{!}{\input{att_hispnoncitizen_shorterpre}}
        \caption{Shorter Pre-Treatment Window}
        \label{fig:att_2}
    \end{subfigure}

    \vspace{0.6cm}

    \begin{subfigure}{0.48\textwidth}
        \centering
        \resizebox{\textwidth}{!}{\input{att_hispnoncitizen_lambdaunif}}
        \caption{Uniform Differencing}
        \label{fig:att_3}
    \end{subfigure}
    \hfill
    \begin{subfigure}{0.48\textwidth}
        \centering
        \resizebox{\textwidth}{!}{\input{att_hispnoncitizen_sc}}
        \caption{No Differencing}
        \label{fig:att_4}
    \end{subfigure}

    \vspace{0.4cm}

    \begin{minipage}{\textwidth}
        \vspace{0.4cm}
        \footnotesize
        \textit{Notes:}  Panel (a) in this figure shows the ATT on Arizona's share of non-citizen Hispanic with a high-school diploma or less with respect to the prime-working age (15-45) state population. Panel (b) presents the ATT on Arizona's proportion of non-citizen Hispanic after applying SCD with a shorter pre-treatment window that starts in January 2003. Panel (c) displays the ATT after applying SCD with a uniform differencing parameter ($\lambda^{\mathsf{unif}}$), and panel (d) uses no differencing at all. The blue and red lines correspond to 95\% CIs constructed using Algorithm \ref{alg:confidence_intervals 1} for repeated cross-sectional data. Arizona's donor pool consists of 46 donor states. Donors with positive SCD weights for each of the robustness exercises are shown in Table \ref{tab:weights_robustness}. Data come from the monthly CPS. Survey weights are used.
    \end{minipage}
\end{figure}

Panel (b) in Figure \ref{fig:main_results} shows LAWA's causal effects (estimated by $\hat \theta_{t}(\hat w)$) on Arizona's internal composition of non-citizen Hispanic along with a 95\% confidence band obtained via Bonferroni's method described in Algorithm \ref{alg:confidence_intervals 1} for repeated cross-sectional data. Similar to the usual pre-trend test used in empirical DID studies \citep{Callaway/SantAnna:JoE:21, Sun/Abraham:JoE:21, Roth/SantAnna/Bilinksi/Poe:JoE2023, Borusyak/Jaravel/Spiess:TRES:24}, the non-statistically significant estimates of the treatment effect during the pre-treatment period provide evidence in favour of the stability assumption of matching weights in SCD. Additionally, the SCD estimates for the post-treatment period indicate a negative effect of LAWA on Arizona's share of non-citizen Hispanics. In line with the 1.5 percentage point reduction reported by \cite{Bohn/Lofstrom/Raphael:TRES:14}, we find that the proportion of this demographic group declined by 2 percentage points on average after July 2007 (equivalent to 125,000 fewer individuals with respect to Arizona's CPS population in June 2007).

\subsection{Robustness Checks}  

To complement our main findings, we conduct four robustness exercises. First, since illegal immigrants tend to be low-educated, we refine our main outcome variable and use Arizona's share of non-citizen Hispanic with high school or less among the population aged 15-45. Second, we shorten the pre-treatment window by starting the sample in January 2003 instead of January 1998. This relaxes the trend-matching requirement in SCD and provides a robustness check against potential overfitting of early-period dynamics. Third, we test how our results change when SCD matches on trends relative to the pre-treatment average rather than the last pre-treatment period by applying our SCD method with a uniform differencing parameter ($\lambda^{\mathsf{unif}}$). Lastly, we report results using SCD without applying any differencing to the outcome, which is equivalent to the SC approach.

Table \ref{tab:weights_robustness} reports the subset of Arizona's donors with positive SCD weights for each robustness exercise. In general, donors contributing to synthetic Arizona differ across specifications. For instance, in column three where we adopt a shorter pre-treatment window, only Kansas and New Jersey also appear in Table \ref{tab:donors_weights}. When we apply SC to our monthly CPS data, we obtain three donors with positive weights (California, Florida, and New Jersey), two of which coincide with those reported in Table \ref{tab:donors_weights}. Interestingly, New Jersey is the only state with positive SCD weights across all exercises, highlighting its relevance as a comparison group for Arizona. 

We present the ATT results for each robustness exercise in Figure \ref{fig:robustness}. Panel (a) reveals that Arizona's share of low-educated non-citizen Hispanics is reduced by 4.6 percentage points after one and a half years of LAWA's enactment, suggesting that the policy's impact was concentrated among less-educated immigrants. Next, in panel (b), we find that reducing the number of pre-treatment periods does not affect our baseline results and estimate an average post-treatment decline of 1.9 percentage points in Arizona's share of non-citizen Hispanic. Finally, panels (c) and (d) show that our main ATT estimates in subsection \ref{sec:Emp_App_results} remain robust to alternative differencing choices in SCD. Specifically, we estimate an average post-LAWA decline of 1.8 percentage points when using uniform differencing and 1.3 percentage points when no differencing is applied.

Overall, these results show that Arizona experienced a significant change in its internal composition after LAWA, suggesting that the policy discouraged undocumented workers from residing in the state.

\section{Conclusion}\label{sec: conclusion}

This paper considers the general framework of causal inference with groupwise matching. We show that many widely used designs, including RCTs, difference-in-differences (DID), synthetic control (SC), synthetic control with differencing (SCD), and synthetic difference-in-differences (SDID), can be understood through the lenses of a Generalized Matching Condition (GMC) and the choice of within-group differencing parameters and matching weights. This perspective provides a clear distinction of how these methods extrapolate untreated outcomes. In particular, we show that DID relies on group-size based weights justified by the parallel trend assumption (PTA), whereas SCD relies on weights chosen by pre-treatment fit and justified by the stability of matching weights across time.

We examine the complementarity between DID and SCD through a regret analysis, showing that DID regret-dominates SCD when SCD's extrapolation error exceeds DID's matching error up to a term that vanishes at the parametric rate. In the opposite case, SCD regret-dominates DID, so neither method uniformly dominates the other. We further establish that, under conditions requiring both pre- and post-treatment parallel trends, DID and SCD are equivalent in terms of GMC. In this case, the weights estimated by SCD using pre-treatment data coincide with the group-size weights implicit in DID. When PTA fails, however, SCD provides an alternative to DID where identification no longer relies on parallel trends, but instead on a Stable Matching Condition (SMC). We develop a uniformly valid inference procedure for SCD based on a Bonferroni approach. Finally, our theoretical insights are illustrated via Monte Carlo experiments and revisiting the empirical application of \cite{Bohn/Lofstrom/Raphael:TRES:14}, who study the impact of the 2007 Legal Arizona Workers Act (LAWA) on Arizona's internal composition. 

While the GMC is formulated in a setting where the target parameter is the average causal effect of a treatment, it is conceivable that the condition extends to the setting where the target parameter is the distribution of the causal effect. This opens up the question for future research of how the causal inference designs such as changes in changes of \cite{Athey/Imbens:Eca:06} and the distributional synthetic control of \cite{Gunsilius:Eca:23} compare.

\bibliographystyle{ecca}
\bibliography{microSCD}

\newpage
\appendix
\setcounter{section}{0}

\renewcommand{\thesection}{\Alph{section}}
\renewcommand{\thetable}{\thesection.\arabic{table}}
\renewcommand{\thefigure}{\thesection.\arabic{figure}}

\counterwithin{table}{section}
\counterwithin{figure}{section}

\fontsize{12pt}{12pt}\selectfont

\begin{center}
    \Large \textsc{Supplemental Note to ``Causal Inference with Groupwise Matching''}
\end{center}

\date{%
    \today%
}

\vspace*{2ex minus 1ex}
\begin{center}
    Ratzanyel Rinc\'on and Kyungchul Song\\
    \textit{Vancouver School of Economics, University of British Columbia}
    \bigskip
    \bigskip
    \bigskip
\end{center}

\section{Proofs of the Results in Sections \ref{sec:causal inf groupwise matching} and \ref{sec: causal inf generalized matching}}

\begin{lemma}
    \label{lemm:theta*}
  Suppose that Assumption \ref{assump:basic} holds. Then, for $t \in \mathcal{T}$,
  \begin{align}
    \theta_t^* = \theta_t(\lambda,w) - e_t(\lambda,w).
\end{align}
\end{lemma}

\begin{proof}
    By (\ref{Y_{it}}) and Assumption \ref{assump:basic}(ii), for all $t \in \mathcal{T}_0$, $m_{0,t} = m_{0,t}(0)$. Furthermore, for all $j  \in \mathcal{G}_{\mathsf{don}}$, $m_{j,t} = m_{j,t}(0)$ for all $t \in \mathcal{T}$, by (\ref{Y_{it}}) and the definition of $m_{j,t}(0)$. Hence, 
    \begin{align*}
        \theta_t(\lambda,w) - e_t(\lambda,w) &= m_{0,t} - m_{0,t}(0) - \sum_{j=1}^K(m_{j,t} - m_{j,t}(0))w_j\\
                                             &= m_{0,t} - m_{0,t}(0).
    \end{align*}
    Now by (\ref{Y_{it}}), if we let $m_{0,t}(1) = \mathbf{E}[Y_{i,t}(1) \mid D_i = 1]$, we can write the last difference as $m_{0,t}(1) - m_{0,t}(0)$. This delivers the desired result.
\end{proof}

\noindent \textbf{Proof of Proposition \ref{prop: GMC factor}: } First, note that 
\begin{align}
    \label{e_t factor}
   e_t(\lambda,w) = \left( \mathbf{E}[ \Lambda_i' \mid G_i = 0] - \sum_{j=1}^K \mathbf{E}[ \Lambda_i' \mid G_i = j] w_j \right)F_t(\lambda).
\end{align}
Thus, (iii) implies (ii), which implies (i). Now, suppose that the full row rank condition holds. Define $\boldsymbol{e}(\lambda,w) = [e_{T^*}(\lambda,w),...,e_{T}(\lambda,w)]'$. Then, we have 
\begin{align*}
    \mathbf{E}[ \Lambda_i' \mid G_i = 0] - \sum_{j=1}^K \mathbf{E}[ \Lambda_i' \mid G_i = j] w_j = \boldsymbol{e}(\lambda,w)' \mathbf{F}(\lambda)' \left(\mathbf{F}(\lambda) \mathbf{F}(\lambda)' \right)^{-1}.
\end{align*} 
Now, suppose that (i) holds. Then, this implies (iii), completing the proof. $\blacksquare$

\begin{lemma}
\label{lemm: PTA connection}
(i) Suppose that $e_{T^*-1}^{\mathsf{DID}}(\lambda) = 0$ holds for some $\lambda \in \Delta_{|\mathcal{T}_0|-1}$. Then, PTA-I holds if and only if PTA($\lambda$) holds.

(ii) Suppose that  for some $\lambda \in \Delta_{|\mathcal{T}_0|-1}$, $e_{j,T^*-1}^{\mathsf{DID}}(\lambda) = 0$ holds for each $j \in \mathcal{G}_{\mathsf{don}}$. Then, PTA-II holds if and only if PTA-U($\lambda$) holds.
\end{lemma}

\begin{proof}
(i) We assume that $e_{T^*-1}^{\mathsf{DID}}(\lambda) = 0$ for some $\lambda \in \Delta_{|\mathcal{T}_0|-1}$. Now, suppose that PTA($\lambda$) holds. Then, $e_t^{\mathsf{DID}}(\lambda) = 0$ for all $t \in \mathcal{T}_1 \cup \{T^*-1\}$. Hence, for all $t \in \mathcal{T}_1$, 
\begin{align*}
    \mathbf{E}[Y_{i,t}(0;\lambda) \mid G_i = 0] &= \mathbf{E}[Y_{i,t}(0;\lambda) \mid G_i \in \mathcal{G}_{\mathsf{don}}] \text{ and }\\
    \mathbf{E}[Y_{i,t-1}(0;\lambda) \mid G_i = 0] &= \mathbf{E}[Y_{i,t-1}(0;\lambda) \mid G_i \in \mathcal{G}_{\mathsf{don}}].
\end{align*}
Subtracting the second equation from the first one, we obtain 
\begin{align*}
    \mathbf{E}[\Delta Y_{i,t}(0) \mid G_i = 0] = \mathbf{E}[\Delta Y_{i,t}(0) \mid G_i \in \mathcal{G}_{\mathsf{don}}],
\end{align*}
for all $t \in \mathcal{T}_1$. Hence, PTA-I holds.

Conversely, suppose that PTA-I holds. Then, 
\begin{align*}
    \mathbf{E}[Y_{i,t}(0;\lambda) \mid G_i = 0] &= \mathbf{E}\left[ Y_{i,t}(0) - Y_{i,T^*-1}(0) \mid G_i = 0 \right] + \mathbf{E}\left[ Y_{i,T^*-1}(0) - \sum_{s \in \mathcal{T}_0} Y_{i,s}(0) \lambda_s \mid G_i = 0 \right]\\
    &= \mathbf{E}\left[ Y_{i,t}(0) - Y_{i,T^*-1}(0) \mid G_i = 0 \right] + \mathbf{E}\left[ Y_{i,T^*-1}(0) - \sum_{s \in \mathcal{T}_0} Y_{i,s}(0) \lambda_s \mid G_i \in \mathcal{G}_{\mathsf{don}}\right],
\end{align*}
because $e_{T^*-1}^{\mathsf{DID}}(\lambda) = 0$. The last sum of two conditional expectations is written as 
\begin{align*}
    &\sum_{\ell=T^*}^t \mathbf{E}[\Delta Y_{i,\ell}(0) \mid G_i = 0] + \mathbf{E}\left[ Y_{i,T^*-1}(0) - \sum_{s \in \mathcal{T}_0} Y_{i,s}(0) \lambda_s \mid G_i \in \mathcal{G}_{\mathsf{don}}\right]\\
    &= \sum_{\ell=T^*}^t \mathbf{E}[\Delta Y_{i,\ell}(0) \mid G_i \in \mathcal{G}_{\mathsf{don}}] + \mathbf{E}\left[ Y_{i,T^*-1}(0) - \sum_{s \in \mathcal{T}_0} Y_{i,s}(0) \lambda_s \mid G_i \in \mathcal{G}_{\mathsf{don}}\right]\\
    &= \mathbf{E}\left[ Y_{i,t}(0) - \sum_{s \in \mathcal{T}_0} Y_{i,s}(0) \lambda_s \mid G_i \in \mathcal{G}_{\mathsf{don}} \right] = \mathbf{E}[Y_{i,t}(0;\lambda) \mid G_i \in \mathcal{G}_{\mathsf{don}}].
\end{align*}
The first equality follows by PTA-I. Hence, PTA($\lambda$) holds.

(ii) We assume that $e_{j,T^*-1}^{\mathsf{DID}}(\lambda) = 0$ for each $j  \in \mathcal{G}_{\mathsf{don}}$, for some $\lambda \in \Delta_{|\mathcal{T}_0|-1}$. First, assume that PTA-U($\lambda$) holds. We choose $j \in \mathcal{G}_{\mathsf{don}}$, and replace the event $G_i \in \mathcal{G}_{\mathsf{don}}$ by the event $G_i = j$ in the proof of (i), to obtain that 
\begin{align*}
    \mathbf{E}[\Delta Y_{i,t}(0) \mid G_i = 0] = \mathbf{E}[\Delta Y_{i,t}(0) \mid G_i = j],
\end{align*}
for all $t \in \mathcal{T}_1$. Since the choice of $j$ was arbitrary, we obtain PTA-II. Conversely, suppose that PTA-II holds. Again, choose $j \in \mathcal{G}_{\mathsf{don}}$, and replace the event $G_i \in \mathcal{G}_{\mathsf{don}}$ by the event $G_i = j$ in the proof of (i) to obtain 
\begin{align*}
    \mathbf{E}[Y_{i,t}(0;\lambda) \mid G_i = 0] = \mathbf{E}[Y_{i,t}(0;\lambda) \mid G_i = j]
\end{align*}
for each $j  \in \mathcal{G}_{\mathsf{don}}$.
\end{proof}

\begin{lemma}
    \label{lemm: PTA2}
    Suppose that Assumption \ref{assump:basic} holds, and let 
    \begin{align*}
        w^{\mathsf{DID}} = [w_{1}^{\mathsf{DID}},...,w_{K}^{\mathsf{DID}}]', 
    \end{align*}
    where $w_{j}^{\mathsf{DID}}$ is as defined in (\ref{DID}). Then, the following statements hold.

    (i) PTA($\lambda$) holds if and only if GMC holds at $(\lambda,w^{\mathsf{DID}})$.
    
    (ii) PTA-U($\lambda$) holds if and only if GMC holds at $(\lambda,w)$ for all $w \in \Delta_{K-1}$.
\end{lemma}

\begin{proof} (i) Notice that $e_t(\lambda,w)$ and $e_t^{\mathsf{DID}}(\lambda)$ have the following relationship: 
\begin{align*}
    e_t^{\mathsf{DID}}(\lambda) &= \mathbf{E}[Y_{i,t}(0;\lambda) \mid G_i = 0] - \mathbf{E}[Y_{i,t}(0;\lambda) \mid G_i \in \mathcal{G}_{\mathsf{don}}]\\
    &= \mathbf{E}[Y_{i,t}(0;\lambda) \mid G_i = 0] - \sum_{j=1}^K \mathbf{E}[Y_{i,t}(0;\lambda) \mid G_i = j] w_j^{\mathsf{DID}} = e_t(\lambda,w^{\mathsf{DID}}), \enspace \text{ for all } t \in \mathcal{T},
\end{align*}
where $w_j^{\mathsf{DID}} = P\{G_i = j \mid G_i \in \mathcal{G}_{\mathsf{don}}\}$. Hence, 
\begin{align*}
    \text{PTA}(\lambda) \text{ holds.} \Leftrightarrow \text{GMC}(\lambda, w^{\mathsf{DID}}) \text{ holds}. 
\end{align*}

(ii) Note that 
\begin{align}
    \label{der}
    \sum_{j=1}^K e_{j,t}^{\mathsf{DID}}(\lambda) w_j &= \sum_{j=1}^K (\mathbf{E}[Y_{i,t}(0;\lambda) \mid G_i = 0] - \mathbf{E}[Y_{i,t}(0;\lambda) \mid G_i = j]) w_j\\ \notag
    &= \mathbf{E}[ Y_{i,t}(0;\lambda) \mid G_i = 0] - \sum_{j=1}^K \mathbf{E}[Y_{i,t}(0;\lambda) \mid G_i = j] w_j = e_t(\lambda,w), \text{ for all } t \in \mathcal{T}.
\end{align}
Hence, \begin{align*}
    \text{PTA-U}(\lambda) \text{ holds.} &\Rightarrow e_t(\lambda,w) = 0, \text{ for all } t \in \mathcal{T}_1 \text{ and all } w \in \Delta_{K-1},\\
    &\text{ i.e., GMC($\lambda,w$) holds for all $w \in \Delta_{K-1}$.}
\end{align*}

Conversely, suppose that $\text{GMC}(\lambda,w)$ holds for all $w \in \Delta_{K-1}$. Then, for any $\ell = 1,...,K$, we have $\tilde w^{\ell} \in \Delta_{K-1}$, where $\tilde w^{\ell}$ denotes the $K$-dimensional vector of zeros except for the $\ell$-th entry which is equal to one. Then, 
\begin{align*}
    \text{GMC holds at } (\lambda,\tilde w^{\ell}) \Rightarrow e_{\ell,t}^{\mathsf{DID}}(\lambda) = \sum_{j=1}^K e_{j,t}^{\mathsf{DID}}(\lambda) \tilde w_j^{\ell} = e_t(\lambda, \tilde w^{\ell}) = 0,
\end{align*}
for all $t \in \mathcal{T}_1$, where the last equality follows from (\ref{der}). We can repeat this for all $\ell = 1,...,K$, to obtain that PTA-U($\lambda$) holds.
\end{proof}

\noindent \textbf{Proof of Proposition \ref{prop: PTA}: } Note that $e_{T^*-1}^{\mathsf{DID}}(\lambda^{\mathsf{DID}}) = 0$. Hence, the desired result follows by Lemmas \ref{lemm: PTA connection} and \ref{lemm: PTA2}. $\blacksquare$

\section{Proofs of the Results in Section \ref{sec: Comparison}}

\noindent \textbf{Proof of Proposition \ref{prop: PTA connection}: } The proposition is the same as Lemma \ref{lemm: PTA connection}. $\blacksquare$\medskip

We turn to the proof of Theorem \ref{thm: regret analysis}. For the results below, we assume that the assumptions of the theorem hold. Recall the definition $\mathsf{MER}_{d,P}(w) = \eta_{d,P}(w) - \inf_{\tilde w \in \mathcal{D}} \mathbf{E}_P[\eta_{d,P}(\tilde w(Z))]$, where $\eta_{d,P}(w) = \mathsf{SSME}_{d,P}(w)$, $d=0,1$. For $d=0,1$, we define 
\begin{align*}
    \hat \eta_d(w) = \frac{1}{|\mathcal{T}_d|} \sum_{t \in \mathcal{T}_d} \hat e_t^2(\lambda,w).
\end{align*}
Despite the notation, we cannot actually recover $\hat \eta_1(w)$ from data, because we do not observe the untreated potential outcomes for the treated group. That is, we do not observe $Y_{i,t}(0)$ for $t \in \mathcal{T}_1$. However, we can construct $\hat \eta_{0}(w)$ from data.

Also, let 
\begin{align*}
    \hat \varepsilon_{j,t} = \hat \mu_{j,t}(\lambda) - \mu_{j,t}(\lambda) \text{ and } \hat \varepsilon_{j,t}(0) = \hat \mu_{j,t}(0;\lambda) - \mu_{j,t}(0;\lambda),
\end{align*}
where $\hat \mu_{j,t}(0;\lambda)$ is defined to be the same as $\hat \mu_{j,t}(\lambda)$ except that $Y_{i,t}$ is replaced by $Y_{i,t}(0)$.

\begin{lemma}
\label{inf in and out}
For each $P \in \mathcal{P}$,
\begin{align*}
    &\inf_{\tilde w \in \mathcal{D}} \mathbf{E}_P[\hat \eta_{0}(\tilde w(Z))] = \inf_{\tilde w \in \mathcal{D}_0} \mathbf{E}_P[\hat \eta_{0}(\tilde w(Z_0))] = \mathbf{E}_P\left[\inf_{w \in \Delta_{K-1}} \hat \eta_{0}(w)\right], \text{ and } \\
    &\inf_{\tilde w \in \mathcal{D}} \mathbf{E}_P[\eta_{0,P}(\tilde w(Z))] = \inf_{w \in \Delta_{K-1}} \eta_{0,P}(w).
\end{align*}
\end{lemma}
\noindent \textbf{Proof: } First, note that 
\begin{align}
    \label{inf eq}
    \inf_{\widetilde w \in \mathcal{D}} \mathbf{E}_P\left[ \hat \eta_{0}(\widetilde w(Z)) \right] &\le \inf_{\widetilde w \in \mathcal{D}_0} \mathbf{E}_P\left[ \hat \eta_{0}(\widetilde w(Z_0)) \right] \\ \notag
    &= \mathbf{E}_P\left[ \inf_{w \in \Delta_{K-1}} \hat \eta_{0}(w) \right] \le \inf_{\widetilde w \in \mathcal{D}} \mathbf{E}_P\left[ \hat \eta_{0}(\widetilde w(Z)) \right],
\end{align} 
where $\mathcal{D}_0$ denotes the $\Delta_{K-1}$-valued maps that are measurable with respect to the $\sigma$-field generated by the pre-treatment data $Z_0$.\footnote{Note that $\hat \eta_{0}(w)$ is continuous in $w$ everywhere. Hence, $\inf_{w \in \Delta_{K-1}} \hat \eta_{0}(w) = \inf_{w \in \Delta_{K-1} \cap \mathbf{Q}^{K}} \hat \eta_{0}(w)$, where $\mathbf{Q}$ is the set of rational numbers. Therefore, $\inf_{w \in \Delta_{K-1}} \hat \eta_{0}(w)$ is a random variable.} The equality above follows because $\hat \eta_{0}(\cdot)$ is measurable with respect to the $\sigma$-field generated by $Z_0$.

The second statement follows because
\begin{align*}
   \inf_{\tilde w \in \mathcal{D}} \mathbf{E}_P[\eta_{0,P}(\tilde w(Z))] \le \inf_{w \in \Delta_{K-1}} \eta_{0,P}(w) =  \mathbf{E}_P\left[ \inf_{w \in \Delta_{K-1}} \eta_{0,P}(w) \right] \le \inf_{\tilde w \in \mathcal{D}} \mathbf{E}_P[\eta_{0,P}(\tilde w(Z))]. 
\end{align*}
The first inequality follows because $\mathcal{D}$ includes constant maps taking values in $\Delta_{K-1}$, and the equality follows because $\eta_{0,P}(\cdot)$ is nonstochastic. $\blacksquare$

\begin{lemma}
    \label{lemma: strong ID SCD}
  For $c>0$ in Assumption \ref{assump: moments},
  \begin{align*}
    \mathsf{SSME}_{0,P}(w) - \inf_{\tilde w \in \Delta_{K-1}} \mathsf{SSME}_{0,P}(\tilde w) \ge c\|w - w_P^{\mathsf{SCD}}\|^2.
\end{align*}
\end{lemma}

\noindent \textbf{Proof: } Let $h(\lambda) = [\mu_{0,1}(\lambda),...,\mu_{0,T^*-1}(\lambda)]'$. Note that 
\begin{align*}
    &\mathsf{SSME}_{0,P}(w) - \inf_{\tilde w \in \Delta_{K-1}} \mathsf{SSME}_{0,P}(\tilde w)\\
    &= \frac{1}{|\mathcal{T}_0|} (\Gamma_P(w - w_P^{\mathsf{SCD}}))'(\Gamma_P(w - w_P^{\mathsf{SCD}})) + \frac{2}{|\mathcal{T}_0|}  (\Gamma_P w_P^{\mathsf{SCD}} - h(\lambda))' \Gamma_P(w - w_P^{\mathsf{SCD}})\\
    &\ge \frac{1}{|\mathcal{T}_0|} (\Gamma_P(w - w_P^{\mathsf{SCD}}))'(\Gamma_P(w - w_P^{\mathsf{SCD}})) \ge \frac{\lambda_{\min}\left( \Gamma_P'\Gamma_P \right)}{|\mathcal{T}_0|}\| w - w_P^{\mathsf{SCD}}\|^2,
\end{align*}
where the inequality follows because $(\Gamma_P w_P^{\mathsf{SCD}} - h(\lambda))' \Gamma_P(w - w_P^{\mathsf{SCD}}) \ge 0$ by the optimality of $w_P^{\mathsf{SCD}}$ (e.g., Propositions 2.1.5 and 2.3.2 of \cite{clarke1990}).
 $\blacksquare$\medskip

Define 
\begin{align*}
    \mathsf{Regret}_{1,P}(\hat w) = \mathbf{E}_P[\ell_1(\hat w)] - \inf_{\widetilde w \in \mathcal{D}} \mathbf{E}_P\left[ \ell_1(\widetilde w(Z))\right],
\end{align*}
and let
\begin{align*}
    R_{n,P}(w) = - \frac{2}{|\mathcal{T}_1|} \sum_{t \in \mathcal{T}_1} e_t(\lambda,w)\left( \theta_t(\lambda,w) - \hat \theta_t(\lambda,w) \right) + \frac{1}{|\mathcal{T}_1|} \sum_{t \in \mathcal{T}_1} \left( \theta_t(\lambda,w) - \hat \theta_t(\lambda,w) \right)^2.
\end{align*}

\begin{lemma}
    \label{lemma: regret approx}
    For any estimator $\hat w \in \Delta_{K-1}$ and for each $P \in \mathcal{P}$,
    \begin{align*}
        \left|\mathsf{Regret}_{1,P}(\hat w) - \mathbf{E}_P\left[\mathsf{MER}_1(\hat w)\right] \right| \le 2 \mathbf{E}_P\left[ \sup_{w \in \Delta_{K-1}} \left| R_{n,P}(w) \right|\right].
    \end{align*}
\end{lemma}

\noindent \textbf{Proof: } Since $e_t(\lambda,w) = \theta_t(\lambda,w) - \theta_t^*$ by Lemma \ref{lemm:theta*}, we write 
\begin{align*}
    \mathbf{E}_P\left[ \left( \theta_t^* - \hat \theta_t(\lambda,\hat w) \right)^2 \right]
    &= \mathbf{E}_P\left[ \left( \theta_t^* - \theta_t(\lambda,\hat w) + \theta_t(\lambda,\hat w) - \hat \theta_t(\lambda,\hat w) \right)^2 \right]\\
    &= \mathbf{E}_P\left[ e_t^2(\lambda,\hat w) \right] - 2 \mathbf{E}_P\left[ e_t(\lambda,\hat w) (\theta_t(\lambda,\hat w) - \hat \theta_t(\lambda,\hat w)) \right]\\
    &\quad \quad + \mathbf{E}_P\left[ (\theta_t(\lambda,\hat w) - \hat \theta_t(\lambda,\hat w))^2 \right].
\end{align*}
Hence, \begin{align*}
    \mathsf{Regret}_{1,P}(\hat w) &= \mathbf{E}_P\left[ \eta_{1,P}(\hat w) + R_{n,P}(\hat w) \right] - \inf_{\widetilde w \in \mathcal{D}} \mathbf{E}_P\left[ \eta_{1,P}(\widetilde w(Z)) + R_{n,P}(\widetilde w(Z)) \right]\\
    &= \mathbf{E}_P\left[ \mathsf{MER}_1(\hat w) \right] + \mathbf{E}_P\left[ R_{n,P}(\hat w) \right]\\
    &\quad \quad - \left\{\inf_{\widetilde w \in \mathcal{D}} \left( \mathbf{E}_P[\eta_{1,P}(\widetilde w(Z))] + \mathbf{E}_P\left[ R_{n,P}(\widetilde w(Z)) \right]\right) - \inf_{\widetilde w \in \mathcal{D}} \mathbf{E}_P[\eta_{1,P}(\widetilde w(Z))]\right\}.
\end{align*}
As for the last term, 
\begin{align*}
    &\left|\inf_{\widetilde w \in \mathcal{D}} \left( \mathbf{E}_P[\eta_{1,P}(\widetilde w(Z))] + \mathbf{E}_P\left[ R_{n,P}(\widetilde w(Z)) \right]\right) - \inf_{\widetilde w \in \mathcal{D}} \mathbf{E}_P[\eta_{1,P}(\widetilde w(Z))]\right| \\
    &\le \inf_{\widetilde w \in \mathcal{D}} \left( \mathbf{E}_P[\eta_{1,P}(\widetilde w(Z))] + \mathbf{E}_P\left[ \sup_{w \in \Delta_{K-1}} \left| R_{n,P}(w) \right|\right]\right) - \inf_{\widetilde w \in \mathcal{D}} \mathbf{E}_P[\eta_{1,P}(\widetilde w(Z))] = \mathbf{E}_P\left[ \sup_{w \in \Delta_{K-1}} \left| R_{n,P}(w) \right|\right].
\end{align*}
Thus, we obtain the desired bound. $\blacksquare$\medskip

We define 
\begin{align*}
    \hat D_{d,1} = \frac{1}{|\mathcal{T}_d|} \sum_{t \in \mathcal{T}_d} \sum_{j=0}^K \left| \hat \varepsilon_{j,t} \right|
    \text{ and }
    \hat D_{d,2}^2 = \frac{1}{|\mathcal{T}_d|} \sum_{t \in \mathcal{T}_d} \sum_{j=0}^K \hat \varepsilon_{j,t}^2.
\end{align*}
We define similarly $\hat D_{d,1}(0)$ and $\hat D_{d,2}^2(0)$ with $\hat \mu_{j,t}(\lambda)$ and $\mu_{j,t}(\lambda)$ replaced by $\hat \mu_{j,t}(0;\lambda)$ and $\mu_{j,t}(0;\lambda)$.

\begin{lemma}
    \label{lemma: Rn}

    For each $P \in \mathcal{P}$ and $w \in \Delta_{K-1}$, the following statements hold.\medskip

    (i) $\left| R_{n,P}(w) \right| \le 8 \overline m \hat D_{1,1} + 2 \hat D_{1,2}^2.$\medskip

    (ii) $\left| \hat \eta_{0}(w) - \eta_{0,P}(w) \right| \le 8 \overline m \hat D_{0,1}(0) + 2 \hat D_{0,2}^2(0).$
\end{lemma}

\noindent \textbf{Proof: } (i) First, by (\ref{bounds}), we have $\left| e_t(\lambda,w) \right| \le 4 \overline m$, for all $w \in \Delta_{K-1}$. Hence, 
\begin{align*}
    |R_{n,P}(\hat w)| \le \frac{8 \overline m}{|\mathcal{T}_1|} \sum_{t \in \mathcal{T}_1} \left| \theta_t(\lambda,\hat w) - \hat \theta_t(\lambda,\hat w)\right| + \frac{1}{|\mathcal{T}_1|} \sum_{t \in \mathcal{T}_1} \left( \theta_t(\lambda,\hat w) - \hat \theta_t(\lambda,\hat w)\right)^2.
\end{align*}
Note that 
\begin{align*}
    \left| \theta_t(\lambda,\hat w) - \hat \theta_t(\lambda,\hat w)\right|
    \le \sum_{j=0}^K \left| \hat \varepsilon_{j,t} \right|.
\end{align*}
Also, note that 
\begin{align*}
    \left( \theta_t(\lambda,\hat w) - \hat \theta_t(\lambda,\hat w)\right)^2
    &\le 2 \hat \varepsilon_{0,t}^2 + 2 \left(\sum_{j=1}^K \hat \varepsilon_{j,t} \hat w_j \right)^2\\
    &\le 2 \hat \varepsilon_{0,t}^2 + 2 \sum_{j=1}^K \hat \varepsilon_{j,t}^2 \hat w_j \le 2 \sum_{j=0}^K \hat \varepsilon_{j,t}^2.
\end{align*}
The second inequality follows from Jensen's inequality. Combining these, we obtain the desired result.\medskip

(ii) Note that 
\begin{align*}
    \left| \hat \eta_{0}(w) - \eta_{0,P}(w) \right| &\le \frac{1}{|\mathcal{T}_0|} \sum_{t \in \mathcal{T}_0} \left| \hat e_t^2(\lambda,w) -  e_t^2(\lambda,w) \right|\\
    &\le \frac{1}{|\mathcal{T}_0|} \sum_{t \in \mathcal{T}_0} \left( \hat e_t(\lambda,w) - e_t(\lambda,w) \right)^2\\
     &\quad \quad + \frac{2}{|\mathcal{T}_0|} \sum_{t \in \mathcal{T}_0} |e_t(\lambda,w)| \left| \hat e_t(\lambda,w) -  e_t(\lambda,w) \right|.
\end{align*}
Now, observe that 
\begin{align*}
    \left| \hat e_t(\lambda,w) -  e_t(\lambda,w) \right| \le \sum_{j=0}^K \left|\hat \varepsilon_{j,t}(0) \right|,
\end{align*}
and
\begin{align*}
    \left( \hat e_t(\lambda,w) -  e_t(\lambda,w) \right)^2 &\le 2 \hat \varepsilon_{0,t}^2(0) + 2 \sum_{j=1}^K \hat \varepsilon_{j,t}^2(0)  \le 2\sum_{j=0}^K \hat \varepsilon_{j,t}^2(0).
\end{align*}
Therefore, $\left| \hat \eta_{0}(w) - \eta_{0,P}(w) \right| \le 8\overline m \hat D_{0,1}(0) + 2 \hat D_{0,2}^2(0).$ $\blacksquare$

\begin{lemma}
 \label{lemm:eta D bound}
 For each $P \in \mathcal{P}$, 
 \begin{align*}
    &\hat \eta_0(w_P^{\mathsf{SCD}}) - \eta_{0,P}(w_P^{\mathsf{SCD}}) - \left( \hat \eta_0(\hat w^{\mathsf{SCD}}) - \eta_{0,P}(\hat w^{\mathsf{SCD}}) \right)\\
    &\le 4(\hat D_{0,2}^2 + 4\overline m \hat D_{0,1})  \sum_{j=1}^K |\hat w_j^{\mathsf{SCD}} - w_{j,P}^{\mathsf{SCD}}|.
 \end{align*}
\end{lemma}

\noindent \textbf{Proof: } We first write $\hat \eta_0(w_P^{\mathsf{SCD}}) - \eta_{0,P}(w_P^{\mathsf{SCD}}) - \left( \hat \eta_0(\hat w^{\mathsf{SCD}}) - \eta_{0,P}(\hat w^{\mathsf{SCD}}) \right)$ as
\begin{align*}
    &\frac{1}{|\mathcal{T}_0|} \sum_{t \in \mathcal{T}_0} (\hat e_t(\lambda,w_P^{\mathsf{SCD}}) + \hat e_t(\lambda,\hat w^{\mathsf{SCD}}))(\hat e_t(\lambda,w_P^{\mathsf{SCD}}) - \hat e_t(\lambda,\hat w^{\mathsf{SCD}}))\\
    &\quad - \frac{1}{|\mathcal{T}_0|} \sum_{t \in \mathcal{T}_0} (e_t(\lambda,w_P^{\mathsf{SCD}}) + e_t(\lambda,\hat w^{\mathsf{SCD}}))(e_t(\lambda,w_P^{\mathsf{SCD}}) - e_t(\lambda,\hat w^{\mathsf{SCD}}))\\
    &= \frac{1}{|\mathcal{T}_0|} \sum_{t \in \mathcal{T}_0} (\hat e_t(\lambda,w_P^{\mathsf{SCD}}) + \hat e_t(\lambda,\hat w^{\mathsf{SCD}}) - (e_t(\lambda,w_P^{\mathsf{SCD}}) + e_t(\lambda,\hat w^{\mathsf{SCD}})))(\hat e_t(\lambda,w_P^{\mathsf{SCD}}) - \hat e_t(\lambda,\hat w^{\mathsf{SCD}}))\\
    &\quad + \frac{1}{|\mathcal{T}_0|} \sum_{t \in \mathcal{T}_0} (e_t(\lambda,w_P^{\mathsf{SCD}}) + e_t(\lambda,\hat w^{\mathsf{SCD}}))(\hat e_t(\lambda,w_P^{\mathsf{SCD}}) - \hat e_t(\lambda,\hat w^{\mathsf{SCD}}) - (e_t(\lambda,w_P^{\mathsf{SCD}}) - e_t(\lambda,\hat w^{\mathsf{SCD}}))).
\end{align*}
By rearranging terms, we can write the sum of the last sums as 
\begin{align*}
    &\frac{1}{|\mathcal{T}_0|} \sum_{t \in \mathcal{T}_0} \left\{ \left( 2 \hat \varepsilon_{0,t} - \sum_{j=1}^K \hat \varepsilon_{j,t} (\hat w_j^{\mathsf{SCD}} + w_{j,P}^{\mathsf{SCD}}) \right) \times \sum_{j=1}^K \hat \mu_{j,t}(\lambda)(\hat w_j^{\mathsf{SCD}} - w_{j,P}^{\mathsf{SCD}}) \right.\\
    &\quad \quad + \left. \left( 2 \mu_{0,t}(\lambda) - \sum_{j=1}^K \mu_{j,t}(\lambda)(\hat w_j^{\mathsf{SCD}} + w_{j,P}^{\mathsf{SCD}}) \right)\times \sum_{j=1}^K  \hat \varepsilon_{j,t} (\hat w_j^{\mathsf{SCD}} - w_{j,P}^{\mathsf{SCD}}) \right\}\\
    &=\frac{1}{|\mathcal{T}_0|} \sum_{t \in \mathcal{T}_0} \left\{ \left( 2 \hat \varepsilon_{0,t} - \sum_{j=1}^K \hat \varepsilon_{j,t} (\hat w_j^{\mathsf{SCD}} + w_{j,P}^{\mathsf{SCD}}) \right) \times \sum_{j=1}^K \hat \varepsilon_{j,t}(\hat w_j^{\mathsf{SCD}} - w_{j,P}^{\mathsf{SCD}}) \right.\\
    &\quad\quad + \left. \left( 2 \hat \varepsilon_{0,t} - \sum_{j=1}^K \hat \varepsilon_{j,t} (\hat w_j^{\mathsf{SCD}} + w_{j,P}^{\mathsf{SCD}}) \right) \times \sum_{j=1}^K \mu_{j,t}(\lambda) (\hat w_j^{\mathsf{SCD}} - w_{j,P}^{\mathsf{SCD}}) \right.\\
    &\quad \quad + \left. \left( 2 \mu_{0,t}(\lambda) - \sum_{j=1}^K \mu_{j,t}(\lambda)(\hat w_j^{\mathsf{SCD}} + w_{j,P}^{\mathsf{SCD}}) \right)\times \sum_{j=1}^K  \hat \varepsilon_{j,t} (\hat w_j^{\mathsf{SCD}} - w_{j,P}^{\mathsf{SCD}}) \right\}\\
    &\le \frac{1}{|\mathcal{T}_0|} \sum_{t \in \mathcal{T}_0} \left\{ 4 \max_{0 \le j \le K}|\hat \varepsilon_{j,t}|^2  + 16 \overline m \max_{1 \le j \le K}|\hat \varepsilon_{j,t}| \right\} \times \sum_{j=1}^K|\hat w_j^{\mathsf{SCD}} - w_{j,P}^{\mathsf{SCD}}|.   
\end{align*}
Since 
\begin{align*}
    \frac{1}{|\mathcal{T}_0|} \sum_{t \in \mathcal{T}_0} \max_{1 \le j \le K}|\hat \varepsilon_{j,t}| \le \hat D_{0,1} \text{ and } \frac{1}{|\mathcal{T}_0|} \sum_{t \in \mathcal{T}_0} \max_{1 \le j \le K}|\hat \varepsilon_{j,t}|^2 \le \hat D_{0,2}^2,
\end{align*}
we obtain the desired result. $\blacksquare$

\begin{lemma}
  \label{lemm: w bound}
  For each $P \in \mathcal{P}$,
  \begin{align*}
    \sum_{j=1}^K | \hat w_j^{\mathsf{SCD}} - w_{j,P}^{\mathsf{SCD}} | \le \frac{4 K(\hat D_{0,2}^2 + 4 \overline m \hat D_{0,1})}{c}.
  \end{align*}
\end{lemma}

\noindent \textbf{Proof: } Note that 
\begin{align*}
    &\hat \eta_0(w_P^{\mathsf{SCD}}) - \eta_{0,P}(w_P^{\mathsf{SCD}}) - \left( \hat \eta_0(\hat w^{\mathsf{SCD}}) - \eta_{0,P}(\hat w^{\mathsf{SCD}}) \right) \\
    &\quad \quad \ge \eta_{0,P}(\hat w^{\mathsf{SCD}}) - \eta_{0,P}(w_P^{\mathsf{SCD}}) \ge c \| \hat w^{\mathsf{SCD}} - w_P^{\mathsf{SCD}}\|^2 \ge \frac{c}{K} \left( \sum_{j=1}^K | \hat w_j^{\mathsf{SCD}} - w_{j,P}^{\mathsf{SCD}} | \right)^2,
\end{align*}
by Lemma \ref{lemma: strong ID SCD}. From Lemma \ref{lemm:eta D bound}, we obtain the desired result. $\blacksquare$

\begin{lemma}
  \label{lemm:eta bound}
  For each $P \in \mathcal{P}$,
  \begin{align*}
    \left| \eta_{1,P}(\hat w^{\mathsf{SCD}}) - \eta_{1,P}(w_P^{\mathsf{SCD}}) \right| 
    \le \frac{64 \overline m^2 K}{c} (\hat D_{0,2}^2 + 4 \overline m \hat D_{0,1}).
  \end{align*}
\end{lemma}

\noindent \textbf{Proof: } First, note that $\eta_{1,P}(\hat w^{\mathsf{SCD}}) - \eta_{1,P}(w_P^{\mathsf{SCD}})$ is equal to
\begin{align*}
    &\frac{1}{|\mathcal{T}_1|}\sum_{t \in \mathcal{T}_1}(e_t^2(\lambda, \hat w^{\mathsf{SCD}}) - e_t^2(\lambda,w_P^{\mathsf{SCD}}))\\
    &= \frac{1}{|\mathcal{T}_1|}\sum_{t \in \mathcal{T}_1}\left(2 \mu_{0,t}(0;\lambda) - \sum_{j=1}^K \mu_{j,t}(0;\lambda)(\hat w_j^{\mathsf{SCD}} + w_{j,P}^{\mathsf{SCD}}) \right) \left( \sum_{j=1}^K \mu_{j,t}(0;\lambda)\left( w_{j,P}^{\mathsf{SCD}} - \hat w_j^{\mathsf{SCD}} \right) \right).
\end{align*}
Hence, by Assumption \ref{assump: moments},
\begin{align*}
    \left| \eta_{1,P}(\hat w^{\mathsf{SCD}}) - \eta_{1,P}(w_P^{\mathsf{SCD}}) \right| \le 16 \overline m^2 \cdot \sum_{j=1}^K \left| w_{j,P}^{\mathsf{SCD}} - \hat w_j^{\mathsf{SCD}} \right|.
\end{align*}
The desired result follows by Lemma \ref{lemm: w bound}. $\blacksquare$

\begin{lemma}
    \label{lemm:Regret Approx SCD}
    There exists a universal constant $C>0$ such that for each $P \in \mathcal{P}$
    \begin{align*}
        &\left| \mathsf{Regret}_{1,P}(\hat w^{\mathsf{SCD}}) - \mathsf{MER}_{1,P}(w_P^{\mathsf{SCD}}) \right| \\
        &\le C \overline m \mathbf{E}_P\left[ \hat D_{1,1}\right] + C \mathbf{E}_P\left[ \hat D_{1,2}^2\right] + \frac{C \overline m^2 K }{c} \left( \mathbf{E}_P\left[\hat D_{0,2}^2\right] + C \overline m \mathbf{E}_P\left[\hat D_{0,1} \right]\right) \\
        &\quad \quad + C \overline m \mathbf{E}_P\left[\hat D_{0,1}(0)\right] + C \mathbf{E}_P\left[\hat D_{0,2}^2(0) \right].
    \end{align*}
\end{lemma}

\noindent \textbf{Proof: } Note that 
\begin{align}
    \label{ineq243}
    \left| \mathsf{Regret}_{1,P}(\hat w^{\mathsf{SCD}}) - \mathsf{MER}_{1,P}(w_P^{\mathsf{SCD}}) \right|
    &\le \left| \mathsf{Regret}_{1,P}(\hat w^{\mathsf{SCD}}) - \mathbf{E}_P\left[\mathsf{MER}_{1,P}(\hat w^{\mathsf{SCD}})\right] \right|\\ \notag
    &\quad \quad +\left| \mathbf{E}_P\left[\mathsf{MER}_{1,P}(\hat w^{\mathsf{SCD}})\right] - \mathsf{MER}_{1,P}(w_P^{\mathsf{SCD}}) \right|.
\end{align}
As for the leading term on the right-hand side, by Lemmas \ref{lemma: regret approx} and \ref{lemma: Rn}(i), 
\begin{align*}
    \left| \mathsf{Regret}_{1,P}(\hat w^{\mathsf{SCD}}) - \mathbf{E}_P\left[\mathsf{MER}_{1,P}(\hat w^{\mathsf{SCD}})\right] \right| &\le 2 \mathbf{E}_P\left[ \sup_{w \in \Delta_{K-1}} \left| R_{n,P}(w) \right|\right]\\
    &\le 2\left( 8 \overline m \mathbf{E}_P\left[ \hat D_{1,1}\right] + 2 \mathbf{E}_P\left[ \hat D_{1,2}^2\right] \right).
\end{align*}
It remains to deal with the last term in (\ref{ineq243}). First, we focus on $\mathsf{MER}_{1,P}(\hat w^{\mathsf{SCD}})$. Define 
\begin{align}
    \label{Regret+}
    \mathsf{Regret}_{0,P}^{+}(\hat w^{\mathsf{SCD}}) = \mathbf{E}_P\left[ \eta_{0,P}(\hat w^{\mathsf{SCD}})\right] - \inf_{\widetilde w \in \mathcal{D}} \mathbf{E}_P\left[ \hat \eta_{0}(\widetilde w(Z))\right].
\end{align}
We write 
\begin{align}
    \label{der2}
    \mathsf{MER}_{1,P}(\hat w^{\mathsf{SCD}}) &= \mathsf{MER}_{1,P}(\hat w^{\mathsf{SCD}}) - \mathsf{MER}_{0,P}(\hat w^{\mathsf{SCD}}) + \mathsf{MER}_{0,P}(\hat w^{\mathsf{SCD}})\\ \notag
    &\quad \quad - \mathsf{Regret}_{0,P}^{+}(\hat w^{\mathsf{SCD}}) + \mathsf{Regret}_{0,P}^{+}(\hat w^{\mathsf{SCD}}).
\end{align} 
We look into $\mathsf{Regret}_{0,P}^{+}(\hat w^{\mathsf{SCD}})$. Note that
\begin{align}
    \label{dev0}
    \mathsf{Regret}_{0,P}^{+}(\hat w^{\mathsf{SCD}}) &\le \mathbf{E}_P\left[ \hat \eta_{0}(\hat w^{\mathsf{SCD}})\right] - \inf_{\widetilde w \in \mathcal{D}} \mathbf{E}_P\left[ \hat \eta_{0}(\widetilde w(Z))\right] + \mathbf{E}_P\left[\sup_{w \in \Delta_{K-1}} \left| \hat \eta_{0}(w) - \eta_{0,P}(w) \right| \right]\\ \notag
    &= \mathbf{E}_P\left[ \inf_{w \in \Delta_{K-1}} \hat \eta_{0}(w) \right] - \inf_{\widetilde w \in \mathcal{D}} \mathbf{E}_P\left[ \hat \eta_{0}(\widetilde w(Z))\right] + \mathbf{E}_P\left[\sup_{w \in \Delta_{K-1}} \left| \hat \eta_{0}(w) - \eta_{0,P}(w) \right| \right]\\ \notag
    &=\mathbf{E}_P\left[\sup_{w \in \Delta_{K-1}} \left| \hat \eta_{0}(w) - \eta_{0,P}(w) \right| \right],
\end{align}
where the first equality uses the definition of $\hat w^{\mathsf{SCD}}$ and the second equality follows by Lemma \ref{inf in and out}. Since 
\begin{align*}
    \mathbf{E}_P\left[\mathsf{MER}_{0,P}(\hat w^{\mathsf{SCD}})\right] = \mathbf{E}_P\left[ \eta_{0,P}(\hat w^{\mathsf{SCD}}) \right] - \inf_{\tilde w \in \mathcal{D}} \mathbf{E}_P[\eta_{0,P}(\tilde w(Z))],
\end{align*}
by the definition of $\mathsf{Regret}_{0,P}^{+}(\hat w^{\mathsf{SCD}})$ in (\ref{Regret+}), we also have 
\begin{align}
    \label{dev1}
    \mathbf{E}_P\left[\mathsf{MER}_{0,P}(\hat w^{\mathsf{SCD}})\right] - \mathsf{Regret}_{0,P}^{+}(\hat w^{\mathsf{SCD}}) &= \inf_{\widetilde w \in \mathcal{D}} \mathbf{E}_P\left[ \hat \eta_{0}(\widetilde w(Z))\right] - \inf_{\tilde w \in \mathcal{D}} \mathbf{E}_P[\eta_{0,P}(\tilde w(Z))]\\ \notag
    &=\mathbf{E}_P\left[ \inf_{w \in \Delta_{K-1}} \hat \eta_{0}(w)\right] - \inf_{w \in \Delta_{K-1}} \eta_{0,P}(w)\\ \notag
    &\le \mathbf{E}_P\left[\sup_{w \in \Delta_{K-1}} \left| \hat \eta_{0}(w) - \eta_{0,P}(w) \right| \right],
\end{align}
where the second equality follows from Lemma \ref{inf in and out}. Therefore, by (\ref{dev0}) and (\ref{dev1}),
\begin{align*}
    \left| \mathbf{E}_P\left[ \mathsf{MER}_{0,P}(\hat w^{\mathsf{SCD}}) \right] \right| \le 2\mathbf{E}_P\left[\sup_{w \in \Delta_{K-1}} \left| \hat \eta_{0}(w) - \eta_{0,P}(w) \right| \right].
\end{align*}
Thus, as for the last term (\ref{ineq243}), noting that $\mathsf{MER}_{0,P}(w_P^{\mathsf{SCD}}) = 0$,
\begin{align}
    \label{ineq213}
    &\left| \mathbf{E}_P\left[\mathsf{MER}_{1,P}(\hat w^{\mathsf{SCD}})\right] - \mathsf{MER}_{1,P}(w_P^{\mathsf{SCD}}) \right| \\ \notag
    &\le \left| \mathbf{E}_P\left[\mathsf{MER}_{1,P}(\hat w^{\mathsf{SCD}}) - \mathsf{MER}_{0,P}(\hat w^{\mathsf{SCD}})\right] - (\mathsf{MER}_{1,P}(w_P^{\mathsf{SCD}}) - \mathsf{MER}_{0,P}(w_P^{\mathsf{SCD}})) \right| \\ \notag
    &\quad \quad + 2\mathbf{E}_P\left[\sup_{w \in \Delta_{K-1}} \left| \hat \eta_{0}(w) - \eta_{0,P}(w) \right| \right].
\end{align}
Now note that 
\begin{align*}
    &\mathsf{MER}_{1,P}(\hat w^{\mathsf{SCD}}) - \mathsf{MER}_{0,P}(\hat w^{\mathsf{SCD}}) - (\mathsf{MER}_{1,P}(w_P^{\mathsf{SCD}}) - \mathsf{MER}_{0,P}(w_P^{\mathsf{SCD}}))\\
    &= \eta_{1,P}(\hat w^{\mathsf{SCD}}) - \eta_{0,P}(\hat w^{\mathsf{SCD}}) - (\eta_{1,P}(w_P^{\mathsf{SCD}}) - \eta_{0,P}(w_P^{\mathsf{SCD}})).
\end{align*}
Hence, 
\begin{align*}
    &\left| \mathbf{E}_P\left[\mathsf{MER}_{1,P}(\hat w^{\mathsf{SCD}}) - \mathsf{MER}_{0,P}(\hat w^{\mathsf{SCD}})\right] - (\mathsf{MER}_{1,P}(w_P^{\mathsf{SCD}}) - \mathsf{MER}_{0,P}(w_P^{\mathsf{SCD}})) \right| \\
    &\quad \le \left| \mathbf{E}_P\left[ \eta_{1,P}(\hat w^{\mathsf{SCD}}) - \eta_{1,P}(w_P^{\mathsf{SCD}})\right] - \mathbf{E}_P\left[ \eta_{0,P}(\hat w^{\mathsf{SCD}}) - \eta_{0,P}(w_P^{\mathsf{SCD}})\right] \right|.
\end{align*}
Note that
\begin{align*}
    0 &\le \eta_{0,P}(\hat w^{\mathsf{SCD}}) - \eta_{0,P}(w_P^{\mathsf{SCD}})\\
      &= \eta_{0,P}(\hat w^{\mathsf{SCD}}) - \hat \eta_{0}(\hat w^{\mathsf{SCD}}) + \hat \eta_{0}(\hat w^{\mathsf{SCD}}) - \eta_{0,P}(w_P^{\mathsf{SCD}})\\
      &\le \sup_{w \in \Delta_{K-1}}\left| \eta_{0,P}(w) - \hat \eta_{0}(w)\right| + \hat \eta_{0}(w_P^{\mathsf{SCD}}) - \eta_{0,P}(w_P^{\mathsf{SCD}}) \le 2 \sup_{w \in \Delta_{K-1}}\left| \eta_{0,P}(w) - \hat \eta_{0}(w)\right|.
\end{align*}
Combining this with Lemma \ref{lemm:eta bound}, we find that 
\begin{align*}
    &\left| \mathbf{E}_P\left[\mathsf{MER}_{1,P}(\hat w^{\mathsf{SCD}}) - \mathsf{MER}_{0,P}(\hat w^{\mathsf{SCD}})\right] - (\mathsf{MER}_{1,P}(w_P^{\mathsf{SCD}}) - \mathsf{MER}_{0,P}(w_P^{\mathsf{SCD}})) \right| \\
    &\le \frac{64 \overline m^2 K}{c} (\mathbf{E}_P\left[ \hat D_{0,2}^2 \right] + 4 \overline m \mathbf{E}_P\left[ \hat D_{0,1} \right]) + 2 \mathbf{E}_P\left[ \sup_{w \in \Delta_{K-1}}\left| \eta_{0,P}(w) - \hat \eta_{0}(w)\right| \right]\\
    &\le \frac{64 \overline m^2 K}{c} ( \mathbf{E}_P\left[\hat D_{0,2}^2\right] + 4 \overline m \mathbf{E}_P\left[\hat D_{0,1}\right]) + 16 \overline m \mathbf{E}_P\left[\hat D_{0,1}(0)\right] + 4 \mathbf{E}_P\left[\hat D_{0,2}^2(0) \right].
\end{align*}
The last inequality follows by Lemma \ref{lemma: Rn}. In light of (\ref{ineq213}), this yields the desired result. $\blacksquare$

 \begin{lemma}
    \label{lemm: D1}
    There exists a universal constant $C>0$ such that for $d=0,1$ and each $P \in \mathcal{P}$,
 \begin{align*}
    \max\left\{\mathbf{E}_P\left[ \hat D_{d,1} \right],\mathbf{E}_P\left[ \hat D_{d,1}(0) \right]\right\} &\le C \overline m \sum_{j=0}^K \inf_{v >0} \left\{ \frac{1}{v\sqrt{n}} + \exp\left(- \frac{n (\tilde p_j -v)^2}{2 \tilde p_j} \right) \right\}, \text{ and }\\
    \max\left\{\mathbf{E}_P\left[ \hat D_{d,2}^2 \right],\mathbf{E}_P\left[ \hat D_{d,2}^2(0) \right]\right\} &\le C \overline m^2 \sum_{j=0}^K \inf_{v >0} \left\{ \frac{1}{v^2 n} + \exp\left(- \frac{n (\tilde p_j -v)^2}{2 \tilde p_j} \right) \right\},
 \end{align*}
 where $\tilde p_j = \frac{1}{n}\sum_{i \in N} P\left\{G_i = j \right\}.$
 \end{lemma}
 
\noindent \textbf{Proof: } We focus on $\hat D_{d,1}$. Note that 
\begin{align*}
    \mathbf{E}_P\left[ \left| \hat m_{j,t} - m_{j,t} \right| \right] =  A_{n,1} + A_{n,2},
\end{align*}
where, with $\hat p_j = \frac{1}{n}\sum_{i \in N} 1\left\{G_i = j \right\}$, we define
\begin{align*}
    A_{n,1} &= \mathbf{E}_P\left[ \left| \hat m_{j,t} - m_{j,t} \right| 1\{\hat p_j \ge v\}\right], \text{ and }\\
    A_{n,2} &= \mathbf{E}_P\left[ \left| \hat m_{j,t} - m_{j,t} \right| 1\{\hat p_j < v\}\right].
\end{align*}
We define $\varepsilon_{ij,t} = Y_{i,t} - \mathbf{E}_P[Y_{i,t} \mid G_i = j]$, and, using $\hat m_{j,t} - m_{j,t} = \frac{1}{n}\sum_{i \in N} (Y_{i,t} - \mu_{j,t})1\{G_i = j\}/\hat p_{j}$, bound 
\begin{align*}
    A_{n,1} \le \frac{1}{v}\mathbf{E}_P\left[ \left| \frac{1}{n}\sum_{i \in N} \varepsilon_{ij,t} 1\{G_i = j\} \right| \right] \le \frac{\overline m}{v\sqrt{n}},
\end{align*}
where the last inequality follows because 
\begin{align*}
    \mathbf{E}_P\left[ \left( \frac{1}{n}\sum_{i \in N} \varepsilon_{ij,t} 1\{G_i = j\} \right)^2 \right] &\le \frac{1}{n^2}\sum_{i \in N}  \mathbf{E}_P\left[ \varepsilon_{ij,t}^2 1\{G_i = j\} \right] \\
    &\le \frac{1}{n^2}\sum_{i \in N}  \mathbf{E}_P\left[ Y_{i,t}^2 1\{G_i = j\} \right] \le \frac{\overline m^2}{n}.
\end{align*}
(Note that $\varepsilon_{ij,t}1\{G_i = j\}$ have mean zero and are independent across $i$'s by Assumption \ref{assump: sampling design}.) Let us turn to $A_{n,2}$. Note that \begin{align*}
    A_{n,2} &\le \mathbf{E}_P\left[  \frac{\sum_{i \in N} \left| \varepsilon_{ij,t} \right| 1\{G_i = j\}}{\sum_{i \in N} 1\{G_i = j\}} 1\{\hat p_j < v\} \right] \\
    &\le \mathbf{E}_P\left[  \frac{\sum_{i \in N} \mathbf{E}_P\left[\left| \varepsilon_{ij,t} \right| \mid G_1,...,G_n\right] 1\{G_i = j\}}{\sum_{i \in N} 1\{G_i = j\}} 1\{\hat p_j < v\} \right] \\
    &\le \mathbf{E}_P\left[  \frac{\sum_{i \in N} \mathbf{E}_P\left[\left| \varepsilon_{ij,t} \right| \mid G_i\right] 1\{G_i = j\}}{\sum_{i \in N} 1\{G_i = j\}} 1\{\hat p_j < v\} \right] \le \overline m P\left\{ \hat p_j < v \right\},
\end{align*}
because $\mathbf{E}_P[\varepsilon_{ij,t}^2 \mid G_i] \le \overline m^2.$ Now, note that by Chernoff's bound (see, e.g., Lemma 2.1 of \cite{Chung/Lu:AC:02}), 
\begin{align*}
    P\left\{ \hat p_j < v \right\} &= P\left\{ \sum_{i \in N} 1\{G_i = j\} - n \tilde p_{j} < n v - n \tilde p_{j} \right\}\\
    &\le \exp\left( - \frac{ n( \tilde p_{j} - v)^2}{2 \tilde p_j}\right).
\end{align*}
Therefore, we have 
\begin{align*}
    \mathbf{E}_P\left[ \left| \hat m_{j,t} - m_{j,t} \right| \right] \le \overline m \inf_{v >0} \left( \frac{1}{v\sqrt{n}} + \exp\left( - \frac{ n( \tilde p_{j} - v)^2}{2 \tilde p_j}\right) \right).
\end{align*}
Hence, 
\begin{align*}
    \mathbf{E}_P\left[ \left| \hat \mu_{j,t}(\lambda) - \mu_{j,t}(\lambda) \right| \right] \le 2 \overline m \inf_{v >0} \left( \frac{1}{v\sqrt{n}} + \exp\left( - \frac{ n( \tilde p_{j} - v)^2}{2 \tilde p_j}\right) \right).
\end{align*}
Using the same arguments, we obtain the same bound for $\mathbf{E}_P\left[ \left| \hat \mu_{j,t}(0;\lambda) - \mu_{j,t}(0;\lambda) \right| \right]$. This gives the first bound of the lemma.

As for the second bound, we consider 
\begin{align*}
    \mathbf{E}_P\left[ \left( \hat m_{j,t} - m_{j,t} \right)^2 \right] =  B_{n,t}^a + B_{n,t}^b,
\end{align*}
where
\begin{align*}
    B_{n,t}^a &= \mathbf{E}_P\left[ \left( \hat m_{j,t} - m_{j,t} \right)^2 1\{\hat p_j \ge v\}\right], \text{ and }\\
    B_{n,t}^b &= \mathbf{E}_P\left[ \left( \hat m_{j,t} - m_{j,t} \right)^2 1\{\hat p_j < v\}\right].
\end{align*}
Similarly as before, note that 
\begin{align*}
    B_{n,t}^a \le \frac{1}{v^2 n^2} \sum_{i \in N} \mathbf{E}_P\left[ \varepsilon_{ij,t}^2 1\{G_i = j\} \right] \le \frac{\overline m^2}{v^2 n},
\end{align*}
and 
\begin{align*}
    B_{n,t}^b \le 4 \overline m^2 P\{\hat p_j < v\} \le 4 \overline m^2 \exp\left( - \frac{n(\tilde p_j - v)^2}{2 \tilde p_j}\right).
\end{align*}
Hence, \begin{align*}
    \mathbf{E}_P\left[ \left( \hat \mu_{j,t}(\lambda) - \mu_{j,t}(\lambda) \right)^2 \right] &\le 2 (B_{n,t}^a + B_{n,t}^b) + 2\sum_{s=1}^{T^*-1} \left( B_{n,s}^a + B_{n,s}^b \right) \lambda_s\\
    &\le 2 (B_{n,t}^a + B_{n,t}^b) + 2\max_{s=1,...,T^*-1} \left( B_{n,s}^a + B_{n,s}^b \right)\\
    &\le 16 \overline m^2 \left\{ \frac{1}{v^2 n} + \exp\left( - \frac{n(\tilde p_j - v)^2}{2 \tilde p_j}\right) \right\},
\end{align*}
where the first inequality follows by Jensen's inequality. Using the same arguments, we obtain the same bound for $\mathbf{E}_P\left[ \left( \hat \mu_{j,t}(0;\lambda) - \mu_{j,t}(0;\lambda) \right)^2 \right]$. This gives the second bound of the lemma. $\blacksquare$

\begin{proposition}
    \label{prop: Regret Approx SCD}
    There exists a universal constant $C>0$ such that
    \begin{align*}
        \sup_{P \in \mathcal{P}} \left| \mathsf{Regret}_{1,P}(\hat w^{\mathsf{SCD}}) - \mathsf{\Delta MER}_{P}(w_P^{\mathsf{SCD}}) \right|
        &\le \frac{C (K+1) \overline m^4}{c}\left\{ \frac{1}{\pi_0 \sqrt{n}} + \frac{1}{\pi_0^2 n} + \exp\left( - \frac{\pi_0 n}{8} \right) \right\}.
    \end{align*}
\end{proposition}

\noindent \textbf{Proof: } Since we can take $v = \tilde p_j/2$ in the bound in Lemma \ref{lemm: D1}, we have 
\begin{align*}
    \inf_{v >0} \left\{ \frac{1}{v\sqrt{n}} + \exp\left(- \frac{n (\tilde p_j -v)^2}{2 \tilde p_j} \right) \right\} &\le \frac{2}{\tilde p_j \sqrt{n}} +\exp\left( - \frac{n \tilde p_j}{8}\right) \\
     &\le \frac{2}{\pi_0 \sqrt{n}} +\exp\left( - \frac{\pi_0 n}{8}\right) =: A_n,
\end{align*}
where the last inequality uses $\tilde p_j \ge \pi_0$. Similarly, we have 
\begin{align*}
    \inf_{v >0} \left\{ \frac{1}{v^2 n} + \exp\left(- \frac{n (\tilde p_j -v)^2}{2 \tilde p_j} \right) \right\} &\le \frac{4}{\tilde p_j^2 n} +\exp\left( - \frac{n \tilde p_j}{8}\right) \\
    &\le \frac{4}{\pi_0^2 n} +\exp\left( - \frac{\pi_0 n}{8}\right) =: B_n.
\end{align*} 
By Lemmas \ref{lemm:Regret Approx SCD} and \ref{lemm: D1} and the bounds above, 
\begin{align*}
        \sup_{P \in \mathcal{P}} \left| \mathsf{Regret}_{1,P}(\hat w^{\mathsf{SCD}}) - \mathsf{MER}_{1,P}(w_P^{\mathsf{SCD}}) \right|
        &\le C \left(\overline m + \frac{\overline m^3 K}{c} \right) \overline m \cdot A_n\\
        &\quad + C \left(1 + \frac{\overline m^2 K}{c} \right) \overline m^2 \cdot B_n,
\end{align*}
with some universal constant $C>0$. The desired result follows because $\mathsf{MER}_{0,P}(w_P^{\mathsf{SCD}}) = 0$ and $\overline m \ge 1$ and $c \le 1$. $\blacksquare$

\begin{lemma}
    \label{lemm: wj DID}
    Suppose that Assumption \ref{assump: sampling design} holds. Then,
    \begin{align*}
        \sup_{P \in \mathcal{P}} \max_{1 \le j \le K} \mathbf{E}_P\left[ \left| \hat w_j^{\mathsf{DID}} - w_{j,P}^{\mathsf{DID}}\right| \right] \le \frac{2}{\pi_0 \sqrt{n}} + \exp\left( - \frac{\pi_0 n}{8} \right).
    \end{align*}
\end{lemma}

\noindent \textbf{Proof: } Let $u_{ij} = 1\{G_i = j\} - P\{G_i = j \mid G_i \in \mathcal{G}_{\mathsf{don}}\}$. We write \begin{align*}
    \hat w_j^{\mathsf{DID}} - w_{j,P}^{\mathsf{DID}} = \frac{\sum_{i \in N} u_{ij} 1\{G_i \in \mathcal{G}_{\mathsf{don}}\}}{\sum_{i \in N} 1\{G_i \in \mathcal{G}_{\mathsf{don}}\}}.
\end{align*}
Hence, with $\hat p \eqdef \frac{1}{n} \sum_{i \in N} 1\{G_i \in \mathcal{G}_{\mathsf{don}}\}$, and arbitrarily chosen $v>0$, we have
\begin{align*}
    \mathbf{E}_P\left[\left| \hat w_j^{\mathsf{DID}} - w_{j,P}^{\mathsf{DID}} \right|\right] \le C_{n,1} +C_{n,2},
\end{align*}
where
\begin{align*}
    C_{n,1} &= \mathbf{E}_P\left[ \left| \frac{\sum_{i \in N} u_{ij} 1\{G_i \in \mathcal{G}_{\mathsf{don}}\}}{\sum_{i \in N} 1\{G_i \in \mathcal{G}_{\mathsf{don}}\}} \right| 1\{\hat p \ge v\}\right] \text{ and }\\
    C_{n,2} &= \mathbf{E}_P\left[ \left| \frac{\sum_{i \in N} u_{ij} 1\{G_i \in \mathcal{G}_{\mathsf{don}}\}}{\sum_{i \in N} 1\{G_i \in \mathcal{G}_{\mathsf{don}}\}} \right| 1\{\hat p < v\}\right].
\end{align*}
Using the same arguments as before, we find that 
\begin{align*}
    C_{n,1} \le \frac{1}{v \sqrt{n}},
\end{align*}
because $|u_{ij}| \le 1$, and that 
\begin{align*}
    C_{n,2} \le \exp\left( - \frac{n(p - v)^2}{2 p}\right),
\end{align*}
where $p \eqdef \frac{1}{n} \sum_{i \in N} P\{G_i \in \mathcal{G}_{\mathsf{don}}\}$. By taking $v = p/2$ and noting $p \ge \pi_0$, we find that
\begin{align*}
    \sup_{P \in \mathcal{P}} \mathbf{E}_P\left[\left| \hat w_j^{\mathsf{DID}} - w_{j,P}^{\mathsf{DID}} \right|\right] \le \frac{2}{\pi_0 \sqrt{n}} + \exp\left( - \frac{\pi_0 n}{8} \right).
\end{align*}
$\blacksquare$

 \begin{proposition}
    \label{prop: Regret Approx DID}
    There exists a universal constant $C>0$ such that
   \begin{align*}
    \sup_{P \in \mathcal{P}} \left| \mathsf{Regret}_{1,P}(\hat w^{\mathsf{DID}}) - \mathsf{MER}_{1,P}(w_P^{\mathsf{DID}})\right| \le C \overline m^2 (K+1)\left\{\frac{1}{\pi_0 \sqrt{n}} + \frac{1}{\pi_0^2 n} + \exp\left( - \frac{\pi_0 n}{8}\right) \right\}.
   \end{align*}
 \end{proposition}

 \noindent \textbf{Proof: } By Lemmas \ref{lemma: regret approx}, \ref{lemma: Rn}(i), and \ref{lemm: D1} with $v = \pi_0/2$ in the bound, we have
 \begin{align*}
    \sup_{P \in \mathcal{P}} \left| \mathsf{Regret}_{1,P}(\hat w^{\mathsf{DID}}) - \mathbf{E}_P\left[\mathsf{MER}_{1,P}(\hat w^{\mathsf{DID}})\right]\right| &\le 2 \sup_{P \in \mathcal{P}} \mathbf{E}_P\left[ \sup_{w \in \Delta_{K-1}} | R_{n,P}(w) | \right]  \\ \notag
    &\le C' \overline m^2 (K+1) \left\{\frac{1}{\pi_0 \sqrt{n}} + \frac{1}{\pi_0^2 n} + \exp\left( - \frac{\pi_0 n}{8}\right) \right\},
 \end{align*}
 for some universal constant $C'$. We bound $\mathbf{E}_P\left[\mathsf{MER}_{1,P}(\hat w^{\mathsf{DID}})\right] - \mathsf{MER}_{1,P}(w_P^{\mathsf{DID}})$ as follows:
 \begin{align}
    \label{terms}   
     &\mathbf{E}_P\left[ \eta_{1,P}(\hat w^{\mathsf{DID}}) - \eta_{1,P}(w_P^{\mathsf{DID}})\right]\\ \notag
     &\le \frac{1}{|\mathcal{T}_1|} \sum_{t \in \mathcal{T}_1} \mathbf{E}_P\left[ \left| e_t(\lambda,\hat w^{\mathsf{DID}}) + e_t(\lambda, w_P^{\mathsf{DID}}) \right|\left| e_t(\lambda,\hat w^{\mathsf{DID}}) - e_t(\lambda, w_P^{\mathsf{DID}}) \right|\right].
 \end{align}
 Note that 
 \begin{align*}
    \left| e_t(\lambda,\hat w^{\mathsf{DID}}) - e_t(\lambda, w_P^{\mathsf{DID}}) \right| &\le \sum_{j=1}^K|\mu_{j,t}(0;\lambda)| |\hat w_j^{\mathsf{DID}} - w_{j,P}^{\mathsf{DID}}| \\
    &\le 2 \overline m \sum_{j=1}^K |\hat w_j^{\mathsf{DID}} - w_{j,P}^{\mathsf{DID}}|. 
 \end{align*}
 On the other hand, 
 \begin{align*}
    \left| e_t(\lambda,\hat w^{\mathsf{DID}}) + e_t(\lambda, w_P^{\mathsf{DID}}) \right| \le 2|\mu_{0,t}(0;\lambda)| + \sum_{j=1}^K |\mu_{j,t}(0;\lambda)|(\hat w_j^{\mathsf{DID}} + w_{j,P}^{\mathsf{DID}}) \le 8 \overline m.
 \end{align*}
 Hence, the last term in (\ref{terms}) is bounded by 
 \begin{align*}
    16 \overline m^2 \sup_{P \in \mathcal{P}} \sum_{j=1}^K \mathbf{E}_P\left[ \left| \hat w_j^{\mathsf{DID}} - w_{j,P}^{\mathsf{DID}}\right| \right] 
    \le 16 \overline m^2 K \left\{\frac{2}{\pi_0 \sqrt{n}} + \exp\left( - \frac{\pi_0 n}{8}\right) \right\},
 \end{align*}
 by Lemma \ref{lemm: wj DID}. Thus, since $\overline m \ge 1$, we have the desired result. $\blacksquare$\medskip

\noindent \textbf{Proof of Theorem \ref{thm: regret analysis}: } The desired result follows from Propositions \ref{prop: Regret Approx SCD} and \ref{prop: Regret Approx DID}. $\blacksquare$

\section{Proofs of the Results in Section \ref{sec: SCD}}

\noindent \textbf{Proof of Theorem \ref{theorem:3.1}: } The first result follows from Lemmas \ref{lemm: w bound} and \ref{lemm: D1}. We can see that the second result follows from the first result using the standard arguments. Details are omitted. $\blacksquare$

\begin{theorem} 
    \label{thm: asymptotic validity}
    Suppose that Assumptions \ref{assump: sampling design} and \ref{assump: nonsingularity}, and (\ref{bounds}) in Assumption \ref{assump: moments} hold. Then, for any $\alpha \in (0,1)$, as $n \rightarrow \infty$, we have
\begin{align*}
  \liminf\limits_{n\rightarrow\infty} \inf_{P \in \mathcal{P}} P\left\{w^*(\lambda) \in \tilde C_{1-\alpha}\right\} \geq 1-\alpha.
\end{align*}
\end{theorem}

\noindent \textbf{Proof: } For any vector $x = (x_k)_{k=1}^K\in \mathbf{R}^K$, we write $J_0[x] = \{1 \le k \le K: x_k = 0\}$. We define 
\begin{align*}
    \Lambda(w) &= \{B_2' \lambda \in \mathbf{R}^{K-1}: w' \lambda = 0, \lambda \ge 0 \} \text{ and }\\
    \Lambda^\circ(w,\hat V(w)) &= \{x \in \mathbf{R}^{K-1}: [B_2 \hat V^{-1}(w) x]_{J_0[w]} \le 0 \}.
\end{align*}
It is not hard to see that $\Lambda(w)$ is a polyhedral cone and $\Lambda^\circ(w,\hat V(w))$ its polar cone along $\|\cdot\|_{\hat V(w)}$, where $\|x\|_{\hat V(w)}^2 = x' \hat V^{-1}(w) x$. We let $Y_n(w) = \sqrt{n} B_2' \hat \varphi(w)$. For any vector $y \in \mathbf{R}^{K-1}$ and a closed convex subset $C \subset \mathbf{R}^{K-1}$, the projection of $y$ onto $C$ along $\|\cdot\|_{\hat V(w)}$ is denoted by $\Pi_{\hat V(w)}(y \mid C)$. Then, we can write $\hat d(w)$ as follows:
\begin{align*}
    \hat d(w) &= |J_0[B_2 \hat V^{-1}(w) B_2'(\hat \varphi(w) - \hat \lambda(w))]|\\
              &= |J_0[B_2 \hat V^{-1}(w) (Y_n(w) - \Pi_{\hat V(w)}(Y_n(w) \mid \Lambda(w)))]|\\ 
              &= |J_0[B_2 \hat V^{-1}(w) \Pi_{\hat V(w)}(Y_n(w) \mid \Lambda^\circ(w, \hat V(w)))]|.
\end{align*}
It suffices to show that for each sequence $P_n \in \mathcal{P}$ and each sequence $w_n \in \mathbb{W}_{P_n}$, 
\begin{align*}
    \lim_{n \to \infty} P_n\left\{ T_{\mathsf{bf}}(w_n) > \hat c_{1-\kappa,\mathsf{bf}}(w_n)\right\} \le \kappa.
\end{align*}
We apply Lemma 3.1 of \cite{Canen/Song:arXiv:25} by setting $L=0$,
\begin{align*}
    Y_n = Y_n(w_n),\enspace \hat \Omega_n = \hat V(w_n),\enspace \Omega_n = V_{P_n}(w_n), \text{ and } \mu_n(w_n) = \sqrt{n} B_2' \varphi_{P_n}(w_n). 
\end{align*}
For this, we check Assumption 3.1 of \cite{Canen/Song:arXiv:25}. First, note that 
\begin{align}
    \label{consist hat V}
    \hat V(w_n) - V_{P_n}(w_n) = o_P(1),
\end{align}
by Assumption \ref{assump: sampling design} and the Law of Large Numbers, and
\begin{align*}
    Z_n \eqdef \hat \Omega_n^{-1/2}(Y_n - \mu_n) = \hat V^{-1/2}(w_n) \sqrt{n}B_2' (\hat \varphi(w_n) - \varphi_{P_n}(w_n)) \to_d N(0,I_{K-1}),
\end{align*}
as $n \to \infty$, by (\ref{consist hat V}) and the Central Limit Theorem applied to independent random variables, together with the condition (\ref{bounds}) in Assumption \ref{assump: moments}. Furthermore, by Assumption \ref{assump: nonsingularity}, for some constants $C,c>0$,
\begin{align*}
    \lambda_{\min}(V_{P_n}(w_n)) > c \text{ and } \|V_{P_n}(w_n)\| < C.
\end{align*}
Thus, Assumption 3.1 in \cite{Canen/Song:arXiv:25} is satisfied, and the desired result follows from their Lemma 3.1. $\blacksquare$\medskip

\noindent \textbf{Proof of Theorem \ref{thm: asymptotic validity2}: } We consider the first statement. By the SMC at $\lambda$, $\theta_t(\lambda, w^*(\lambda)) = \theta_t^*.$ Let 
\begin{align*}
    \tilde \psi_{ij,t} = \frac{n}{n_t} \frac{1\{G_{i,t} = j\}}{\hat p_{j,t}}(y_{i,t} - \mu_{j,t}).
\end{align*}
Recall $\hat p_{j,t} = n_{j,t}/n_t$. Note that 
\begin{align}
    \label{asymp lin}
    \sqrt{n}(\hat \theta_t(w) - \theta_t^*) &= \frac{1}{\sqrt{n}}\sum_{i=1}^n \left(\tilde \psi_{i0,t} - \sum_{j=1}^K \tilde \psi_{ij,t} w_{j}\right)\\ \notag
    &= \frac{1}{\sqrt{n}}\sum_{i=1}^n \left(\psi_{i0,t} - \sum_{j=1}^K \psi_{ij,t} w_{j}\right) + R_n(w),
\end{align}
where 
\begin{align*}
    R_n(w) &= \frac{1}{\sqrt{n}}\sum_{i=1}^n \frac{n}{n_t} 1\{G_{i,t} = 0\} \left(\frac{1}{\hat p_{0,t}} - \frac{1}{p_{0,t}} \right)(y_{i,t} - \mu_{0,t}) \\
    &\quad \quad + \sum_{j=1}^K \frac{1}{\sqrt{n}}\sum_{i=1}^n \frac{n}{n_t} 1\{G_{i,t} = j\} \left(\frac{1}{\hat p_{j,t}} - \frac{1}{p_{j,t}} \right)(y_{i,t} - \mu_{j,t}) w_{j}.
\end{align*}
Using standard arguments, we can show that 
\begin{align*}
    \sup_{w \in \Delta_{K-1}} |R_n(w)| = o_P(1),
\end{align*}
as $n \to \infty$. Using similar arguments, we can show that 
\begin{align*}
    \sup_{w \in \Delta_{K-1}} |\hat \sigma_t^2(w) - \sigma_t^2(w)| = o_P(1),
\end{align*}
where $\sigma^2(w) = \frac{1}{n}\sum_{i=1}^n \mathbf{E}_{P_n}\left[(\psi_{i0,t} - \sum_{j=1}^K \psi_{ij,t} w_j)^2\right]$.

Thus, we find that
\begin{align}
    \label{bd231}
    \sup_{w \in \Delta_{K-1}} P_n\left\{ \left| \frac{\sqrt{n}(\hat \theta_t(w) - \theta_t^*)}{\hat \sigma_t(w)} \right| > z_{1-\beta(\alpha,\kappa)} \right\} \le \alpha - \kappa + o(1),
\end{align}
by the Central Limit Theorem applied to the asymptotic linear representation in (\ref{asymp lin}).

First, take any sequence $P_n \in \mathcal{P}$. Note that 
\begin{align*}
    P_n\{\theta_t^* \in C_{1-\alpha}^{\mathsf{bf}}\} &= P_n\left\{ \inf_{w \in \tilde C_{1-\kappa}} \left| \frac{\sqrt{n}(\hat \theta_t(w) - \theta_t^*)}{\hat \sigma_t(w)} \right| \le z_{1-\beta(\alpha,\kappa)} \right\} \\
    &\ge \inf_{w \in \mathbb{W}_{P_n}} P_n\left\{ w \in \tilde C_{1-\kappa}, \left| \frac{\sqrt{n}(\hat \theta_t(w) - \theta_t^*)}{\hat \sigma_t(w)} \right| \le z_{1-\beta(\alpha,\kappa)} \right\}\\
    &\ge 1 - \sup_{w \in \mathbb{W}_{P_n}} P_n\left\{ w \notin \tilde C_{1-\kappa} \right\} - \sup_{w \in \mathbb{W}_{P_n}} P_n\left\{ \left| \frac{\sqrt{n}(\hat \theta_t(w) - \theta_t^*)}{\hat \sigma_t(w)} \right| > z_{1-\beta(\alpha,\kappa)} \right\}\\
    &\ge 1 - \kappa - (\alpha - \kappa) + o(1) = 1 - \alpha + o(1),
\end{align*}
by Theorem \ref{thm: asymptotic validity} and (\ref{bd231}). $\blacksquare$
\newpage

\section{Additional Simulation Results}\label{sec: AdditionalResults}
\input{a_big_table_finallargep_B}

\input{a_big_table_finalheteroF_B}

\input{a_big_table_finaltloading_B}

\end{document}

%% file: a_big_table_finalOG_B.tex
\begin{table}[t]
\centering
\caption{Results for Monte Carlo Simulations.}
\label{table:simulation_results_newOG}
\small
\begin{tabular*}{\textwidth}{@{\extracolsep{\fill}}ccccccccc}
\toprule\toprule
\multicolumn{3}{c}{Parameters} &
\multicolumn{2}{c}{MAD} &
\multicolumn{2}{c}{Coverage Probability} &
\multicolumn{2}{c}{CI Length} \\
\midrule
$K$ & $T$ & $n$ & SCD & CSDID & SCD & CSDID & SCD & CSDID \\
\midrule
\multicolumn{9}{l}{\textit{Scenario A: PTA and SMC hold}} \\[0.2em]
10 & 50  & 1500 & 0.079 & 0.145 & 0.999 & 0.998 & 1.475 & 1.166 \\
10 & 50  & 3000 & 0.059 & 0.105 & 1.000 & 0.999 & 0.967 & 0.828 \\
10 & 100 & 1500 & 0.079 & 0.148 & 1.000 & 0.998 & 1.397 & 1.228 \\
10 & 100 & 3000 & 0.054 & 0.100 & 1.000 & 1.000 & 0.907 & 0.875 \\
30 & 50  & 1500 & 0.128 & 0.229 & 1.000 & 0.997 & 2.456 & 1.770 \\
30 & 50  & 3000 & 0.094 & 0.163 & 1.000 & 1.000 & 1.824 & 1.279 \\
30 & 100 & 1500 & 0.130 & 0.228 & 1.000 & 0.999 & 2.280 & 1.864 \\
30 & 100 & 3000 & 0.090 & 0.157 & 1.000 & 0.999 & 1.715 & 1.349 \\
\midrule
\multicolumn{9}{l}{\textit{Scenario B: PTA fails but SMC holds}} \\[0.2em]
10 & 50  & 1500 & 0.092 & 1.914 & 0.999 & - & 1.002 & - \\
10 & 50  & 3000 & 0.069 & 1.909 & 0.996 & - & 0.683 & - \\
10 & 100 & 1500 & 0.092 & 1.861 & 0.997 & - & 0.965 & - \\
10 & 100 & 3000 & 0.062 & 1.857 & 0.998 & - & 0.652 & - \\
30 & 50  & 1500 & 0.140 & 2.083 & 0.999 & - & 1.590 & - \\
30 & 50  & 3000 & 0.100 & 2.073 & 0.997 & - & 1.072 & - \\
30 & 100 & 1500 & 0.137 & 2.006 & 0.992 & - & 1.378 & - \\
30 & 100 & 3000 & 0.098 & 1.982 & 0.993 & - & 0.927 & - \\
\midrule
\multicolumn{9}{l}{\textit{Scenario C: PTA holds but SMC fails}} \\[0.2em]
10 & 50  & 1500 & 1.163 & 0.145 & - & 0.999 & - & 1.167 \\
10 & 50  & 3000 & 1.159 & 0.105 & - & 0.999 & - & 0.828 \\
10 & 100 & 1500 & 1.105 & 0.148 & - & 0.998 & - & 1.229 \\
10 & 100 & 3000 & 1.130 & 0.100 & - & 1.000 & - & 0.875 \\
30 & 50  & 1500 & 1.317 & 0.229 & - & 0.998 & - & 1.771 \\
30 & 50  & 3000 & 1.340 & 0.163 & - & 1.000 & - & 1.278 \\
30 & 100 & 1500 & 1.209 & 0.228 & - & 0.999 & - & 1.866 \\
30 & 100 & 3000 & 1.198 & 0.157 & - & 0.999 & - & 1.350 \\
\bottomrule\bottomrule
\end{tabular*}
\medskip
\medskip
\vspace{0.01cm}
\parbox{6.4in}{\footnotesize
Notes: This table presents the simulation results for three scenarios: Scenario A assumes both Parallel Trends Assumption (PTA) and Stable Market Condition (SMC) hold, Scenario B assumes PTA fails but SMC holds, and Scenario C assumes PTA holds but SMC fails. The table reports Mean Absolute Deviation (MAD), Coverage Probability, and Confidence Interval (CI) Length for our proposed method of Synthetic Control with Differencing (SCD) and the Difference-in-Differences estimator proposed by \cite{Callaway/SantAnna:JoE:21} (CSDID). Inference for SCD is conducted using Algorithm \ref{alg:confidence_intervals 1}. Inference results are only reported for cases where the target parameter is identified. The number of MC simulations is 1,000.}
\end{table}

%% file: synthetic_hispnoncitizen.tex
\begin{tikzpicture}[x=1pt,y=1pt]
\definecolor{fillColor}{RGB}{255,255,255}
\path[use as bounding box,fill=fillColor,fill opacity=0.00] (0,0) rectangle (325.21,216.81);
\begin{scope}
\path[clip] (  0.00,  0.00) rectangle (325.21,216.81);
\definecolor{drawColor}{RGB}{255,255,255}
\definecolor{fillColor}{RGB}{255,255,255}

\path[draw=drawColor,line width= 0.4pt,line join=round,line cap=round,fill=fillColor] (  0.00,  0.00) rectangle (325.21,216.81);
\end{scope}
\begin{scope}
\path[clip] ( 26.29, 37.88) rectangle (321.47,213.06);
\definecolor{fillColor}{RGB}{255,255,255}

\path[fill=fillColor] ( 26.29, 37.88) rectangle (321.47,213.06);
\definecolor{drawColor}{gray}{0.92}

\path[draw=drawColor,line width= 0.2pt,line join=round] ( 26.29, 44.46) --
	(321.47, 44.46);

\path[draw=drawColor,line width= 0.2pt,line join=round] ( 26.29, 95.68) --
	(321.47, 95.68);

\path[draw=drawColor,line width= 0.2pt,line join=round] ( 26.29,146.91) --
	(321.47,146.91);

\path[draw=drawColor,line width= 0.2pt,line join=round] ( 26.29,198.14) --
	(321.47,198.14);

\path[draw=drawColor,line width= 0.4pt,line join=round] ( 26.29, 70.07) --
	(321.47, 70.07);

\path[draw=drawColor,line width= 0.4pt,line join=round] ( 26.29,121.30) --
	(321.47,121.30);

\path[draw=drawColor,line width= 0.4pt,line join=round] ( 26.29,172.52) --
	(321.47,172.52);
\definecolor{drawColor}{RGB}{0,0,0}

\path[draw=drawColor,line width= 1.4pt,line join=round] ( 39.70,169.92) --
	( 41.58,139.14) --
	( 43.46, 93.60) --
	( 45.33, 96.89) --
	( 47.21, 77.92) --
	( 49.09,101.27) --
	( 50.96,125.22) --
	( 52.84,127.51) --
	( 54.72,177.28) --
	( 56.59,147.11) --
	( 58.47,124.45) --
	( 60.35,133.62) --
	( 62.22,110.35) --
	( 64.10, 88.54) --
	( 65.98,120.02) --
	( 67.85,164.33) --
	( 69.73,146.88) --
	( 71.61,181.91) --
	( 73.48,166.11) --
	( 75.36,138.61) --
	( 77.24,145.11) --
	( 79.11,114.51) --
	( 80.99, 85.21) --
	( 82.86, 93.75) --
	( 84.74, 97.59) --
	( 86.62, 88.72) --
	( 88.49,103.36) --
	( 90.37,134.17) --
	( 92.25,144.15) --
	( 94.12,147.65) --
	( 96.00,142.60) --
	( 97.88,134.68) --
	( 99.75,130.15) --
	(101.63,122.75) --
	(103.51,106.59) --
	(105.38,100.97) --
	(107.26, 82.73) --
	(109.14, 76.65) --
	(111.01, 84.80) --
	(112.89, 94.14) --
	(114.77,114.91) --
	(116.64,113.98) --
	(118.52,167.05) --
	(120.40,161.99) --
	(122.27,142.88) --
	(124.15,146.15) --
	(126.02,100.82) --
	(127.90,130.75) --
	(129.78,113.12) --
	(131.65,142.59) --
	(133.53,127.55) --
	(135.41, 79.08) --
	(137.28, 91.41) --
	(139.16, 98.77) --
	(141.04,111.64) --
	(142.91,117.25) --
	(144.79,127.90) --
	(146.67,127.49) --
	(148.54,115.84) --
	(150.42,126.88) --
	(152.30,152.90) --
	(154.17,155.14) --
	(156.05,142.72) --
	(157.93,116.25) --
	(159.80,101.35) --
	(161.68,111.08) --
	(163.56,122.36) --
	(165.43,120.72) --
	(167.31,150.93) --
	(169.18,151.59) --
	(171.06,143.97) --
	(172.94,139.99) --
	(174.81,186.85) --
	(176.69,163.82) --
	(178.57,190.82) --
	(180.44,186.55) --
	(182.32,140.89) --
	(184.20,159.20) --
	(186.07,151.19) --
	(187.95,148.14) --
	(189.83,184.46) --
	(191.70,202.71) --
	(193.58,205.10) --
	(195.46,154.37) --
	(197.33, 85.32) --
	(199.21, 84.19) --
	(201.09,106.51) --
	(202.96,113.64) --
	(204.84,133.59) --
	(206.72,165.93) --
	(208.59,169.50) --
	(210.47,150.55) --
	(212.35,168.89) --
	(214.22,167.71) --
	(216.10,201.39) --
	(217.97,192.91) --
	(219.85,160.45) --
	(221.73,136.88) --
	(223.60,128.09) --
	(225.48,120.20) --
	(227.36,161.03) --
	(229.23,149.85) --
	(231.11,185.38) --
	(232.99,166.17) --
	(234.86,146.13) --
	(236.74,152.21) --
	(238.62,154.69) --
	(240.49,149.06) --
	(242.37,158.07) --
	(244.25,153.79) --
	(246.12,153.81) --
	(248.00,138.08) --
	(249.88,163.16) --
	(251.75,166.12) --
	(253.63,147.90) --
	(255.51,164.15) --
	(257.38,168.04) --
	(259.26,154.53) --
	(261.13,153.19) --
	(263.01,122.70) --
	(264.89,127.51) --
	(266.76,111.00) --
	(268.64,124.73) --
	(270.52,135.42) --
	(272.39,133.81) --
	(274.27,131.78) --
	(276.15,124.93) --
	(278.02,105.01) --
	(279.90, 98.79) --
	(281.78,110.13) --
	(283.65,119.26) --
	(285.53, 81.75) --
	(287.41, 78.06) --
	(289.28, 60.22) --
	(291.16, 72.13) --
	(293.04, 77.15) --
	(294.91, 69.75) --
	(296.79, 85.12) --
	(298.67,107.78) --
	(300.54, 89.77) --
	(302.42,100.69) --
	(304.29,108.95) --
	(306.17, 81.28) --
	(308.05, 80.72);
\definecolor{drawColor}{gray}{0.70}

\path[draw=drawColor,line width= 0.7pt,line join=round] ( 39.70,112.51) --
	( 41.58,118.46) --
	( 43.46,114.93) --
	( 45.33,120.20) --
	( 47.21,121.33) --
	( 49.09,121.16) --
	( 50.96,122.72) --
	( 52.84,120.12) --
	( 54.72,120.56) --
	( 56.59,122.85) --
	( 58.47,122.31) --
	( 60.35,119.26) --
	( 62.22,119.96) --
	( 64.10,118.71) --
	( 65.98,120.27) --
	( 67.85,124.82) --
	( 69.73,121.97) --
	( 71.61,124.16) --
	( 73.48,122.43) --
	( 75.36,118.71) --
	( 77.24,119.55) --
	( 79.11,124.72) --
	( 80.99,128.72) --
	( 82.86,125.76) --
	( 84.74,121.58) --
	( 86.62,120.29) --
	( 88.49,117.55) --
	( 90.37,122.35) --
	( 92.25,119.92) --
	( 94.12,124.84) --
	( 96.00,124.62) --
	( 97.88,126.93) --
	( 99.75,127.53) --
	(101.63,121.61) --
	(103.51,120.05) --
	(105.38,121.10) --
	(107.26,111.34) --
	(109.14,121.15) --
	(111.01,127.20) --
	(112.89,128.95) --
	(114.77,133.32) --
	(116.64,134.87) --
	(118.52,124.68) --
	(120.40,124.74) --
	(122.27,130.58) --
	(124.15,129.51) --
	(126.02,126.80) --
	(127.90,128.20) --
	(129.78,122.76) --
	(131.65,124.18) --
	(133.53,126.03) --
	(135.41,124.94) --
	(137.28,119.16) --
	(139.16,124.57) --
	(141.04,125.40) --
	(142.91,122.91) --
	(144.79,128.25) --
	(146.67,130.60) --
	(148.54,132.07) --
	(150.42,139.66) --
	(152.30,144.54) --
	(154.17,143.37) --
	(156.05,145.75) --
	(157.93,139.94) --
	(159.80,140.98) --
	(161.68,141.11) --
	(163.56,138.24) --
	(165.43,139.84) --
	(167.31,139.64) --
	(169.18,142.14) --
	(171.06,142.57) --
	(172.94,148.72) --
	(174.81,147.28) --
	(176.69,149.15) --
	(178.57,150.55) --
	(180.44,138.00) --
	(182.32,133.88) --
	(184.20,133.73) --
	(186.07,130.54) --
	(187.95,133.87) --
	(189.83,150.60) --
	(191.70,156.00) --
	(193.58,160.21) --
	(195.46,162.22) --
	(197.33,156.02) --
	(199.21,145.67) --
	(201.09,140.28) --
	(202.96,140.35) --
	(204.84,144.10) --
	(206.72,156.10) --
	(208.59,152.69) --
	(210.47,159.54) --
	(212.35,167.98) --
	(214.22,162.43) --
	(216.10,172.43) --
	(217.97,166.77) --
	(219.85,164.98) --
	(221.73,157.01) --
	(223.60,156.11) --
	(225.48,153.17) --
	(227.36,155.47) --
	(229.23,162.60) --
	(231.11,164.29) --
	(232.99,163.28) --
	(234.86,165.64) --
	(236.74,163.03) --
	(238.62,164.23) --
	(240.49,169.31) --
	(242.37,163.73) --
	(244.25,156.16) --
	(246.12,164.33) --
	(248.00,158.65) --
	(249.88,159.29) --
	(251.75,166.12) --
	(253.63,170.79) --
	(255.51,175.57) --
	(257.38,183.01) --
	(259.26,175.37) --
	(261.13,165.58) --
	(263.01,168.52) --
	(264.89,159.89) --
	(266.76,161.18) --
	(268.64,160.68) --
	(270.52,151.27) --
	(272.39,158.31) --
	(274.27,166.79) --
	(276.15,170.78) --
	(278.02,173.68) --
	(279.90,173.42) --
	(281.78,166.17) --
	(283.65,148.37) --
	(285.53,153.47) --
	(287.41,154.79) --
	(289.28,152.42) --
	(291.16,157.81) --
	(293.04,153.89) --
	(294.91,156.71) --
	(296.79,159.04) --
	(298.67,164.20) --
	(300.54,173.22) --
	(302.42,171.19) --
	(304.29,166.29) --
	(306.17,159.38) --
	(308.05,148.45);
\definecolor{drawColor}{RGB}{0,0,0}

\path[draw=drawColor,line width= 0.7pt,dash pattern=on 4pt off 4pt ,line join=round] (253.63, 37.88) -- (253.63,213.06);

\node[text=drawColor,anchor=base,inner sep=0pt, outer sep=0pt, scale=  0.43] at (283.65, 51.56) {Post-treatment};

\node[text=drawColor,anchor=base,inner sep=0pt, outer sep=0pt, scale=  0.43] at (225.48, 51.56) {Pre-treatment};
\definecolor{drawColor}{gray}{0.20}

\path[draw=drawColor,line width= 0.4pt,line join=round,line cap=round] ( 26.29, 37.88) rectangle (321.47,213.06);
\end{scope}
\begin{scope}
\path[clip] (  0.00,  0.00) rectangle (325.21,216.81);
\definecolor{drawColor}{gray}{0.30}

\node[text=drawColor,anchor=base east,inner sep=0pt, outer sep=0pt, scale=  0.60] at ( 22.91, 68.00) {0.06};

\node[text=drawColor,anchor=base east,inner sep=0pt, outer sep=0pt, scale=  0.60] at ( 22.91,119.23) {0.08};

\node[text=drawColor,anchor=base east,inner sep=0pt, outer sep=0pt, scale=  0.60] at ( 22.91,170.46) {0.10};
\end{scope}
\begin{scope}
\path[clip] (  0.00,  0.00) rectangle (325.21,216.81);
\definecolor{drawColor}{gray}{0.20}

\path[draw=drawColor,line width= 0.4pt,line join=round] ( 24.41, 70.07) --
	( 26.29, 70.07);

\path[draw=drawColor,line width= 0.4pt,line join=round] ( 24.41,121.30) --
	( 26.29,121.30);

\path[draw=drawColor,line width= 0.4pt,line join=round] ( 24.41,172.52) --
	( 26.29,172.52);
\end{scope}
\begin{scope}
\path[clip] (  0.00,  0.00) rectangle (325.21,216.81);
\definecolor{drawColor}{gray}{0.20}

\path[draw=drawColor,line width= 0.4pt,line join=round] ( 39.70, 36.00) --
	( 39.70, 37.88);

\path[draw=drawColor,line width= 0.4pt,line join=round] ( 84.74, 36.00) --
	( 84.74, 37.88);

\path[draw=drawColor,line width= 0.4pt,line join=round] (129.78, 36.00) --
	(129.78, 37.88);

\path[draw=drawColor,line width= 0.4pt,line join=round] (174.81, 36.00) --
	(174.81, 37.88);

\path[draw=drawColor,line width= 0.4pt,line join=round] (219.85, 36.00) --
	(219.85, 37.88);

\path[draw=drawColor,line width= 0.4pt,line join=round] (264.89, 36.00) --
	(264.89, 37.88);
\end{scope}
\begin{scope}
\path[clip] (  0.00,  0.00) rectangle (325.21,216.81);
\definecolor{drawColor}{gray}{0.30}

\node[text=drawColor,anchor=base,inner sep=0pt, outer sep=0pt, scale=  0.60] at ( 39.70, 30.37) {1998};

\node[text=drawColor,anchor=base,inner sep=0pt, outer sep=0pt, scale=  0.60] at ( 84.74, 30.37) {2000};

\node[text=drawColor,anchor=base,inner sep=0pt, outer sep=0pt, scale=  0.60] at (129.78, 30.37) {2002};

\node[text=drawColor,anchor=base,inner sep=0pt, outer sep=0pt, scale=  0.60] at (174.81, 30.37) {2004};

\node[text=drawColor,anchor=base,inner sep=0pt, outer sep=0pt, scale=  0.60] at (219.85, 30.37) {2006};

\node[text=drawColor,anchor=base,inner sep=0pt, outer sep=0pt, scale=  0.60] at (264.89, 30.37) {2008};
\end{scope}
\begin{scope}
\path[clip] (  0.00,  0.00) rectangle (325.21,216.81);
\definecolor{fillColor}{RGB}{255,255,255}

\path[fill=fillColor] (106.80,  3.75) rectangle (240.95, 13.20);
\end{scope}
\begin{scope}
\path[clip] (  0.00,  0.00) rectangle (325.21,216.81);
\definecolor{fillColor}{RGB}{255,255,255}

\path[fill=fillColor] (110.55,  3.75) rectangle (125.00, 18.20);
\definecolor{drawColor}{RGB}{0,0,0}

\path[draw=drawColor,line width= 1.4pt,line join=round] (112.00, 10.98) -- (123.56, 10.98);
\end{scope}
\begin{scope}
\path[clip] (  0.00,  0.00) rectangle (325.21,216.81);
\definecolor{fillColor}{RGB}{255,255,255}

\path[fill=fillColor] (159.85,  3.75) rectangle (174.31, 18.20);
\definecolor{drawColor}{gray}{0.70}

\path[draw=drawColor,line width= 0.7pt,line join=round] (161.30, 10.98) -- (172.86, 10.98);
\end{scope}
\begin{scope}
\path[clip] (  0.00,  0.00) rectangle (325.21,216.81);
\definecolor{drawColor}{RGB}{0,0,0}

\node[text=drawColor,anchor=base west,inner sep=0pt, outer sep=0pt, scale=  0.80] at (128.75,  8.22) {Arizona};
\end{scope}
\begin{scope}
\path[clip] (  0.00,  0.00) rectangle (325.21,216.81);
\definecolor{drawColor}{RGB}{0,0,0}

\node[text=drawColor,anchor=base west,inner sep=0pt, outer sep=0pt, scale=  0.80] at (178.06,  8.22) {Synthetic Arizona};
\end{scope}
\end{tikzpicture}

%% file: att_hispnoncitizen_shorterpre.tex
\begin{tikzpicture}[x=1pt,y=1pt]
\definecolor{fillColor}{RGB}{255,255,255}
\path[use as bounding box,fill=fillColor,fill opacity=0.00] (0,0) rectangle (325.21,216.81);
\begin{scope}
\path[clip] (  0.00,  0.00) rectangle (325.21,216.81);
\definecolor{drawColor}{RGB}{255,255,255}
\definecolor{fillColor}{RGB}{255,255,255}

\path[draw=drawColor,line width= 0.4pt,line join=round,line cap=round,fill=fillColor] (  0.00,  0.00) rectangle (325.21,216.81);
\end{scope}
\begin{scope}
\path[clip] ( 31.29, 37.88) rectangle (321.47,213.06);
\definecolor{fillColor}{RGB}{255,255,255}

\path[fill=fillColor] ( 31.29, 37.88) rectangle (321.47,213.06);
\definecolor{drawColor}{gray}{0.92}

\path[draw=drawColor,line width= 0.2pt,line join=round] ( 31.29, 58.74) --
	(321.47, 58.74);

\path[draw=drawColor,line width= 0.2pt,line join=round] ( 31.29, 91.12) --
	(321.47, 91.12);

\path[draw=drawColor,line width= 0.2pt,line join=round] ( 31.29,123.50) --
	(321.47,123.50);

\path[draw=drawColor,line width= 0.2pt,line join=round] ( 31.29,155.87) --
	(321.47,155.87);

\path[draw=drawColor,line width= 0.2pt,line join=round] ( 31.29,188.25) --
	(321.47,188.25);

\path[draw=drawColor,line width= 0.4pt,line join=round] ( 31.29, 42.55) --
	(321.47, 42.55);

\path[draw=drawColor,line width= 0.4pt,line join=round] ( 31.29, 74.93) --
	(321.47, 74.93);

\path[draw=drawColor,line width= 0.4pt,line join=round] ( 31.29,107.31) --
	(321.47,107.31);

\path[draw=drawColor,line width= 0.4pt,line join=round] ( 31.29,139.68) --
	(321.47,139.68);

\path[draw=drawColor,line width= 0.4pt,line join=round] ( 31.29,172.06) --
	(321.47,172.06);

\path[draw=drawColor,line width= 0.4pt,line join=round] ( 31.29,204.44) --
	(321.47,204.44);
\definecolor{drawColor}{RGB}{70,130,180}
\definecolor{fillColor}{RGB}{70,130,180}

\path[draw=drawColor,line width= 0.3pt,line join=round,line cap=round,fill=fillColor] ( 44.79,142.81) circle (  2.38);

\path[draw=drawColor,line width= 0.3pt,line join=round,line cap=round,fill=fillColor] ( 47.96,148.14) circle (  2.38);

\path[draw=drawColor,line width= 0.3pt,line join=round,line cap=round,fill=fillColor] ( 51.13,141.79) circle (  2.38);

\path[draw=drawColor,line width= 0.3pt,line join=round,line cap=round,fill=fillColor] ( 54.31,132.48) circle (  2.38);

\path[draw=drawColor,line width= 0.3pt,line join=round,line cap=round,fill=fillColor] ( 57.48,123.28) circle (  2.38);

\path[draw=drawColor,line width= 0.3pt,line join=round,line cap=round,fill=fillColor] ( 60.65,124.27) circle (  2.38);

\path[draw=drawColor,line width= 0.3pt,line join=round,line cap=round,fill=fillColor] ( 63.82,131.40) circle (  2.38);

\path[draw=drawColor,line width= 0.3pt,line join=round,line cap=round,fill=fillColor] ( 66.99,129.09) circle (  2.38);

\path[draw=drawColor,line width= 0.3pt,line join=round,line cap=round,fill=fillColor] ( 70.16,141.61) circle (  2.38);

\path[draw=drawColor,line width= 0.3pt,line join=round,line cap=round,fill=fillColor] ( 73.33,140.58) circle (  2.38);

\path[draw=drawColor,line width= 0.3pt,line join=round,line cap=round,fill=fillColor] ( 76.50,136.55) circle (  2.38);

\path[draw=drawColor,line width= 0.3pt,line join=round,line cap=round,fill=fillColor] ( 79.67,133.04) circle (  2.38);

\path[draw=drawColor,line width= 0.3pt,line join=round,line cap=round,fill=fillColor] ( 82.84,158.14) circle (  2.38);

\path[draw=drawColor,line width= 0.3pt,line join=round,line cap=round,fill=fillColor] ( 86.01,146.44) circle (  2.38);

\path[draw=drawColor,line width= 0.3pt,line join=round,line cap=round,fill=fillColor] ( 89.18,159.84) circle (  2.38);

\path[draw=drawColor,line width= 0.3pt,line join=round,line cap=round,fill=fillColor] ( 92.35,162.93) circle (  2.38);

\path[draw=drawColor,line width= 0.3pt,line join=round,line cap=round,fill=fillColor] ( 95.52,144.13) circle (  2.38);

\path[draw=drawColor,line width= 0.3pt,line join=round,line cap=round,fill=fillColor] ( 98.69,150.26) circle (  2.38);

\path[draw=drawColor,line width= 0.3pt,line join=round,line cap=round,fill=fillColor] (101.87,146.13) circle (  2.38);

\path[draw=drawColor,line width= 0.3pt,line join=round,line cap=round,fill=fillColor] (105.04,144.52) circle (  2.38);

\path[draw=drawColor,line width= 0.3pt,line join=round,line cap=round,fill=fillColor] (108.21,152.46) circle (  2.38);

\path[draw=drawColor,line width= 0.3pt,line join=round,line cap=round,fill=fillColor] (111.38,159.43) circle (  2.38);

\path[draw=drawColor,line width= 0.3pt,line join=round,line cap=round,fill=fillColor] (114.55,158.85) circle (  2.38);

\path[draw=drawColor,line width= 0.3pt,line join=round,line cap=round,fill=fillColor] (117.72,132.97) circle (  2.38);

\path[draw=drawColor,line width= 0.3pt,line join=round,line cap=round,fill=fillColor] (120.89,103.91) circle (  2.38);

\path[draw=drawColor,line width= 0.3pt,line join=round,line cap=round,fill=fillColor] (124.06,109.28) circle (  2.38);

\path[draw=drawColor,line width= 0.3pt,line join=round,line cap=round,fill=fillColor] (127.23,122.72) circle (  2.38);

\path[draw=drawColor,line width= 0.3pt,line join=round,line cap=round,fill=fillColor] (130.40,126.94) circle (  2.38);

\path[draw=drawColor,line width= 0.3pt,line join=round,line cap=round,fill=fillColor] (133.57,134.06) circle (  2.38);

\path[draw=drawColor,line width= 0.3pt,line join=round,line cap=round,fill=fillColor] (136.74,145.35) circle (  2.38);

\path[draw=drawColor,line width= 0.3pt,line join=round,line cap=round,fill=fillColor] (139.91,149.95) circle (  2.38);

\path[draw=drawColor,line width= 0.3pt,line join=round,line cap=round,fill=fillColor] (143.08,137.37) circle (  2.38);

\path[draw=drawColor,line width= 0.3pt,line join=round,line cap=round,fill=fillColor] (146.25,142.75) circle (  2.38);

\path[draw=drawColor,line width= 0.3pt,line join=round,line cap=round,fill=fillColor] (149.43,146.23) circle (  2.38);

\path[draw=drawColor,line width= 0.3pt,line join=round,line cap=round,fill=fillColor] (152.60,154.99) circle (  2.38);

\path[draw=drawColor,line width= 0.3pt,line join=round,line cap=round,fill=fillColor] (155.77,155.24) circle (  2.38);

\path[draw=drawColor,line width= 0.3pt,line join=round,line cap=round,fill=fillColor] (158.94,138.17) circle (  2.38);

\path[draw=drawColor,line width= 0.3pt,line join=round,line cap=round,fill=fillColor] (162.11,128.27) circle (  2.38);

\path[draw=drawColor,line width= 0.3pt,line join=round,line cap=round,fill=fillColor] (165.28,124.44) circle (  2.38);

\path[draw=drawColor,line width= 0.3pt,line join=round,line cap=round,fill=fillColor] (168.45,119.97) circle (  2.38);

\path[draw=drawColor,line width= 0.3pt,line join=round,line cap=round,fill=fillColor] (171.62,142.83) circle (  2.38);

\path[draw=drawColor,line width= 0.3pt,line join=round,line cap=round,fill=fillColor] (174.79,137.03) circle (  2.38);

\path[draw=drawColor,line width= 0.3pt,line join=round,line cap=round,fill=fillColor] (177.96,152.92) circle (  2.38);

\path[draw=drawColor,line width= 0.3pt,line join=round,line cap=round,fill=fillColor] (181.13,145.23) circle (  2.38);

\path[draw=drawColor,line width= 0.3pt,line join=round,line cap=round,fill=fillColor] (184.30,132.67) circle (  2.38);

\path[draw=drawColor,line width= 0.3pt,line join=round,line cap=round,fill=fillColor] (187.47,136.99) circle (  2.38);

\path[draw=drawColor,line width= 0.3pt,line join=round,line cap=round,fill=fillColor] (190.64,137.32) circle (  2.38);

\path[draw=drawColor,line width= 0.3pt,line join=round,line cap=round,fill=fillColor] (193.81,131.74) circle (  2.38);

\path[draw=drawColor,line width= 0.3pt,line join=round,line cap=round,fill=fillColor] (196.98,141.13) circle (  2.38);

\path[draw=drawColor,line width= 0.3pt,line join=round,line cap=round,fill=fillColor] (200.16,138.77) circle (  2.38);

\path[draw=drawColor,line width= 0.3pt,line join=round,line cap=round,fill=fillColor] (203.33,138.11) circle (  2.38);

\path[draw=drawColor,line width= 0.3pt,line join=round,line cap=round,fill=fillColor] (206.50,129.63) circle (  2.38);

\path[draw=drawColor,line width= 0.3pt,line join=round,line cap=round,fill=fillColor] (209.67,140.21) circle (  2.38);

\path[draw=drawColor,line width= 0.3pt,line join=round,line cap=round,fill=fillColor] (212.84,139.68) circle (  2.38);
\definecolor{drawColor}{RGB}{178,34,34}
\definecolor{fillColor}{RGB}{178,34,34}

\path[draw=drawColor,line width= 0.3pt,line join=round,line cap=round,fill=fillColor] (216.01,128.47) circle (  2.38);

\path[draw=drawColor,line width= 0.3pt,line join=round,line cap=round,fill=fillColor] (219.18,134.00) circle (  2.38);

\path[draw=drawColor,line width= 0.3pt,line join=round,line cap=round,fill=fillColor] (222.35,133.74) circle (  2.38);

\path[draw=drawColor,line width= 0.3pt,line join=round,line cap=round,fill=fillColor] (225.52,131.87) circle (  2.38);

\path[draw=drawColor,line width= 0.3pt,line join=round,line cap=round,fill=fillColor] (228.69,136.28) circle (  2.38);

\path[draw=drawColor,line width= 0.3pt,line join=round,line cap=round,fill=fillColor] (231.86,117.02) circle (  2.38);

\path[draw=drawColor,line width= 0.3pt,line join=round,line cap=round,fill=fillColor] (235.03,123.17) circle (  2.38);

\path[draw=drawColor,line width= 0.3pt,line join=round,line cap=round,fill=fillColor] (238.20,114.28) circle (  2.38);

\path[draw=drawColor,line width= 0.3pt,line join=round,line cap=round,fill=fillColor] (241.37,119.32) circle (  2.38);

\path[draw=drawColor,line width= 0.3pt,line join=round,line cap=round,fill=fillColor] (244.54,130.07) circle (  2.38);

\path[draw=drawColor,line width= 0.3pt,line join=round,line cap=round,fill=fillColor] (247.72,127.57) circle (  2.38);

\path[draw=drawColor,line width= 0.3pt,line join=round,line cap=round,fill=fillColor] (250.89,123.59) circle (  2.38);

\path[draw=drawColor,line width= 0.3pt,line join=round,line cap=round,fill=fillColor] (254.06,120.74) circle (  2.38);

\path[draw=drawColor,line width= 0.3pt,line join=round,line cap=round,fill=fillColor] (257.23,108.69) circle (  2.38);

\path[draw=drawColor,line width= 0.3pt,line join=round,line cap=round,fill=fillColor] (260.40,103.32) circle (  2.38);

\path[draw=drawColor,line width= 0.3pt,line join=round,line cap=round,fill=fillColor] (263.57,113.46) circle (  2.38);

\path[draw=drawColor,line width= 0.3pt,line join=round,line cap=round,fill=fillColor] (266.74,124.93) circle (  2.38);

\path[draw=drawColor,line width= 0.3pt,line join=round,line cap=round,fill=fillColor] (269.91,102.61) circle (  2.38);

\path[draw=drawColor,line width= 0.3pt,line join=round,line cap=round,fill=fillColor] (273.08,104.47) circle (  2.38);

\path[draw=drawColor,line width= 0.3pt,line join=round,line cap=round,fill=fillColor] (276.25, 96.74) circle (  2.38);

\path[draw=drawColor,line width= 0.3pt,line join=round,line cap=round,fill=fillColor] (279.42,100.52) circle (  2.38);

\path[draw=drawColor,line width= 0.3pt,line join=round,line cap=round,fill=fillColor] (282.59,103.92) circle (  2.38);

\path[draw=drawColor,line width= 0.3pt,line join=round,line cap=round,fill=fillColor] (285.76, 98.55) circle (  2.38);

\path[draw=drawColor,line width= 0.3pt,line join=round,line cap=round,fill=fillColor] (288.93,102.09) circle (  2.38);

\path[draw=drawColor,line width= 0.3pt,line join=round,line cap=round,fill=fillColor] (292.10,115.21) circle (  2.38);

\path[draw=drawColor,line width= 0.3pt,line join=round,line cap=round,fill=fillColor] (295.28,100.50) circle (  2.38);

\path[draw=drawColor,line width= 0.3pt,line join=round,line cap=round,fill=fillColor] (298.45,106.43) circle (  2.38);

\path[draw=drawColor,line width= 0.3pt,line join=round,line cap=round,fill=fillColor] (301.62,113.32) circle (  2.38);

\path[draw=drawColor,line width= 0.3pt,line join=round,line cap=round,fill=fillColor] (304.79,103.36) circle (  2.38);

\path[draw=drawColor,line width= 0.3pt,line join=round,line cap=round,fill=fillColor] (307.96,104.87) circle (  2.38);
\definecolor{drawColor}{RGB}{70,130,180}

\path[draw=drawColor,line width= 0.4pt,line join=round] ( 44.48,172.27) --
	( 45.11,172.27);

\path[draw=drawColor,line width= 0.4pt,line join=round] ( 44.79,172.27) --
	( 44.79, 91.09);

\path[draw=drawColor,line width= 0.4pt,line join=round] ( 44.48, 91.09) --
	( 45.11, 91.09);

\path[draw=drawColor,line width= 0.4pt,line join=round] ( 47.65,180.49) --
	( 48.28,180.49);

\path[draw=drawColor,line width= 0.4pt,line join=round] ( 47.96,180.49) --
	( 47.96, 97.12);

\path[draw=drawColor,line width= 0.4pt,line join=round] ( 47.65, 97.12) --
	( 48.28, 97.12);

\path[draw=drawColor,line width= 0.4pt,line join=round] ( 50.82,172.72) --
	( 51.45,172.72);

\path[draw=drawColor,line width= 0.4pt,line join=round] ( 51.13,172.72) --
	( 51.13, 95.36);

\path[draw=drawColor,line width= 0.4pt,line join=round] ( 50.82, 95.36) --
	( 51.45, 95.36);

\path[draw=drawColor,line width= 0.4pt,line join=round] ( 53.99,170.80) --
	( 54.62,170.80);

\path[draw=drawColor,line width= 0.4pt,line join=round] ( 54.31,170.80) --
	( 54.31, 84.37);

\path[draw=drawColor,line width= 0.4pt,line join=round] ( 53.99, 84.37) --
	( 54.62, 84.37);

\path[draw=drawColor,line width= 0.4pt,line join=round] ( 57.16,159.37) --
	( 57.79,159.37);

\path[draw=drawColor,line width= 0.4pt,line join=round] ( 57.48,159.37) --
	( 57.48, 72.61);

\path[draw=drawColor,line width= 0.4pt,line join=round] ( 57.16, 72.61) --
	( 57.79, 72.61);

\path[draw=drawColor,line width= 0.4pt,line join=round] ( 60.33,159.76) --
	( 60.96,159.76);

\path[draw=drawColor,line width= 0.4pt,line join=round] ( 60.65,159.76) --
	( 60.65, 75.09);

\path[draw=drawColor,line width= 0.4pt,line join=round] ( 60.33, 75.09) --
	( 60.96, 75.09);

\path[draw=drawColor,line width= 0.4pt,line join=round] ( 63.50,171.42) --
	( 64.13,171.42);

\path[draw=drawColor,line width= 0.4pt,line join=round] ( 63.82,171.42) --
	( 63.82, 75.06);

\path[draw=drawColor,line width= 0.4pt,line join=round] ( 63.50, 75.06) --
	( 64.13, 75.06);

\path[draw=drawColor,line width= 0.4pt,line join=round] ( 66.67,166.93) --
	( 67.31,166.93);

\path[draw=drawColor,line width= 0.4pt,line join=round] ( 66.99,166.93) --
	( 66.99, 71.74);

\path[draw=drawColor,line width= 0.4pt,line join=round] ( 66.67, 71.74) --
	( 67.31, 71.74);

\path[draw=drawColor,line width= 0.4pt,line join=round] ( 69.84,180.97) --
	( 70.48,180.97);

\path[draw=drawColor,line width= 0.4pt,line join=round] ( 70.16,180.97) --
	( 70.16, 91.31);

\path[draw=drawColor,line width= 0.4pt,line join=round] ( 69.84, 91.31) --
	( 70.48, 91.31);

\path[draw=drawColor,line width= 0.4pt,line join=round] ( 73.01,179.30) --
	( 73.65,179.30);

\path[draw=drawColor,line width= 0.4pt,line join=round] ( 73.33,179.30) --
	( 73.33, 92.11);

\path[draw=drawColor,line width= 0.4pt,line join=round] ( 73.01, 92.11) --
	( 73.65, 92.11);

\path[draw=drawColor,line width= 0.4pt,line join=round] ( 76.18,175.66) --
	( 76.82,175.66);

\path[draw=drawColor,line width= 0.4pt,line join=round] ( 76.50,175.66) --
	( 76.50, 92.57);

\path[draw=drawColor,line width= 0.4pt,line join=round] ( 76.18, 92.57) --
	( 76.82, 92.57);

\path[draw=drawColor,line width= 0.4pt,line join=round] ( 79.35,167.96) --
	( 79.99,167.96);

\path[draw=drawColor,line width= 0.4pt,line join=round] ( 79.67,167.96) --
	( 79.67, 86.57);

\path[draw=drawColor,line width= 0.4pt,line join=round] ( 79.35, 86.57) --
	( 79.99, 86.57);

\path[draw=drawColor,line width= 0.4pt,line join=round] ( 82.52,193.73) --
	( 83.16,193.73);

\path[draw=drawColor,line width= 0.4pt,line join=round] ( 82.84,193.73) --
	( 82.84,105.20);

\path[draw=drawColor,line width= 0.4pt,line join=round] ( 82.52,105.20) --
	( 83.16,105.20);

\path[draw=drawColor,line width= 0.4pt,line join=round] ( 85.69,184.37) --
	( 86.33,184.37);

\path[draw=drawColor,line width= 0.4pt,line join=round] ( 86.01,184.37) --
	( 86.01,101.89);

\path[draw=drawColor,line width= 0.4pt,line join=round] ( 85.69,101.89) --
	( 86.33,101.89);

\path[draw=drawColor,line width= 0.4pt,line join=round] ( 88.87,195.62) --
	( 89.50,195.62);

\path[draw=drawColor,line width= 0.4pt,line join=round] ( 89.18,195.62) --
	( 89.18,114.96);

\path[draw=drawColor,line width= 0.4pt,line join=round] ( 88.87,114.96) --
	( 89.50,114.96);

\path[draw=drawColor,line width= 0.4pt,line join=round] ( 92.04,205.10) --
	( 92.67,205.10);

\path[draw=drawColor,line width= 0.4pt,line join=round] ( 92.35,205.10) --
	( 92.35,113.13);

\path[draw=drawColor,line width= 0.4pt,line join=round] ( 92.04,113.13) --
	( 92.67,113.13);

\path[draw=drawColor,line width= 0.4pt,line join=round] ( 95.21,186.61) --
	( 95.84,186.61);

\path[draw=drawColor,line width= 0.4pt,line join=round] ( 95.52,186.61) --
	( 95.52, 97.90);

\path[draw=drawColor,line width= 0.4pt,line join=round] ( 95.21, 97.90) --
	( 95.84, 97.90);

\path[draw=drawColor,line width= 0.4pt,line join=round] ( 98.38,191.96) --
	( 99.01,191.96);

\path[draw=drawColor,line width= 0.4pt,line join=round] ( 98.69,191.96) --
	( 98.69,101.09);

\path[draw=drawColor,line width= 0.4pt,line join=round] ( 98.38,101.09) --
	( 99.01,101.09);

\path[draw=drawColor,line width= 0.4pt,line join=round] (101.55,185.97) --
	(102.18,185.97);

\path[draw=drawColor,line width= 0.4pt,line join=round] (101.87,185.97) --
	(101.87, 94.22);

\path[draw=drawColor,line width= 0.4pt,line join=round] (101.55, 94.22) --
	(102.18, 94.22);

\path[draw=drawColor,line width= 0.4pt,line join=round] (104.72,184.61) --
	(105.35,184.61);

\path[draw=drawColor,line width= 0.4pt,line join=round] (105.04,184.61) --
	(105.04, 92.85);

\path[draw=drawColor,line width= 0.4pt,line join=round] (104.72, 92.85) --
	(105.35, 92.85);

\path[draw=drawColor,line width= 0.4pt,line join=round] (107.89,187.13) --
	(108.52,187.13);

\path[draw=drawColor,line width= 0.4pt,line join=round] (108.21,187.13) --
	(108.21,102.31);

\path[draw=drawColor,line width= 0.4pt,line join=round] (107.89,102.31) --
	(108.52,102.31);

\path[draw=drawColor,line width= 0.4pt,line join=round] (111.06,192.90) --
	(111.69,192.90);

\path[draw=drawColor,line width= 0.4pt,line join=round] (111.38,192.90) --
	(111.38,116.75);

\path[draw=drawColor,line width= 0.4pt,line join=round] (111.06,116.75) --
	(111.69,116.75);

\path[draw=drawColor,line width= 0.4pt,line join=round] (114.23,193.51) --
	(114.86,193.51);

\path[draw=drawColor,line width= 0.4pt,line join=round] (114.55,193.51) --
	(114.55,117.52);

\path[draw=drawColor,line width= 0.4pt,line join=round] (114.23,117.52) --
	(114.86,117.52);

\path[draw=drawColor,line width= 0.4pt,line join=round] (117.40,166.74) --
	(118.04,166.74);

\path[draw=drawColor,line width= 0.4pt,line join=round] (117.72,166.74) --
	(117.72, 92.66);

\path[draw=drawColor,line width= 0.4pt,line join=round] (117.40, 92.66) --
	(118.04, 92.66);

\path[draw=drawColor,line width= 0.4pt,line join=round] (120.57,131.35) --
	(121.21,131.35);

\path[draw=drawColor,line width= 0.4pt,line join=round] (120.89,131.35) --
	(120.89, 70.13);

\path[draw=drawColor,line width= 0.4pt,line join=round] (120.57, 70.13) --
	(121.21, 70.13);

\path[draw=drawColor,line width= 0.4pt,line join=round] (123.74,138.22) --
	(124.38,138.22);

\path[draw=drawColor,line width= 0.4pt,line join=round] (124.06,138.22) --
	(124.06, 59.98);

\path[draw=drawColor,line width= 0.4pt,line join=round] (123.74, 59.98) --
	(124.38, 59.98);

\path[draw=drawColor,line width= 0.4pt,line join=round] (126.91,154.25) --
	(127.55,154.25);

\path[draw=drawColor,line width= 0.4pt,line join=round] (127.23,154.25) --
	(127.23, 78.16);

\path[draw=drawColor,line width= 0.4pt,line join=round] (126.91, 78.16) --
	(127.55, 78.16);

\path[draw=drawColor,line width= 0.4pt,line join=round] (130.08,160.59) --
	(130.72,160.59);

\path[draw=drawColor,line width= 0.4pt,line join=round] (130.40,160.59) --
	(130.40, 78.12);

\path[draw=drawColor,line width= 0.4pt,line join=round] (130.08, 78.12) --
	(130.72, 78.12);

\path[draw=drawColor,line width= 0.4pt,line join=round] (133.25,170.35) --
	(133.89,170.35);

\path[draw=drawColor,line width= 0.4pt,line join=round] (133.57,170.35) --
	(133.57, 85.23);

\path[draw=drawColor,line width= 0.4pt,line join=round] (133.25, 85.23) --
	(133.89, 85.23);

\path[draw=drawColor,line width= 0.4pt,line join=round] (136.43,179.10) --
	(137.06,179.10);

\path[draw=drawColor,line width= 0.4pt,line join=round] (136.74,179.10) --
	(136.74, 92.15);

\path[draw=drawColor,line width= 0.4pt,line join=round] (136.43, 92.15) --
	(137.06, 92.15);

\path[draw=drawColor,line width= 0.4pt,line join=round] (139.60,185.27) --
	(140.23,185.27);

\path[draw=drawColor,line width= 0.4pt,line join=round] (139.91,185.27) --
	(139.91,100.01);

\path[draw=drawColor,line width= 0.4pt,line join=round] (139.60,100.01) --
	(140.23,100.01);

\path[draw=drawColor,line width= 0.4pt,line join=round] (142.77,172.98) --
	(143.40,172.98);

\path[draw=drawColor,line width= 0.4pt,line join=round] (143.08,172.98) --
	(143.08, 97.30);

\path[draw=drawColor,line width= 0.4pt,line join=round] (142.77, 97.30) --
	(143.40, 97.30);

\path[draw=drawColor,line width= 0.4pt,line join=round] (145.94,173.53) --
	(146.57,173.53);

\path[draw=drawColor,line width= 0.4pt,line join=round] (146.25,173.53) --
	(146.25,104.11);

\path[draw=drawColor,line width= 0.4pt,line join=round] (145.94,104.11) --
	(146.57,104.11);

\path[draw=drawColor,line width= 0.4pt,line join=round] (149.11,179.53) --
	(149.74,179.53);

\path[draw=drawColor,line width= 0.4pt,line join=round] (149.43,179.53) --
	(149.43, 98.56);

\path[draw=drawColor,line width= 0.4pt,line join=round] (149.11, 98.56) --
	(149.74, 98.56);

\path[draw=drawColor,line width= 0.4pt,line join=round] (152.28,186.54) --
	(152.91,186.54);

\path[draw=drawColor,line width= 0.4pt,line join=round] (152.60,186.54) --
	(152.60,114.07);

\path[draw=drawColor,line width= 0.4pt,line join=round] (152.28,114.07) --
	(152.91,114.07);

\path[draw=drawColor,line width= 0.4pt,line join=round] (155.45,188.66) --
	(156.08,188.66);

\path[draw=drawColor,line width= 0.4pt,line join=round] (155.77,188.66) --
	(155.77,110.07);

\path[draw=drawColor,line width= 0.4pt,line join=round] (155.45,110.07) --
	(156.08,110.07);

\path[draw=drawColor,line width= 0.4pt,line join=round] (158.62,173.43) --
	(159.25,173.43);

\path[draw=drawColor,line width= 0.4pt,line join=round] (158.94,173.43) --
	(158.94, 97.11);

\path[draw=drawColor,line width= 0.4pt,line join=round] (158.62, 97.11) --
	(159.25, 97.11);

\path[draw=drawColor,line width= 0.4pt,line join=round] (161.79,164.07) --
	(162.42,164.07);

\path[draw=drawColor,line width= 0.4pt,line join=round] (162.11,164.07) --
	(162.11, 90.29);

\path[draw=drawColor,line width= 0.4pt,line join=round] (161.79, 90.29) --
	(162.42, 90.29);

\path[draw=drawColor,line width= 0.4pt,line join=round] (164.96,157.07) --
	(165.60,157.07);

\path[draw=drawColor,line width= 0.4pt,line join=round] (165.28,157.07) --
	(165.28, 87.35);

\path[draw=drawColor,line width= 0.4pt,line join=round] (164.96, 87.35) --
	(165.60, 87.35);

\path[draw=drawColor,line width= 0.4pt,line join=round] (168.13,152.32) --
	(168.77,152.32);

\path[draw=drawColor,line width= 0.4pt,line join=round] (168.45,152.32) --
	(168.45, 74.98);

\path[draw=drawColor,line width= 0.4pt,line join=round] (168.13, 74.98) --
	(168.77, 74.98);

\path[draw=drawColor,line width= 0.4pt,line join=round] (171.30,175.11) --
	(171.94,175.11);

\path[draw=drawColor,line width= 0.4pt,line join=round] (171.62,175.11) --
	(171.62, 97.35);

\path[draw=drawColor,line width= 0.4pt,line join=round] (171.30, 97.35) --
	(171.94, 97.35);

\path[draw=drawColor,line width= 0.4pt,line join=round] (174.47,169.76) --
	(175.11,169.76);

\path[draw=drawColor,line width= 0.4pt,line join=round] (174.79,169.76) --
	(174.79, 89.95);

\path[draw=drawColor,line width= 0.4pt,line join=round] (174.47, 89.95) --
	(175.11, 89.95);

\path[draw=drawColor,line width= 0.4pt,line join=round] (177.64,188.20) --
	(178.28,188.20);

\path[draw=drawColor,line width= 0.4pt,line join=round] (177.96,188.20) --
	(177.96, 97.08);

\path[draw=drawColor,line width= 0.4pt,line join=round] (177.64, 97.08) --
	(178.28, 97.08);

\path[draw=drawColor,line width= 0.4pt,line join=round] (180.81,176.18) --
	(181.45,176.18);

\path[draw=drawColor,line width= 0.4pt,line join=round] (181.13,176.18) --
	(181.13, 99.76);

\path[draw=drawColor,line width= 0.4pt,line join=round] (180.81, 99.76) --
	(181.45, 99.76);

\path[draw=drawColor,line width= 0.4pt,line join=round] (183.99,162.93) --
	(184.62,162.93);

\path[draw=drawColor,line width= 0.4pt,line join=round] (184.30,162.93) --
	(184.30, 83.00);

\path[draw=drawColor,line width= 0.4pt,line join=round] (183.99, 83.00) --
	(184.62, 83.00);

\path[draw=drawColor,line width= 0.4pt,line join=round] (187.16,171.97) --
	(187.79,171.97);

\path[draw=drawColor,line width= 0.4pt,line join=round] (187.47,171.97) --
	(187.47, 81.27);

\path[draw=drawColor,line width= 0.4pt,line join=round] (187.16, 81.27) --
	(187.79, 81.27);

\path[draw=drawColor,line width= 0.4pt,line join=round] (190.33,170.54) --
	(190.96,170.54);

\path[draw=drawColor,line width= 0.4pt,line join=round] (190.64,170.54) --
	(190.64, 85.28);

\path[draw=drawColor,line width= 0.4pt,line join=round] (190.33, 85.28) --
	(190.96, 85.28);

\path[draw=drawColor,line width= 0.4pt,line join=round] (193.50,166.64) --
	(194.13,166.64);

\path[draw=drawColor,line width= 0.4pt,line join=round] (193.81,166.64) --
	(193.81, 77.05);

\path[draw=drawColor,line width= 0.4pt,line join=round] (193.50, 77.05) --
	(194.13, 77.05);

\path[draw=drawColor,line width= 0.4pt,line join=round] (196.67,175.90) --
	(197.30,175.90);

\path[draw=drawColor,line width= 0.4pt,line join=round] (196.98,175.90) --
	(196.98, 87.05);

\path[draw=drawColor,line width= 0.4pt,line join=round] (196.67, 87.05) --
	(197.30, 87.05);

\path[draw=drawColor,line width= 0.4pt,line join=round] (199.84,171.90) --
	(200.47,171.90);

\path[draw=drawColor,line width= 0.4pt,line join=round] (200.16,171.90) --
	(200.16, 89.53);

\path[draw=drawColor,line width= 0.4pt,line join=round] (199.84, 89.53) --
	(200.47, 89.53);

\path[draw=drawColor,line width= 0.4pt,line join=round] (203.01,166.09) --
	(203.64,166.09);

\path[draw=drawColor,line width= 0.4pt,line join=round] (203.33,166.09) --
	(203.33, 90.80);

\path[draw=drawColor,line width= 0.4pt,line join=round] (203.01, 90.80) --
	(203.64, 90.80);

\path[draw=drawColor,line width= 0.4pt,line join=round] (206.18,159.07) --
	(206.81,159.07);

\path[draw=drawColor,line width= 0.4pt,line join=round] (206.50,159.07) --
	(206.50, 94.33);

\path[draw=drawColor,line width= 0.4pt,line join=round] (206.18, 94.33) --
	(206.81, 94.33);

\path[draw=drawColor,line width= 0.4pt,line join=round] (209.35,172.33) --
	(209.98,172.33);

\path[draw=drawColor,line width= 0.4pt,line join=round] (209.67,172.33) --
	(209.67,107.70);

\path[draw=drawColor,line width= 0.4pt,line join=round] (209.35,107.70) --
	(209.98,107.70);

\path[draw=drawColor,line width= 0.4pt,line join=round] (212.52,139.68) --
	(213.16,139.68);

\path[draw=drawColor,line width= 0.4pt,line join=round] (212.84,139.68) --
	(212.84,139.68);

\path[draw=drawColor,line width= 0.4pt,line join=round] (212.52,139.68) --
	(213.16,139.68);
\definecolor{drawColor}{RGB}{178,34,34}

\path[draw=drawColor,line width= 0.4pt,line join=round] (215.69,157.44) --
	(216.33,157.44);

\path[draw=drawColor,line width= 0.4pt,line join=round] (216.01,157.44) --
	(216.01,100.35);

\path[draw=drawColor,line width= 0.4pt,line join=round] (215.69,100.35) --
	(216.33,100.35);

\path[draw=drawColor,line width= 0.4pt,line join=round] (218.86,167.67) --
	(219.50,167.67);

\path[draw=drawColor,line width= 0.4pt,line join=round] (219.18,167.67) --
	(219.18, 99.99);

\path[draw=drawColor,line width= 0.4pt,line join=round] (218.86, 99.99) --
	(219.50, 99.99);

\path[draw=drawColor,line width= 0.4pt,line join=round] (222.03,168.96) --
	(222.67,168.96);

\path[draw=drawColor,line width= 0.4pt,line join=round] (222.35,168.96) --
	(222.35, 96.24);

\path[draw=drawColor,line width= 0.4pt,line join=round] (222.03, 96.24) --
	(222.67, 96.24);

\path[draw=drawColor,line width= 0.4pt,line join=round] (225.20,164.85) --
	(225.84,164.85);

\path[draw=drawColor,line width= 0.4pt,line join=round] (225.52,164.85) --
	(225.52, 96.40);

\path[draw=drawColor,line width= 0.4pt,line join=round] (225.20, 96.40) --
	(225.84, 96.40);

\path[draw=drawColor,line width= 0.4pt,line join=round] (228.37,169.38) --
	(229.01,169.38);

\path[draw=drawColor,line width= 0.4pt,line join=round] (228.69,169.38) --
	(228.69, 97.52);

\path[draw=drawColor,line width= 0.4pt,line join=round] (228.37, 97.52) --
	(229.01, 97.52);

\path[draw=drawColor,line width= 0.4pt,line join=round] (231.55,148.30) --
	(232.18,148.30);

\path[draw=drawColor,line width= 0.4pt,line join=round] (231.86,148.30) --
	(231.86, 80.12);

\path[draw=drawColor,line width= 0.4pt,line join=round] (231.55, 80.12) --
	(232.18, 80.12);

\path[draw=drawColor,line width= 0.4pt,line join=round] (234.72,152.61) --
	(235.35,152.61);

\path[draw=drawColor,line width= 0.4pt,line join=round] (235.03,152.61) --
	(235.03, 87.02);

\path[draw=drawColor,line width= 0.4pt,line join=round] (234.72, 87.02) --
	(235.35, 87.02);

\path[draw=drawColor,line width= 0.4pt,line join=round] (237.89,144.04) --
	(238.52,144.04);

\path[draw=drawColor,line width= 0.4pt,line join=round] (238.20,144.04) --
	(238.20, 74.48);

\path[draw=drawColor,line width= 0.4pt,line join=round] (237.89, 74.48) --
	(238.52, 74.48);

\path[draw=drawColor,line width= 0.4pt,line join=round] (241.06,149.83) --
	(241.69,149.83);

\path[draw=drawColor,line width= 0.4pt,line join=round] (241.37,149.83) --
	(241.37, 77.17);

\path[draw=drawColor,line width= 0.4pt,line join=round] (241.06, 77.17) --
	(241.69, 77.17);

\path[draw=drawColor,line width= 0.4pt,line join=round] (244.23,159.99) --
	(244.86,159.99);

\path[draw=drawColor,line width= 0.4pt,line join=round] (244.54,159.99) --
	(244.54, 90.52);

\path[draw=drawColor,line width= 0.4pt,line join=round] (244.23, 90.52) --
	(244.86, 90.52);

\path[draw=drawColor,line width= 0.4pt,line join=round] (247.40,158.63) --
	(248.03,158.63);

\path[draw=drawColor,line width= 0.4pt,line join=round] (247.72,158.63) --
	(247.72, 92.30);

\path[draw=drawColor,line width= 0.4pt,line join=round] (247.40, 92.30) --
	(248.03, 92.30);

\path[draw=drawColor,line width= 0.4pt,line join=round] (250.57,158.18) --
	(251.20,158.18);

\path[draw=drawColor,line width= 0.4pt,line join=round] (250.89,158.18) --
	(250.89, 94.23);

\path[draw=drawColor,line width= 0.4pt,line join=round] (250.57, 94.23) --
	(251.20, 94.23);

\path[draw=drawColor,line width= 0.4pt,line join=round] (253.74,157.43) --
	(254.37,157.43);

\path[draw=drawColor,line width= 0.4pt,line join=round] (254.06,157.43) --
	(254.06, 87.49);

\path[draw=drawColor,line width= 0.4pt,line join=round] (253.74, 87.49) --
	(254.37, 87.49);

\path[draw=drawColor,line width= 0.4pt,line join=round] (256.91,141.52) --
	(257.54,141.52);

\path[draw=drawColor,line width= 0.4pt,line join=round] (257.23,141.52) --
	(257.23, 74.10);

\path[draw=drawColor,line width= 0.4pt,line join=round] (256.91, 74.10) --
	(257.54, 74.10);

\path[draw=drawColor,line width= 0.4pt,line join=round] (260.08,136.78) --
	(260.72,136.78);

\path[draw=drawColor,line width= 0.4pt,line join=round] (260.40,136.78) --
	(260.40, 70.42);

\path[draw=drawColor,line width= 0.4pt,line join=round] (260.08, 70.42) --
	(260.72, 70.42);

\path[draw=drawColor,line width= 0.4pt,line join=round] (263.25,146.50) --
	(263.89,146.50);

\path[draw=drawColor,line width= 0.4pt,line join=round] (263.57,146.50) --
	(263.57, 66.72);

\path[draw=drawColor,line width= 0.4pt,line join=round] (263.25, 66.72) --
	(263.89, 66.72);

\path[draw=drawColor,line width= 0.4pt,line join=round] (266.42,158.96) --
	(267.06,158.96);

\path[draw=drawColor,line width= 0.4pt,line join=round] (266.74,158.96) --
	(266.74, 66.82);

\path[draw=drawColor,line width= 0.4pt,line join=round] (266.42, 66.82) --
	(267.06, 66.82);

\path[draw=drawColor,line width= 0.4pt,line join=round] (269.59,136.98) --
	(270.23,136.98);

\path[draw=drawColor,line width= 0.4pt,line join=round] (269.91,136.98) --
	(269.91, 59.85);

\path[draw=drawColor,line width= 0.4pt,line join=round] (269.59, 59.85) --
	(270.23, 59.85);

\path[draw=drawColor,line width= 0.4pt,line join=round] (272.76,140.29) --
	(273.40,140.29);

\path[draw=drawColor,line width= 0.4pt,line join=round] (273.08,140.29) --
	(273.08, 56.34);

\path[draw=drawColor,line width= 0.4pt,line join=round] (272.76, 56.34) --
	(273.40, 56.34);

\path[draw=drawColor,line width= 0.4pt,line join=round] (275.93,131.81) --
	(276.57,131.81);

\path[draw=drawColor,line width= 0.4pt,line join=round] (276.25,131.81) --
	(276.25, 55.01);

\path[draw=drawColor,line width= 0.4pt,line join=round] (275.93, 55.01) --
	(276.57, 55.01);

\path[draw=drawColor,line width= 0.4pt,line join=round] (279.11,130.42) --
	(279.74,130.42);

\path[draw=drawColor,line width= 0.4pt,line join=round] (279.42,130.42) --
	(279.42, 60.22);

\path[draw=drawColor,line width= 0.4pt,line join=round] (279.11, 60.22) --
	(279.74, 60.22);

\path[draw=drawColor,line width= 0.4pt,line join=round] (282.28,136.20) --
	(282.91,136.20);

\path[draw=drawColor,line width= 0.4pt,line join=round] (282.59,136.20) --
	(282.59, 59.91);

\path[draw=drawColor,line width= 0.4pt,line join=round] (282.28, 59.91) --
	(282.91, 59.91);

\path[draw=drawColor,line width= 0.4pt,line join=round] (285.45,131.39) --
	(286.08,131.39);

\path[draw=drawColor,line width= 0.4pt,line join=round] (285.76,131.39) --
	(285.76, 54.99);

\path[draw=drawColor,line width= 0.4pt,line join=round] (285.45, 54.99) --
	(286.08, 54.99);

\path[draw=drawColor,line width= 0.4pt,line join=round] (288.62,136.71) --
	(289.25,136.71);

\path[draw=drawColor,line width= 0.4pt,line join=round] (288.93,136.71) --
	(288.93, 65.50);

\path[draw=drawColor,line width= 0.4pt,line join=round] (288.62, 65.50) --
	(289.25, 65.50);

\path[draw=drawColor,line width= 0.4pt,line join=round] (291.79,145.11) --
	(292.42,145.11);

\path[draw=drawColor,line width= 0.4pt,line join=round] (292.10,145.11) --
	(292.10, 78.47);

\path[draw=drawColor,line width= 0.4pt,line join=round] (291.79, 78.47) --
	(292.42, 78.47);

\path[draw=drawColor,line width= 0.4pt,line join=round] (294.96,131.15) --
	(295.59,131.15);

\path[draw=drawColor,line width= 0.4pt,line join=round] (295.28,131.15) --
	(295.28, 66.40);

\path[draw=drawColor,line width= 0.4pt,line join=round] (294.96, 66.40) --
	(295.59, 66.40);

\path[draw=drawColor,line width= 0.4pt,line join=round] (298.13,135.46) --
	(298.76,135.46);

\path[draw=drawColor,line width= 0.4pt,line join=round] (298.45,135.46) --
	(298.45, 72.63);

\path[draw=drawColor,line width= 0.4pt,line join=round] (298.13, 72.63) --
	(298.76, 72.63);

\path[draw=drawColor,line width= 0.4pt,line join=round] (301.30,144.13) --
	(301.93,144.13);

\path[draw=drawColor,line width= 0.4pt,line join=round] (301.62,144.13) --
	(301.62, 74.18);

\path[draw=drawColor,line width= 0.4pt,line join=round] (301.30, 74.18) --
	(301.93, 74.18);

\path[draw=drawColor,line width= 0.4pt,line join=round] (304.47,139.54) --
	(305.10,139.54);

\path[draw=drawColor,line width= 0.4pt,line join=round] (304.79,139.54) --
	(304.79, 60.21);

\path[draw=drawColor,line width= 0.4pt,line join=round] (304.47, 60.21) --
	(305.10, 60.21);

\path[draw=drawColor,line width= 0.4pt,line join=round] (307.64,138.86) --
	(308.28,138.86);

\path[draw=drawColor,line width= 0.4pt,line join=round] (307.96,138.86) --
	(307.96, 45.84);

\path[draw=drawColor,line width= 0.4pt,line join=round] (307.64, 45.84) --
	(308.28, 45.84);
\definecolor{drawColor}{RGB}{0,0,0}

\path[draw=drawColor,line width= 0.4pt,dash pattern=on 4pt off 4pt ,line join=round] ( 31.29,139.68) -- (321.47,139.68);
\definecolor{drawColor}{gray}{0.20}

\path[draw=drawColor,line width= 0.4pt,line join=round,line cap=round] ( 31.29, 37.88) rectangle (321.47,213.06);
\end{scope}
\begin{scope}
\path[clip] (  0.00,  0.00) rectangle (325.21,216.81);
\definecolor{drawColor}{gray}{0.30}

\node[text=drawColor,anchor=base east,inner sep=0pt, outer sep=0pt, scale=  0.60] at ( 27.91, 40.48) {-0.075};

\node[text=drawColor,anchor=base east,inner sep=0pt, outer sep=0pt, scale=  0.60] at ( 27.91, 72.86) {-0.050};

\node[text=drawColor,anchor=base east,inner sep=0pt, outer sep=0pt, scale=  0.60] at ( 27.91,105.24) {-0.025};

\node[text=drawColor,anchor=base east,inner sep=0pt, outer sep=0pt, scale=  0.60] at ( 27.91,137.62) {0.000};

\node[text=drawColor,anchor=base east,inner sep=0pt, outer sep=0pt, scale=  0.60] at ( 27.91,170.00) {0.025};

\node[text=drawColor,anchor=base east,inner sep=0pt, outer sep=0pt, scale=  0.60] at ( 27.91,202.37) {0.050};
\end{scope}
\begin{scope}
\path[clip] (  0.00,  0.00) rectangle (325.21,216.81);
\definecolor{drawColor}{gray}{0.20}

\path[draw=drawColor,line width= 0.4pt,line join=round] ( 29.41, 42.55) --
	( 31.29, 42.55);

\path[draw=drawColor,line width= 0.4pt,line join=round] ( 29.41, 74.93) --
	( 31.29, 74.93);

\path[draw=drawColor,line width= 0.4pt,line join=round] ( 29.41,107.31) --
	( 31.29,107.31);

\path[draw=drawColor,line width= 0.4pt,line join=round] ( 29.41,139.68) --
	( 31.29,139.68);

\path[draw=drawColor,line width= 0.4pt,line join=round] ( 29.41,172.06) --
	( 31.29,172.06);

\path[draw=drawColor,line width= 0.4pt,line join=round] ( 29.41,204.44) --
	( 31.29,204.44);
\end{scope}
\begin{scope}
\path[clip] (  0.00,  0.00) rectangle (325.21,216.81);
\definecolor{drawColor}{gray}{0.20}

\path[draw=drawColor,line width= 0.4pt,line join=round] ( 44.79, 36.00) --
	( 44.79, 37.88);

\path[draw=drawColor,line width= 0.4pt,line join=round] (120.89, 36.00) --
	(120.89, 37.88);

\path[draw=drawColor,line width= 0.4pt,line join=round] (196.98, 36.00) --
	(196.98, 37.88);

\path[draw=drawColor,line width= 0.4pt,line join=round] (273.08, 36.00) --
	(273.08, 37.88);
\end{scope}
\begin{scope}
\path[clip] (  0.00,  0.00) rectangle (325.21,216.81);
\definecolor{drawColor}{gray}{0.30}

\node[text=drawColor,anchor=base,inner sep=0pt, outer sep=0pt, scale=  0.60] at ( 44.79, 30.37) {2003};

\node[text=drawColor,anchor=base,inner sep=0pt, outer sep=0pt, scale=  0.60] at (120.89, 30.37) {2005};

\node[text=drawColor,anchor=base,inner sep=0pt, outer sep=0pt, scale=  0.60] at (196.98, 30.37) {2007};

\node[text=drawColor,anchor=base,inner sep=0pt, outer sep=0pt, scale=  0.60] at (273.08, 30.37) {2009};
\end{scope}
\begin{scope}
\path[clip] (  0.00,  0.00) rectangle (325.21,216.81);
\definecolor{fillColor}{RGB}{255,255,255}

\path[fill=fillColor] (142.49,  3.75) rectangle (210.26, 13.20);
\end{scope}
\begin{scope}
\path[clip] (  0.00,  0.00) rectangle (325.21,216.81);
\definecolor{fillColor}{RGB}{255,255,255}

\path[fill=fillColor] (142.49,  3.75) rectangle (156.94, 18.20);
\definecolor{drawColor}{RGB}{70,130,180}
\definecolor{fillColor}{RGB}{70,130,180}

\path[draw=drawColor,line width= 0.3pt,line join=round,line cap=round,fill=fillColor] (149.72, 10.98) circle (  2.38);

\path[draw=drawColor,line width= 0.4pt,line join=round] (143.93, 10.98) -- (155.50, 10.98);
\end{scope}
\begin{scope}
\path[clip] (  0.00,  0.00) rectangle (325.21,216.81);
\definecolor{fillColor}{RGB}{255,255,255}

\path[fill=fillColor] (176.57,  3.75) rectangle (191.03, 18.20);
\definecolor{drawColor}{RGB}{178,34,34}
\definecolor{fillColor}{RGB}{178,34,34}

\path[draw=drawColor,line width= 0.3pt,line join=round,line cap=round,fill=fillColor] (183.80, 10.98) circle (  2.38);

\path[draw=drawColor,line width= 0.4pt,line join=round] (178.02, 10.98) -- (189.58, 10.98);
\end{scope}
\begin{scope}
\path[clip] (  0.00,  0.00) rectangle (325.21,216.81);
\definecolor{drawColor}{RGB}{0,0,0}

\node[text=drawColor,anchor=base west,inner sep=0pt, outer sep=0pt, scale=  0.80] at (160.69,  8.22) {Pre};
\end{scope}
\begin{scope}
\path[clip] (  0.00,  0.00) rectangle (325.21,216.81);
\definecolor{drawColor}{RGB}{0,0,0}

\node[text=drawColor,anchor=base west,inner sep=0pt, outer sep=0pt, scale=  0.80] at (194.78,  8.22) {Post};
\end{scope}
\end{tikzpicture}

%% file: a_big_table_finallargep_B.tex
\begin{table}[h]
\centering
\caption{Results for Monte Carlo Simulations with Different Group Sizes.}
\label{table:simulation_results largep}
\small
\begin{tabular*}{\textwidth}{@{\extracolsep{\fill}}ccccccccc}
\toprule\toprule
\multicolumn{3}{c}{Parameters} &
\multicolumn{2}{c}{MAD} &
\multicolumn{2}{c}{Coverage Probability} &
\multicolumn{2}{c}{CI Length} \\
\midrule
$K$ & $T$ & $n$ & SCD & CSDID & SCD & CSDID & SCD & CSDID \\
\midrule
\multicolumn{9}{l}{\textit{Scenario A: PTA and SMC hold}} \\[0.2em]
10 & 50  & 1500 & 0.088 & 0.147 & 1.000 & 0.999 & 1.552 & 1.218 \\
10 & 50  & 3000 & 0.062 & 0.106 & 1.000 & 0.998 & 1.030 & 0.868 \\
10 & 100 & 1500 & 0.082 & 0.148 & 1.000 & 1.000 & 1.475 & 1.284 \\
10 & 100 & 3000 & 0.059 & 0.109 & 0.999 & 0.997 & 0.971 & 0.917 \\
30 & 50  & 1500 & 0.139 & 0.235 & 1.000 & 0.997 & 2.497 & 1.801 \\
30 & 50  & 3000 & 0.094 & 0.172 & 1.000 & 0.997 & 1.865 & 1.305 \\
30 & 100 & 1500 & 0.131 & 0.226 & 1.000 & 0.999 & 2.316 & 1.893 \\
30 & 100 & 3000 & 0.093 & 0.160 & 1.000 & 0.999 & 1.718 & 1.366 \\
\midrule
\multicolumn{9}{l}{\textit{Scenario B: PTA fails but SMC holds}} \\[0.2em]
10 & 50  & 1500 & 0.089 & 1.623 & 0.998 & - & 0.985 & - \\
10 & 50  & 3000 & 0.065 & 1.625 & 0.997 & - & 0.672 & - \\
10 & 100 & 1500 & 0.085 & 1.596 & 0.996 & - & 0.963 & - \\
10 & 100 & 3000 & 0.062 & 1.592 & 0.995 & - & 0.657 & - \\
30 & 50  & 1500 & 0.142 & 1.981 & 0.999 & - & 1.605 & - \\
30 & 50  & 3000 & 0.098 & 1.964 & 0.996 & - & 1.110 & - \\
30 & 100 & 1500 & 0.133 & 1.890 & 0.997 & - & 1.434 & - \\
30 & 100 & 3000 & 0.093 & 1.886 & 0.994 & - & 0.964 & - \\
\midrule
\multicolumn{9}{l}{\textit{Scenario C: PTA holds but SMC fails}} \\[0.2em]
10 & 50  & 1500 & 1.039 & 0.147 & - & 0.999 & - & 1.218 \\
10 & 50  & 3000 & 1.034 & 0.106 & - & 0.998 & - & 0.869 \\
10 & 100 & 1500 & 1.009 & 0.148 & - & 1.000 & - & 1.284 \\
10 & 100 & 3000 & 1.009 & 0.109 & - & 0.997 & - & 0.917 \\
30 & 50  & 1500 & 1.282 & 0.235 & - & 0.997 & - & 1.802 \\
30 & 50  & 3000 & 1.294 & 0.172 & - & 0.997 & - & 1.306 \\
30 & 100 & 1500 & 1.162 & 0.226 & - & 0.999 & - & 1.894 \\
30 & 100 & 3000 & 1.149 & 0.160 & - & 0.999 & - & 1.366 \\
\bottomrule\bottomrule
\end{tabular*}
\medskip
\medskip
\vspace{0.01cm}
\parbox{6.4in}{\footnotesize
\textit{Notes:}  This table presents the simulation results for Scenarios A, B, and C described in subsection \ref{subsec: Monte Carlo}, when there is one large group in the donor pool. More specifically, we set $p = [0.8/K, \ldots, 0.8/K, 0.2]'$ for $K=10$, and $p = [0.925/K, \ldots, 0.925/K, 0.075]'$ for $K=30$. The table reports Mean Absolute Deviation (MAD), Coverage Probability, and Confidence Interval (CI) Length for our proposed method of Synthetic Control with Differencing (SCD) and the Difference-in-Differences estimator proposed by \cite{Callaway/SantAnna:JoE:21} (CSDID). Inference for SCD is conducted using Algorithm \ref{alg:confidence_intervals 1}. Inference results are only reported for cases where the target parameter is identified. The number of MC simulations is 1,000.}
\end{table}

%% file: a_big_table_finalheteroF_B.tex
\begin{table}[t]
\centering
\caption{Results for Monte Carlo Simulations with Different Time Factors.}
\label{table:simulation_results_heteroF}
\small
\begin{tabular*}{\textwidth}{@{\extracolsep{\fill}}ccccccccc}
\toprule\toprule
\multicolumn{3}{c}{Parameters} &
\multicolumn{2}{c}{MAD} &
\multicolumn{2}{c}{Coverage Probability} &
\multicolumn{2}{c}{CI Length} \\
\midrule
$K$ & $T$ & $n$ & SCD & CSDID & SCD & CSDID & SCD & CSDID \\
\midrule
\multicolumn{9}{l}{\textit{Scenario A: PTA and SMC hold}} \\[0.2em]
10 & 50  & 1500 & 0.079 & 0.147 & 0.999 & 0.998 & 1.472 & 1.166 \\
10 & 50  & 3000 & 0.059 & 0.104 & 1.000 & 0.999 & 0.964 & 0.829 \\
10 & 100 & 1500 & 0.079 & 0.147 & 1.000 & 0.999 & 1.403 & 1.232 \\
10 & 100 & 3000 & 0.054 & 0.100 & 1.000 & 0.998 & 0.905 & 0.872 \\
30 & 50  & 1500 & 0.128 & 0.229 & 1.000 & 0.999 & 2.449 & 1.766 \\
30 & 50  & 3000 & 0.094 & 0.163 & 1.000 & 1.000 & 1.830 & 1.275 \\
30 & 100 & 1500 & 0.130 & 0.228 & 0.999 & 0.998 & 2.277 & 1.867 \\
30 & 100 & 3000 & 0.090 & 0.157 & 1.000 & 0.999 & 1.721 & 1.351 \\
\midrule
\multicolumn{9}{l}{\textit{Scenario B: PTA fails but SMC holds}} \\[0.2em]
10 & 50  & 1500 & 0.091 & 1.914 & 1.000 & - & 1.003 & - \\
10 & 50  & 3000 & 0.068 & 1.908 & 0.996 & - & 0.683 & - \\
10 & 100 & 1500 & 0.095 & 1.861 & 0.995 & - & 0.969 & - \\
10 & 100 & 3000 & 0.062 & 1.858 & 0.997 & - & 0.654 & - \\
30 & 50  & 1500 & 0.138 & 2.082 & 1.000 & - & 1.593 & - \\
30 & 50  & 3000 & 0.103 & 2.074 & 0.997 & - & 1.069 & - \\
30 & 100 & 1500 & 0.139 & 2.005 & 0.994 & - & 1.383 & - \\
30 & 100 & 3000 & 0.096 & 1.982 & 0.991 & - & 0.926 & - \\
\midrule
\multicolumn{9}{l}{\textit{Scenario C: PTA holds but SMC fails}} \\[0.2em]
10 & 50  & 1500 & 1.173 & 0.147 & - & 0.998 & - & 1.165 \\
10 & 50  & 3000 & 1.176 & 0.104 & - & 0.999 & - & 0.830 \\
10 & 100 & 1500 & 1.179 & 0.147 & - & 0.999 & - & 1.233 \\
10 & 100 & 3000 & 1.186 & 0.100 & - & 0.999 & - & 0.872 \\
30 & 50  & 1500 & 1.358 & 0.229 & - & 0.999 & - & 1.770 \\
30 & 50  & 3000 & 1.399 & 0.163 & - & 1.000 & - & 1.275 \\
30 & 100 & 1500 & 1.332 & 0.228 & - & 0.998 & - & 1.867 \\
30 & 100 & 3000 & 1.354 & 0.157 & - & 0.999 & - & 1.352 \\
\bottomrule\bottomrule
\end{tabular*}
\medskip
\medskip
\vspace{0.01cm}
\parbox{6.4in}{\footnotesize
Notes: This table presents the simulation results for Scenarios A, B, and C described in subsection \ref{subsec: Monte Carlo}, modifying the distribution of time factors as follows: $F_t \sim N(\xi \sqrt{t}, 0.5^2 \cdot I_3)$, where $\xi=[0.01, 0.02, 0.04]'$. The table reports Mean Absolute Deviation (MAD), Coverage Probability, and Confidence Interval (CI) Length for our proposed method of Synthetic Control with Differencing (SCD) and the Difference-in-Differences estimator proposed by \cite{Callaway/SantAnna:JoE:21} (CSDID). Inference for SCD is conducted using Algorithm \ref{alg:confidence_intervals 1}. Inference results are only reported for cases where the target parameter is identified. The number of MC simulations is 1,000.}
\end{table}

%% file: a_big_table_finaltloading_B.tex
\begin{table}[t]
\centering
\caption{Results for Monte Carlo Simulations with Different Factor Loadings.}
\label{table:simulation_results_tloading}
\small
\begin{tabular*}{\textwidth}{@{\extracolsep{\fill}}ccccccccc}
\toprule\toprule
\multicolumn{3}{c}{Parameters} &
\multicolumn{2}{c}{MAD} &
\multicolumn{2}{c}{Coverage Probability} &
\multicolumn{2}{c}{CI Length} \\
\midrule
$K$ & $T$ & $n$ & SCD & CSDID & SCD & CSDID & SCD & CSDID \\
\midrule
\multicolumn{9}{l}{\textit{Scenario A: PTA and SMC hold}} \\[0.2em]
10 & 50  & 1500 & 0.079 & 0.158 & 0.999 & 1.000 & 1.810 & 1.274 \\
10 & 50  & 3000 & 0.059 & 0.116 & 1.000 & 0.998 & 1.224 & 0.907 \\
10 & 100 & 1500 & 0.079 & 0.160 & 1.000 & 1.000 & 1.731 & 1.342 \\
10 & 100 & 3000 & 0.054 & 0.112 & 1.000 & 1.000 & 1.165 & 0.955 \\
30 & 50  & 1500 & 0.128 & 0.260 & 1.000 & 0.998 & 2.965 & 1.954 \\
30 & 50  & 3000 & 0.094 & 0.178 & 1.000 & 1.000 & 2.245 & 1.414 \\
30 & 100 & 1500 & 0.130 & 0.249 & 1.000 & 0.998 & 2.795 & 2.069 \\
30 & 100 & 3000 & 0.090 & 0.175 & 1.000 & 0.999 & 2.130 & 1.488 \\
\midrule
\multicolumn{9}{l}{\textit{Scenario B: PTA fails but SMC holds}} \\[0.2em]
10 & 50  & 1500 & 0.094 & 1.918 & 0.999 & - & 1.202 & - \\
10 & 50  & 3000 & 0.070 & 1.906 & 1.000 & - & 0.833 & - \\
10 & 100 & 1500 & 0.093 & 1.863 & 0.998 & - & 1.165 & - \\
10 & 100 & 3000 & 0.065 & 1.858 & 0.997 & - & 0.804 & - \\
30 & 50  & 1500 & 0.141 & 2.084 & 1.000 & - & 1.952 & - \\
30 & 50  & 3000 & 0.101 & 2.074 & 0.999 & - & 1.350 & - \\
30 & 100 & 1500 & 0.138 & 2.009 & 0.996 & - & 1.728 & - \\
30 & 100 & 3000 & 0.098 & 1.980 & 0.999 & - & 1.192 & - \\
\midrule
\multicolumn{9}{l}{\textit{Scenario C: PTA holds but SMC fails}} \\[0.2em]
10 & 50  & 1500 & 1.159 & 0.158 & - & 1.000 & - & 1.276 \\
10 & 50  & 3000 & 1.155 & 0.116 & - & 0.998 & - & 0.909 \\
10 & 100 & 1500 & 1.134 & 0.160 & - & 1.000 & - & 1.342 \\
10 & 100 & 3000 & 1.128 & 0.112 & - & 1.000 & - & 0.955 \\
30 & 50  & 1500 & 1.339 & 0.260 & - & 0.998 & - & 1.956 \\
30 & 50  & 3000 & 1.341 & 0.178 & - & 1.000 & - & 1.414 \\
30 & 100 & 1500 & 1.230 & 0.249 & - & 0.998 & - & 2.070 \\
30 & 100 & 3000 & 1.213 & 0.175 & - & 0.999 & - & 1.490 \\
\bottomrule\bottomrule
\end{tabular*}
\medskip
\medskip
\vspace{0.01cm}
\parbox{6.4in}{\footnotesize
Notes: This table presents the simulation results for Scenarios A, B, and C described in subsection \ref{subsec: Monte Carlo}, modifying the distribution of individual factor loadings as follows: $\Lambda_i|G_i \sim t_5(m_{G_i}, I_3)$. The table reports Mean Absolute Deviation (MAD), Coverage Probability, and Confidence Interval (CI) Length for our proposed method of Synthetic Control with Differencing (SCD) and the Difference-in-Differences estimator proposed by \cite{Callaway/SantAnna:JoE:21} (CSDID). Inference for SCD is conducted using Algorithm \ref{alg:confidence_intervals 1}. Inference results are only reported for cases where the target parameter is identified. The number of MC simulations is 1,000.}
\end{table}

%% file: Causal_Inf_Groupwise_Matching_A46.bbl
\begin{thebibliography}{37}
\providecommand{\natexlab}[1]{#1}

\bibitem[{Abadie(2021)}]{Abadie:JEL2021}
\textsc{Abadie, A.} (2021). Using synthetic controls: Feasibility, data
  requirements, and methodological aspects. \textit{Journal of Economic
  Literature}, \textbf{59}, 391--425.

\bibitem[{Abadie \textit{et~al.}(2010)Abadie, Diamond and
  Hainmueller}]{Abadie/Diamond/Hainmueller:10:JASA}
\textsc{---}, \textsc{Diamond, A.} and \textsc{Hainmueller, J.} (2010).
  Synthetic control methods for comparative case studies: Estimating the effect
  of california's tobacco control program. \textit{Journal of the American
  Statistical Association}, \textbf{105}, 493--505.

\bibitem[{Arkhangelsky \textit{et~al.}(2021)Arkhangelsky, Athey, Hirshberg,
  Imbens and Wager}]{Arkhangelsky/Athey/Hirshberg/Imbens/Wager:AER:21}
\textsc{Arkhangelsky, D.}, \textsc{Athey, S.}, \textsc{Hirshberg, D.~A.},
  \textsc{Imbens, G.~W.} and \textsc{Wager, S.} (2021). Synthetic
  difference-in-differences. \textit{American Economic Review}, \textbf{111},
  4088--4118.

\bibitem[{Athey and Imbens(2006)}]{Athey/Imbens:Eca:06}
\textsc{Athey, S.} and \textsc{Imbens, G.~W.} (2006). Identification and
  inference in nonlinear difference-in-differences models.
  \textit{Econometrica}, \textbf{74}, 431--497.

\bibitem[{Bilinski and Hatfield(2019)}]{Bilinski/Hatfield:19:arXiv}
\textsc{Bilinski, A.} and \textsc{Hatfield, L.~A.} (2019). Nothing to see here?
  non-inferiority approaches to parallel trends and other model assumptions.
  \textit{arXiv:1805.03273v5 [stat.ME]}.

\bibitem[{Bohn \textit{et~al.}(2014)Bohn, Lofstrom and
  Raphael}]{Bohn/Lofstrom/Raphael:TRES:14}
\textsc{Bohn, S.}, \textsc{Lofstrom, M.} and \textsc{Raphael, S.} (2014). Did
  the 2007 legal arizona workers act reduce the state's unauthorized immigrant
  population? \textit{The Review of Economics and Statistics}, \textbf{96},
  258--269.

\bibitem[{Borusyak \textit{et~al.}(2024)Borusyak, Jaravel and
  Spiess}]{Borusyak/Jaravel/Spiess:TRES:24}
\textsc{Borusyak, K.}, \textsc{Jaravel, X.} and \textsc{Spiess, J.} (2024).
  Revisiting event-study designs: Robust and efficient estimation. \textit{The
  Review of Economic Studies}, \textbf{91}, 3253--3285.

\bibitem[{Callaway and Sant'Anna(2021)}]{Callaway/SantAnna:JoE:21}
\textsc{Callaway, B.} and \textsc{Sant'Anna, P.~H.} (2021).
  Difference-in-differences with multiple time periods. \textit{Journal of
  Econometrics}, \textbf{225}, 200--230.

\bibitem[{Canen and Song(2025)}]{Canen/Song:arXiv:25}
\textsc{Canen, N.} and \textsc{Song, K.} (2025). Simple inference on a
  simplex-valued weight. \textit{arXiv:2501.15692v1 [econ.EM]}.

\bibitem[{Chen(2023)}]{Chen:Eca:23}
\textsc{Chen, J.} (2023). Synthetic control as online linear regression.
  \textit{Econometrica}, \textbf{91}, 465--491.

\bibitem[{Chen \textit{et~al.}(2025)Chen, Sant'Anna and
  Xie}]{Chen/SantAnna/Xie:arXiv:25}
\textsc{Chen, X.}, \textsc{Sant'Anna, P.~H.} and \textsc{Xie, H.} (2025).
  Efficient difference-in-differences and event study estimators.
  \textit{arXiv:2506.17729v1 [econ.EM]}.

\bibitem[{Chung and Lu(2002)}]{Chung/Lu:AC:02}
\textsc{Chung, F.} and \textsc{Lu, L.} (2002). Connected components in random
  graphs with given expected degree sequences. \textit{Annals of Combinatorics
  2002 6:2}, \textbf{6}, 125--145.

\bibitem[{Clarke(1990)}]{clarke1990}
\textsc{Clarke, F.~H.} (1990). \textit{Optimization and Nonsmooth Analysis},
  \textit{Classics in Applied Mathematics}, vol.~5. New York.

\bibitem[{de~Chaisemartin and
  D'Haultfœuille(2020)}]{deChaisemartin/DHaultfoeulle:AER:20}
\textsc{de~Chaisemartin, C.} and \textsc{D'Haultfœuille, X.} (2020). Two-way
  fixed effects estimators with heterogeneous treatment effects.
  \textit{American Economic Review}, \textbf{110}, 2964--96.

\bibitem[{de~Chaisemartin and
  D'Haultfœuille(2023)}]{deChaisemartin/DHaultfoeulle:EJ:23}
\textsc{---} and \textsc{---} (2023). Two-way fixed effects and
  differences-in-differences with heterogeneous treatment effects: a survey.
  \textit{The Econometrics Journal}, \textbf{26}, C1--C30.

\bibitem[{Doudchenko and Imbens(2017)}]{Doudchenko/Imbens:ARXIV:17}
\textsc{Doudchenko, N.} and \textsc{Imbens, G.~W.} (2017). Balancing,
  regression, difference-in-differences and synthetic control methods: A
  synthesis. \textit{arXiv}, pp. 1--36.

\bibitem[{Ferman and Pinto(2021)}]{Ferman/Pinto:QE:21}
\textsc{Ferman, B.} and \textsc{Pinto, C.} (2021). Synthetic controls with
  imperfect pretreatment fit. \textit{Quantitative Economics}, \textbf{12},
  1197--1221.

\bibitem[{Freyaldenhoven \textit{et~al.}(2019)Freyaldenhoven, Hansen and
  Shapiro}]{Freyaldenhoven/Hansen/Shapiro:19:AER}
\textsc{Freyaldenhoven, S.}, \textsc{Hansen, C.} and \textsc{Shapiro, J.~M.}
  (2019). Pre-event trends in the panel event-study design. \textit{American
  Economic Review}, \textbf{109}, 3307--3338.

\bibitem[{Goodman-Bacon(2021)}]{Goodman-Bacon:21:JOE}
\textsc{Goodman-Bacon, A.} (2021). Difference-in-differences with variation in
  treatment timing. \textit{Journal of Econometrics}, \textbf{225}, 254--277.

\bibitem[{Gunsilius(2023)}]{Gunsilius:Eca:23}
\textsc{Gunsilius, F.~F.} (2023). Distributional synthetic controls.
  \textit{Econometrica}, \textbf{91}, 1105--1117.

\bibitem[{Heckman \textit{et~al.}(1997)Heckman, Ichimura and
  Todd}]{Heckman/Ichimura/Todd:TRES:97}
\textsc{Heckman, J.~J.}, \textsc{Ichimura, H.} and \textsc{Todd, P.~E.} (1997).
  Matching as an econometric evaluation estimator: Evidence from evaluating a
  job training programme. \textit{Review of Economic Studies}, \textbf{64},
  605--654.

\bibitem[{Kahn-Lang and Lang(2020)}]{Kahn-Lang/Lang:20:JBES}
\textsc{Kahn-Lang, A.} and \textsc{Lang, K.} (2020). The promise and pitfalls
  of differences-in-differences: Reflections on 16 and pregnant and other
  applications. \textit{Journal of Business \& Economic Statistics},
  \textbf{38}, 613--620.

\bibitem[{Kellogg \textit{et~al.}(2021)Kellogg, Mogstad, Pouliot and
  Torgovitsky}]{Kellogg/Mogstad/Pouliot/Torgovitsky:21:JASA}
\textsc{Kellogg, M.}, \textsc{Mogstad, M.}, \textsc{Pouliot, G.~A.} and
  \textsc{Torgovitsky, A.} (2021). Combining matching and synthetic control to
  tradeoff biases from extrapolation and interpolation. \textit{Journal of the
  American Statistical Association}, \textbf{116}, 1804--1816.

\bibitem[{Kuersteiner and Prucha(2020)}]{Kuersteiner/Prucha:Eca:20}
\textsc{Kuersteiner, G.~M.} and \textsc{Prucha, I.~R.} (2020). Dynamic spatial
  panel models: Networks, common shocks, and sequential exogeneity.
  \textit{Econometrica}, \textbf{88}, 2109--2146.

\bibitem[{Kwon and Roth(2024)}]{Kwon/Roth:24:AEAPP}
\textsc{Kwon, S.} and \textsc{Roth, J.} (2024). (empirical) bayes approaches to
  parallel trends. \textit{American Economic Association Papers and
  Proceedings}, \textbf{114}, 606--609.

\bibitem[{Lee and Wooldridge(2024)}]{Lee/Wooldridge:24:WP}
\textsc{Lee, S.~J.} and \textsc{Wooldridge, J.~M.} (2024). A simple
  transformation approach to difference-in-difference estimation for panel
  data. \textit{Working Paper}.

\bibitem[{Liu(2025)}]{Liu:25:arXiv}
\textsc{Liu, Y.} (2025). Synthetic parallel trends. \textit{arXiv:2511.05870v1
  [econ.EM]}.

\bibitem[{Manski and Pepper(2018)}]{Manski/Pepper:18:ReStat}
\textsc{Manski, C.~F.} and \textsc{Pepper, J.~V.} (2018). How do right-to-carry
  laws affect crime rates? coping with ambiguity using bounded-variation
  assumptions. \textit{Review of Economics and Statistics}, \textbf{100},
  232--244.

\bibitem[{Rambachan and Roth(2023)}]{Rambachan/Roth:23:ReStud}
\textsc{Rambachan, A.} and \textsc{Roth, J.} (2023). A more credible approach
  to parallel trends. \textit{Review of Economic Studies}, \textbf{90},
  2555--2591.

\bibitem[{Robbins \textit{et~al.}(2017)Robbins, Saunders and
  Kilmer}]{Robbins/Saunders/Kilmer:17:JASA}
\textsc{Robbins, M.~W.}, \textsc{Saunders, J.} and \textsc{Kilmer, B.} (2017).
  A framework for synthetic control methods with high-dimensional, micro-level
  data: Evaluating a neighborhood-specific crime intervention. \textit{Journal
  of the American Statistical Association}, \textbf{112}, 109--126.

\bibitem[{Roth \textit{et~al.}(2023)Roth, Sant'Anna, Bilinski and
  Poe}]{Roth/SantAnna/Bilinksi/Poe:JoE2023}
\textsc{Roth, J.}, \textsc{Sant'Anna, P.~H.}, \textsc{Bilinski, A.} and
  \textsc{Poe, J.} (2023). What's trending in difference-in-differences? a
  synthesis of the recent econometrics literature. \textit{Journal of
  Econometrics}, \textbf{235}.

\bibitem[{Shi \textit{et~al.}(2022)Shi, Sridhar, Misra and
  Blei}]{Shi/Sridhar/Misra/Blei:22:AISTATS}
\textsc{Shi, C.}, \textsc{Sridhar, D.}, \textsc{Misra, V.} and \textsc{Blei,
  D.~M.} (2022). On the assumptions of synthetic control methods.
  \textit{Proceedings of the 25th International Conference on Artificial
  Intelligence and Statistics (AISTATS)}.

\bibitem[{Smith and Todd(2005)}]{Smith/Todd:JoE:05}
\textsc{Smith, J.~A.} and \textsc{Todd, P.~E.} (2005). Does matching overcome
  lalonde's critique of nonexperimental estimators? \textit{Journal of
  Econometrics}, \textbf{125}, 305--353.

\bibitem[{Sun and Abraham(2021)}]{Sun/Abraham:JoE:21}
\textsc{Sun, L.} and \textsc{Abraham, S.} (2021). Estimating dynamic treatment
  effects in event studies with heterogeneous treatment effects.
  \textit{Journal of Econometrics}, \textbf{225}, 175--199.

\bibitem[{Sun \textit{et~al.}(2025)Sun, Xie and Zhang}]{Sun/Xie/Zhang:arXiv:25}
\textsc{Sun, Y.}, \textsc{Xie, H.} and \textsc{Zhang, Y.} (2025).
  Difference-in-differences meets synthetic control: Doubly robust
  identification and estimation. \textit{arXiv:2503.11375v1 [econ.EM]}.

\bibitem[{Wooldrige(2021)}]{Wooldridge:21:WP}
\textsc{Wooldrige, J.~M.} (2021). Two-way fixed effects, the two-way mundlak
  regression, and difference-in-differences estimators. \textit{SSRN 3906345}.

\bibitem[{Xu(2017)}]{Xu:PA:17}
\textsc{Xu, Y.} (2017). Generalized synthetic control method: Causal inference
  with interactive fixed effects models. \textit{Political Analysis},
  \textbf{25}, 57--76.

\end{thebibliography}
